\documentclass[sigconf, nonacm]{acmart}

\newcommand\vldbdoi{XX.XX/XXX.XX}
\newcommand\vldbpages{XXX-XXX}
\newcommand\vldbvolume{17}
\newcommand\vldbissue{1}
\newcommand\vldbyear{2023}
\newcommand\vldbauthors{\authors}
\newcommand\vldbtitle{\shorttitle} 
\newcommand\vldbavailabilityurl{https://github.com/DujianDing/AQUAPRO}
\newcommand\vldbpagestyle{plain} 

\let\oldnl\nl
\newcommand{\nonl}{\renewcommand{\nl}{\let\nl\oldnl}}

\let\OLDthebibliography\thebibliography
\renewcommand\thebibliography[1]{
  \OLDthebibliography{#1}
  \setlength{\parskip}{0pt}
  \setlength{\itemsep}{0pt plus 0.3ex}
}

\newcommand{\squishlist}{
 \begin{list}{$\bullet$}
  { \setlength{\itemsep}{0pt}
     \setlength{\parsep}{3pt}
     \setlength{\topsep}{3pt}
     \setlength{\partopsep}{0pt}
     \setlength{\leftmargin}{1.5em}
     \setlength{\labelwidth}{1em}
     \setlength{\labelsep}{0.5em} } }

\newcommand{\squishlisttwo}{
 \begin{list}{$\bullet$}
  { \setlength{\itemsep}{0pt}
    \setlength{\parsep}{0pt}
    \setlength{	opsep}{0pt}
    \setlength{\partopsep}{0pt}
    \setlength{\leftmargin}{2em}
    \setlength{\labelwidth}{1.5em}
    \setlength{\labelsep}{0.5em} } }

\newcommand{\squishend}{
  \end{list}  }

\newcommand{\pmt}{\mathbf{\Pi}} 
\newcommand{\oracle}{O}
\newcommand{\proxy}{P}
\newcommand{\ours}{\textsc{Aquapro}\xspace}
\newcommand{\pqa}{PQA\xspace}
\newcommand{\pqe}{PQE\xspace}
\newcommand{\csa}{CSC\xspace}
\newcommand{\cse}{CSE\xspace}
\newcommand{\supg}{SUPG\xspace}
\newcommand{\bltopk}{Top-K\xspace}
\newcommand{\scantest}{Scan2Test\xspace}
\newcommand{\sampletest}{Sample2Test\xspace}
\newcommand{\reals}{\mathbf{R}}
\newcommand{\aup}{\overline{a}}  
\newcommand{\alow}{\underline{a}} 
\newcommand{\tprlow}[1]{\underline{TPR}(#1)}
\newcommand{\maxitr}{max\_itr}
\newcommand{\topkproxy}[1]{D_{#1}}
\newcommand{\optdepth}{k^*}
\newcommand{\optdepthtop}{\overline{k}}
\newcommand{\sampleoptdepth}[1]{\overline{k}_{#1}}
\newcommand{\topphi}[1]{D_{#1}}
\newcommand{\topproxy}[1]{D_{#1}}
\newcommand{\precis}{M_p}
\newcommand{\recall}{M_r}
\newcommand{\measure}{M}
\newcommand{\measurecomp}{\overline{M}}

\newcommand{\closure}{C}
\newcommand{\mlow}{\underline{m}}
\newcommand{\mlowlow}{\underline{\underline{m}}}

\newcommand{\clow}{\underline{c}}
\newcommand{\mud}{\mu_{D}}
\newcommand{\nearneighbor}{NN^O}
\newcommand{\pos}{PoS}
\newcommand{\poi}{f}
\newcommand{\indexproxy}{I}

\newcommand{\eou}{EOC}
\newcommand{\epsr}{\epsilon_{r}}
\newcommand{\deltar}{\delta_{r}}

\newcommand{\sunion}{\mathcal{S}}
\newcommand{\kunion}[1][\mathcal{S}]{k_{#1}}
\newcommand{\query}{AOS-FRNN }

\newenvironment{proofsketch}{%
  \renewcommand{\proofname}{ Proof Sketch}\proof}{\endproof}

\newtheorem{lemma}{Lemma}
\newtheorem{assumption}{Assumption}
\newtheorem{fact}{Fact}
\newtheorem{claim}{Claim}
\newtheorem{problem}{Problem}
\newtheorem{example}{Example}

\usepackage{framed}
\usepackage{xspace} 
\usepackage[show]{boxnotes}
\usepackage[table]{colortbl}
\usepackage{soul}
\usepackage{ulem}
\usepackage{subcaption}
\usepackage{multirow}
\usepackage{amsmath}
\usepackage[linesnumbered,ruled,vlined]{algorithm2e}
\DeclareMathOperator*{\argmax}{argmax}

\SetKwInput{KwInput}{Input}                
\SetKwInput{KwOutput}{Output}              

\newcommand{\sihem}[1] {{\color{blue}S: #1}}

\newcommand{\dujian}[1]{\textcolor{black}{#1}}
\newcommand{\eat}[1]{}

\begin{document}
\title{On Efficient Approximate Queries over Machine Learning Models}

\author{Dujian Ding}
\affiliation{%
  \institution{University of British Columbia}
  \city{Vancouver}
  \state{Canada}
}
\email{dujian@cs.ubc.ca}

\author{Sihem Amer-Yahia}
\affiliation{%
  \institution{CNRS, Univ. Grenoble Alpes}
  \city{Grenoble}
  \state{France}
}
\email{sihem.amer-yahia@cnrs.fr}

\author{Laks VS Lakshmanan}
\affiliation{%
  \institution{University of British Columbia}
  \city{Vancouver}
  \state{Canada}
}
\email{laks@cs.ubc.ca}

\renewcommand{\shortauthors}{Dujian, Sihem, Laks}

\begin{abstract}
The question of answering queries over ML predictions has been gaining attention in the database community. This question is challenging because finding high quality answers  by invoking an \textit{oracle} such as a human expert or an expensive deep neural network model on every single item in the DB and then applying the query, can be prohibitive. We develop a novel unified   framework for approximate query answering by leveraging a \textit{proxy} to minimize the oracle usage of finding high quality answers for both Precision-Target (PT) and Recall-Target (RT) queries. 
Our framework uses a judicious combination of invoking the expensive oracle on data samples and applying the cheap proxy on the DB objects. 
It relies on two assumptions.  Under the \textsc{Proxy Quality} assumption, 
we develop two algorithms: \pqa that efficiently finds high quality answers with high probability and no oracle calls, and \pqe, a heuristic extension that achieves empirically good performance with a small number of oracle calls. 
Alternatively, under the \textsc{Core Set Closure} assumption, we develop two algorithms: \csa that efficiently returns high quality answers with high probability and minimal oracle usage, and \cse, which extends it to more general settings. Our extensive experiments on five real-world datasets on both query types, PT and RT, demonstrate that our algorithms outperform the state-of-the-art and achieve high result quality with provable statistical guarantees. 
\end{abstract}

\maketitle

\pagestyle{\vldbpagestyle}
\begingroup\small\noindent\raggedright\textbf{PVLDB Reference Format:}\\
\vldbauthors. \vldbtitle. PVLDB, \vldbvolume(\vldbissue): \vldbpages, \vldbyear.\\
\href{https://doi.org/\vldbdoi}{doi:\vldbdoi}
\endgroup
\begingroup
\renewcommand\thefootnote{}\footnote{\noindent
This work is licensed under the Creative Commons BY-NC-ND 4.0 International License. Visit \url{https://creativecommons.org/licenses/by-nc-nd/4.0/} to view a copy of this license. For any use beyond those covered by this license, obtain permission by emailing \href{mailto:info@vldb.org}{info@vldb.org}. Copyright is held by the owner/author(s). Publication rights licensed to the VLDB Endowment. \\
\raggedright Proceedings of the VLDB Endowment, Vol. \vldbvolume, No. \vldbissue\ %
ISSN 2150-8097. \\
\href{https://doi.org/\vldbdoi}{doi:\vldbdoi} \\
}\addtocounter{footnote}{-1}\endgroup

\ifdefempty{\vldbavailabilityurl}{}{
\vspace{.3cm}
\begingroup\small\noindent\raggedright\textbf{PVLDB Artifact Availability:}\\
The source code, data, and/or other artifacts have been made available at \url{\vldbavailabilityurl}.
\endgroup
}

\vspace*{-1ex} 
\section{Introduction} \label{sec:intro} 
Several applications at the frontier of databases (DBs) and machine learning (ML)  require support for query processing over ML models. In image retrieval for instance, querying a DB corresponds to finding images whose neural representations are close to an input query image, given a distance measure \cite{DBLP:journals/pvldb/KangEABZ17,DBLP:journals/pvldb/KangGBHZ20}. 
Similarly, in the medical domain, a typical query would look for patients  whose predicted clinical condition is similar to an input patient (see Figure \ref{fig:example}) using a Deep Neural Network (DNN)~\cite{DBLP:conf/cikm/RodriguesPGA20,DBLP:journals/isci/RodriguesGSBA21}. 
A straightforward way of answering these queries is to apply  the neural models exhaustively on all objects (e.g., images or patients) in the DB, and then return the objects that satisfy the query. \dujian{This is prohibitive   because applying DNNs and involving human expertise are both expensive.}  
\textit{In this paper, we propose an approximate query processing approach with provable guarantees that leverages a cheap proxy for the neural model} and uses a judicious combination of invoking the expensive oracle model on data samples and applying the cheap proxy on the DB. 
\begin{figure} [!tbp]
    \includegraphics[width=0.4\textwidth]{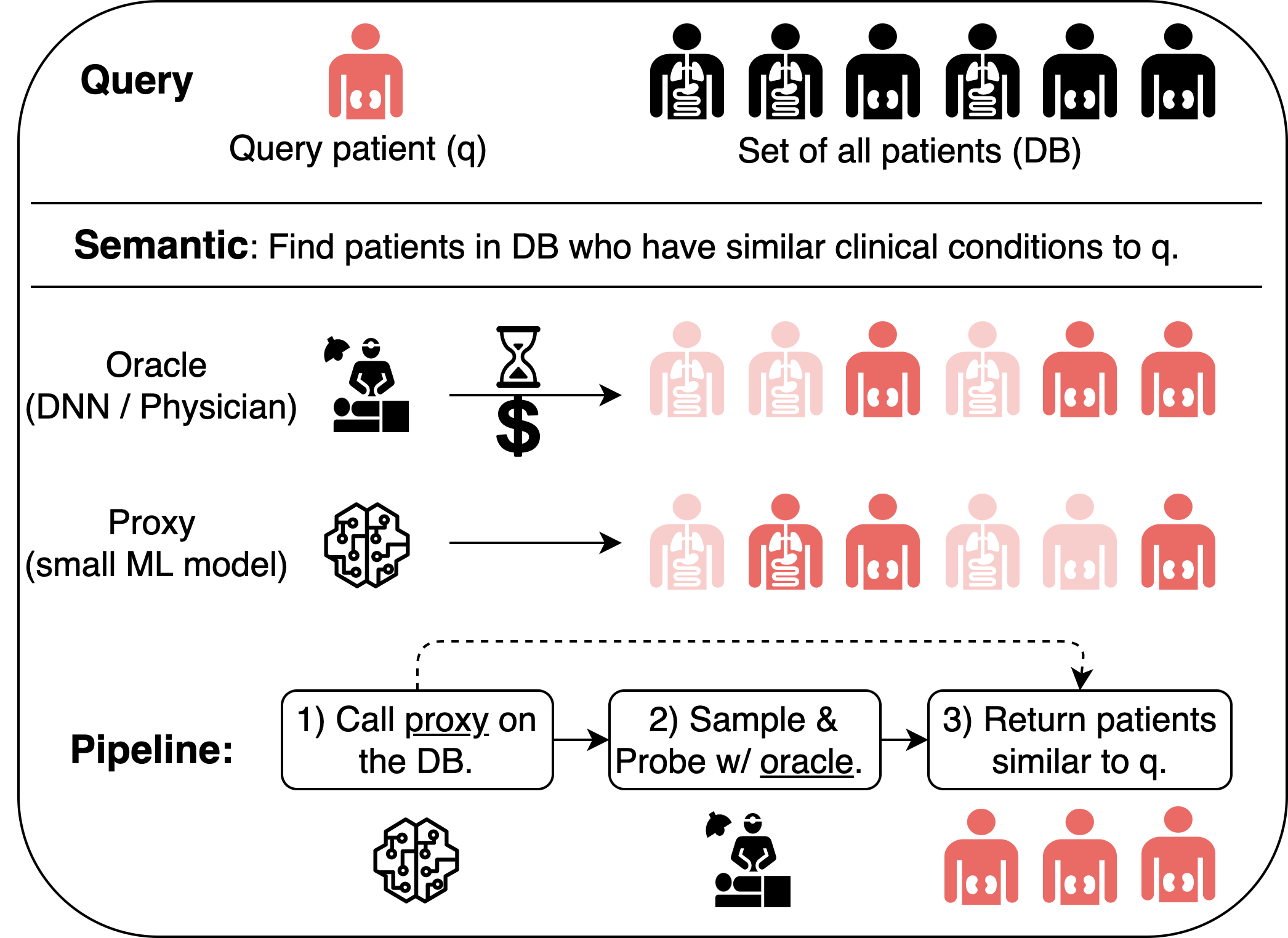}
      \vspace*{-2mm}
  \caption{Query over ML predictions in medical domain.}
        \vspace*{-3mm}
  \label{fig:example}
\end{figure}

The main focus of query processing over ML models has been to ensure efficiency without compromising accuracy~\cite{9094012}.  One line of work, query inference, provides native relational support for ML operators using containerized solutions such as Amazon Aurora,\footnote{\url{https://aws.amazon.com/fr/sagemaker/}} or in-application solutions such as Google's BigQuery ML\footnote{\url{https://cloud.google.com/bigquery-ml/docs}} and Microsoft's Raven~\cite{DBLP:conf/cidr/KaranasosIPSPPX20}. 
Another line develops adaptive predictions for NNs by pruning examples based on their classification  in early layers~\cite{DBLP:conf/icml/BolukbasiWDS17}. \textit{Our aim is to enable queries in a way that is agnostic to the underlying prediction model. Hence, we develop an in-application approach where queries can be invoked on any  ML prediction model. }

\dujian{Recent work ~\cite{DBLP:journals/pvldb/KangEABZ17,DBLP:conf/sigmod/LuCKC18,10.1145/3448016.3452786,DBLP:journals/pvldb/KangGBHZ20}
proposes to use cheap \textit{proxy} models that approximate ground truth oracle labels. 
Proxies are small neural models that either provide a confidence score~\cite{DBLP:journals/pvldb/KangGBHZ20,10.14778/3547305.3547310,DBLP:conf/sigmod/LuCKC18} or distribution~\cite{10.1145/3448016.3452786} for their predicted labels.
Probabilistic predicates (PP)~\cite{DBLP:conf/sigmod/LuCKC18} and CORE~\cite{10.14778/3547305.3547310} employ light-weight proxies to filter out unpromissing objects and empirically improve data reduction rates in query execution plans. Probabilistic Top-K~\cite{10.1145/3448016.3452786} trains proxy models to generate oracle label distribution and delivers approximate Top-K solutions.
}
Recently, in~\cite{DBLP:journals/pvldb/KangGBHZ20}, the authors study queries with a minimum precision target (PT) or recall target (RT), and a fixed user-specified budget on the number of oracle calls. However, 
\dujian{(i) setting an oracle budget is hard to get right. An underestimated budget may lead to trivial answers while overestimation causes unnecessary oracle usage}, and (ii) setting \textit{only} a minimum precision \textit{or} \textit{only} a minimum recall target, runs the risk of returning valid but uninformative answers: for RT, returning all  objects in the DB is valid but has very poor precision; 
for PT, returning the empty set is valid but it has zero recall and is hence not useful in practice. 

In this work, we consider oracle/proxy models with multi dimensional outputs. 
We propose a more useful problem of minimizing  oracle usage for finding answers that meet a precision or recall target with provable statistical guarantees while achieving a maximal complementary rate (CR). The CR for an RT (resp. PT) query is precision (resp. recall). 
More formally, given a PT (resp. RT) query, we seek 
answers that (1) satisfy a target precision (resp. recall) 
with a probability higher than a desired threshold and (2) incur a minimal number of oracle calls, and (3) achieve the maximal CR subject to the oracle usage incurred in (2). We aim to minimize oracle usage since the oracle is significantly more expensive than the proxy.

Our problem raises three challenges: (1) identify high quality answers with statistical guarantees, (2) design strategies that exactly or approximately  minimize oracle usage, and (3) achieve maximal CR subject to (2).
We develop a class of strategies that are agnostic to the prediction model and are applicable both to RT and PT queries. 
The key idea of our approach is to approximate an oracle with a cheaper \textit{proxy} model~\cite{DBLP:journals/pvldb/KangEABZ17,DBLP:journals/pvldb/KangGBHZ20}. 
In practice, the proxy could be a smaller and lower latency neural model. 
We consider a general pipeline for query answering which consists of three stages: (1) apply proxy on the DB, (2) sample \& probe with oracle, (3) compute and return answers (see Figure~\ref{fig:example}). 
We instantiate our pipeline
under two alternative assumptions. Under the \textsc{Proxy Quality} assumption, the proxy quality w.r.t. the oracle is quantified in a probabilistic manner which allows us to return high quality answers right after applying the proxy on the DB. We develop Algorithm \pqa which efficiently finds high probability valid answers of maximal expected CR with zero oracle calls. We additionally design Algorithm \pqe, a heuristic extension to \pqa, to calibrate the correlation between the oracle and the proxy by incurring some oracle calls.  If the proxy quality is hard to quantify, we have the \textsc{Core Set Closure} assumption under which we uniformly sample and probe a subset of objects to estimate valid answers to a given query. We introduce the notion of \textit{core set} to find the optimal sample size and number of samples so as to ensure a minimal expected oracle usage to identify a valid answer with high probability. We use the proxy to improve answer CR heuristically. This leads to Algorithm \csa, which efficiently returns high probability valid answers with a minimal expected number of oracle calls, and an empirically good CR. We also design Algorithm \cse, a generalization of \csa, which calibrates core sets with extra oracle calls and ensures high success probability. 

\dujian{We conduct experiments on five real-world datasets and  compare our algorithms to four baselines from recent work: (1) \supg~\cite{DBLP:journals/pvldb/KangGBHZ20}, (2) \bltopk~\cite{10.1145/3448016.3452786}, a probabilistic Top-K approach that uses oracle score distribution to deliver approximate Top-K answers, (3) \sampletest, a sample-based baseline adapted from the literature~\cite{DBLP:conf/sigmod/LuCKC18}, and (4) \scantest, a simple baseline that returns answers by applying oracle on all objects, which is also compared with in ~\cite{10.1145/3448016.3452786}.
}
Our experiments demonstrate that our algorithms find high quality answers with statistical guarantees \dujian{even  when baselines fail}. 
More specifically, we analyze \pqa and verify the optimality of its CR and success probability guarantee under the {\sc Proxy Quality} assumption.
\dujian{We compare \pqa with \bltopk on a synthetic dataset and demonstrate that \pqa returns high quality answers with zero oracle call while \bltopk incurs a huge oracle cost.}
We analyze \csa to demonstrate its minimal oracle usage and success probability guarantee under the {\sc Core Set Closure} assumption.
We compare \pqe, \cse, and \dujian{baselines} in terms of success probability and CR, under various oracle settings. We find that for RT queries, \cse has the best oracle efficiency and for PT queries, \pqe is the most oracle efficient approach. Finally, we study scalability and find that \cse is the most efficient approach outperforming \dujian{the strongest baseline} by up to $87.5\%$. 

In sum, we make the following contributions. 
\squishlist 
\item We propose the problem of answering PT and RT queries with minimal oracle usage and maximal CR while meeting precision or recall targets with high probability (\S~\ref{sec:problems}). 
\item We propose two assumptions ({\sc Proxy Quality} and {\sc Core Set Closure}), around which we develop four algorithms (\pqa, \pqe, \csa, and \cse) to solve the problem efficiently (\S~\ref{sec:algos}).
\item We run extensive experiments on five real-world datasets (\S~\ref{sec:experiment}) and show  that: (i) our approaches yield valid answers with high probability; (ii) our approaches significantly outperform the state of the art w.r.t.  CR and cost. 
\eat{We conduct extensive experiments on five real-world datasets to examine our solutions in terms of success probability, CR, and time efficiency.
\begin{itemize}
    \item[a)] For success probability, all our approaches deliver valid answers with a high probability, while \supg can fail the target by a margin of $10\%$.
    \item[b)] For CR, \cse outperforms \supg up to $33\%$ for RT queries and \pqe outperforms \supg up to $12\%$ for PT queries. 
    \item[c)] For time efficiency, \cse saves up to $87.5\%$ overall time in comparison to \supg for RT queries, and \pqe saves up to $38.66\%$ in comparison to \supg for PT queries.
\end{itemize}
}

\squishend 

\eat{In summary, we make the following contributions:
\squishlist
    \item Formalize the problem of answering approximate PT and RT queries as minimizing cost under the constraint of achieving a target CR with statistical guarantees (\S~\ref{sec:problems}).
    \item Develop algorithms that follow a class of strategies based on sample and probe. 
    \item Show that our algorithms address the combinatorial nature of our problem with theoretical guarantees under some assumptions (\S~\ref{sec:approach} and \S~\ref{sec:algos}). 
Due to limited space, we only provide proof sketches. The complete details can be found in \cite{fullVersion}. 
    \item Run extensive experiments on five real-world datasets examine our solutions in terms of CR, success probability guarantee, and response time, demonstrating that our solutions  outperform existing baselines and that our assumptions allow us to attain higher performance (\S~\ref{sec:experiment}).   
\squishend 
}
{\sl Complete details of proofs \dujian{as well as additional experiments} can be found in the full version \cite{aquaprofull}. } 
\vspace*{-1ex} 
\section{Problem Studied}
\label{sec:problems}

\subsection{Use Cases}
\label{sec:examples}

\vspace*{-1ex} 
\begin{example}[\textbf{Image Retrieval}] 
The problem is to find images similar to a query image  \cite{chen2021deep,deniziak2016content,zheng2018SIFT}. 
Metadata-based approaches use textual descriptions of images for quickly measuring  similarity, but their quality heavily relies on image annotations \cite{deniziak2016content}. Current approaches for content-based image retrieval are built upon deep neural networks which provide high accuracy but are computationally expensive. 
Our goal is to support efficient high quality approximate image retrieval queries \cite{yue2016deep}. 
\end{example}

\vspace*{-2ex} 
\begin{example}[\textbf{Preventive Medicine}]
One of the greatest obstacles of preventive medicine is the limited time a physician has \cite{love1994cancer,brull1999missed,franklin1959preventive,DBLP:conf/cikm/RodriguesPGA20}. 
Clinical Risk Prediction Models (CRPMs) are being developed to facilitate decision-making. CRPMs serve as prognosis prediction systems and predict the occurrence of specific diseases based on personalized medical records. 
Our goal is to extend queries to include CRPMs while offering statistical guarantees \cite{jose2021lig,choi2016doctor,li2020behrt}.
\end{example}

\vspace*{-2ex} 
\begin{example}[\textbf{Video Analytics}]
While DNNs have become effective for querying videos \cite{redmon2016look}, their inference cost becomes prohibitive as the model size increases. For example, to identify frames with a given class (e.g., ambulance) on a month-long traffic video, an advanced object detector such as YOLOv2 \cite{redmon2016yolo9000} needs about 190 GPU hours and \$380 for a cloud service \cite{hsieh2018focus}. A specialized model can achieve high efficiency, e.g., up to $340 \times$ faster than the full DNN, with \textit{sacrificed} accuracy \cite{DBLP:journals/pvldb/KangEABZ17}. Our goal is to efficiently generate high quality query answers  by balancing the use of expensive high-accuracy models and cheap low-accuracy proxies  \cite{hsieh2018focus,DBLP:journals/pvldb/KangEABZ17,DBLP:journals/pvldb/KangGBHZ20}. 
\end{example}


\vspace*{-1ex} 
\subsection{Query}
Our queries generalize \textit{Fixed-Radius Near Neighbor}  (FRNN) queries \cite{bentley1975frnn}. Given a dataset $D$, a query object $q$, a radius $r$, and a distance function $dist$,  an FRNN query asks for all \textit{near neighbors} of $q$ within radius $r$, i.e., 
    $NN(q,r) = \{x\in D \mid dist(x, q) \leq r\}.$  
In this paper, we are mainly interested in near neighbors of objects w.r.t. latent features, using a  distance function defined on these features. 
\dujian{In preventive medicine \cite{franklin1959preventive,DBLP:conf/cikm/RodriguesPGA20}, a latent feature may be the infection risk of a disease, which can be inferred from patient history, drug usage, and demographics.
}
Latent features  \dujian{can be discovered by human experts~\cite{DBLP:conf/sigmod/LuCKC18}} or powerful neural models, which we refer to as \textit{oracle}, denoted  $\oracle$. The near neighbors of query object $q$ w.r.t. $\oracle$ and radius $r$ are defined as 
    $NN^\oracle(q,r) = \{x\in D \mid dist(\oracle(x), \oracle(q)) \leq r\}$.  
We will use the notation $\nearneighbor$ when the query object $q$ and radius $r$ are clear from the context.
An object $x \in D$ is an \textit{oracle neighbor} of a query object w.r.t. radius $r$ if $x \in \nearneighbor$. 
Retrieving the exact $\nearneighbor$ requires calling the oracle on every single object in the DB, which is prohibitively expensive. Instead, we are interested in finding \textit{high quality} answers   \textit{with high probability} (w.h.p.).

For any subset $S \subseteq D$, we denote by $N_S = |S \cap \nearneighbor|$  the number of oracle neighbors in $S$. Define:
\begin{small}
\begin{equation}
\begin{split}
    \precis(S) &= N_S/|S| \hspace{0.5cm}
    \recall(S) = N_S/|\nearneighbor|
\end{split}
\end{equation}
\end{small}
A query specifies a user-given measure $\measure$, which can be either $\precis$ for precision or $\recall$ for recall, to measure answer quality, and a target $\gamma \in (0,1)$. In the former case, it is called a \textit{Precision-Target} (PT) query and in the latter, \textit{Recall-Target} (RT) query. We call $Ans \subseteq D$ a \textit{valid} answer iff $\measure(Ans) \geq \gamma$. For $M$, we use $\measurecomp$  to denote its \textit{complementary rate} (CR): when $\measure = \precis$ (resp., $\recall$), $\measurecomp$ stands for $\recall$ (resp., $\precis$). 
Given a query, we are interested in returning valid answers w.h.p. For any $S \subseteq D$, the \textit{probability of success} for $M(S)$ is $\pos(S, M, \gamma) := Pr[\measure(S) \geq \gamma]$. We  generalize FRNN queries to \textit{Approximate Oracle-Sensitive FRNN} (\query) \textit{queries}.   
\begin{definition} [\query Query]
Given a dataset $D$, a query object $q$, a radius $r$, a failure rate $\delta$, a main measure $\measure$ and corresponding target $\gamma \in (0,1)$, an \query query asks for a valid answer $Ans \subseteq D$ w.h.p.,  i.e., such that %
$\pos(Ans, M, \gamma) \geq 1 - \delta$. 
  \end{definition}
Effectively processing an \query query requires  determining: (1) \textit{How many oracle calls are required to find a  valid answer w.h.p.?} and (2) \textit{How good is the returned answer under a given CR?} The first question is important since oracle invocations are expensive and must be reduced. The second question is important because a technically valid answer could be uninformative. For instance, if a user specifies $\measure=\precis$ with a high target $\gamma$, the empty set is always a valid answer. Similarly, returning (nearly) the whole dataset is always a valid answer when $\measure=\recall$. 
\begin{problem} [\query Problem]
\label{problem_def}
Given a dataset $D$ and an \query query $Q$, find a valid answer $Ans \subseteq D$ to $Q$ w.h.p. such that (i) the number of oracle calls incurred is \textit{minimal} and (ii) the complementary rate $\measurecomp(Ans)$ achieved is \textit{maximal}  subject to 
(i). 
\end{problem}
The \query Problem is challenging given that we want to optimize two objectives (i.e., oracle usage and CR) under validity and success probability constraints. We will show that \textit{under certain conditions, we can efficiently return high probability valid answers with minimal or zero oracle calls and maximal expected CR}.

\vspace*{-1ex} 
\section{Approach Overview}
\label{sec:approach}
The key idea of our approach is to approximate an oracle with a cheaper \textit{proxy} model ~\cite{DBLP:journals/pvldb/KangEABZ17,DBLP:journals/pvldb/KangGBHZ20}. 
In practice, compared to an expensive oracle $O$, a proxy $\proxy$ could be a smaller and lower latency neural model. For brevity, when a query object $q$ is clear from the context, we use $dist^P(x)$ (resp. $dist^O(x)$) to denote $dist(\proxy(x), \oracle(q))$ (resp. $dist(\oracle(x), \oracle(q))$), for any $x \in D$. 

Given a dataset $D$, define an index function $\indexproxy: D \to \{i \mid 1\leq i \leq |D|\}$ that enumerates data objects in increasing order of their proxy distance, i.e., $\forall x_i, x_j \in D$, $\indexproxy(x_i) \leq \indexproxy(x_j)$ if $dist^P(x_i) \leq dist^P(x_j)$.
Denote by $\topproxy{k} = \{x \in D \mid 1 \leq \indexproxy(x) \leq k\}$  the $k$ nearest neighbors of the query object w.r.t. the proxy distance. 
$\topkproxy{0}$ is the empty set.
{Given a query object $q$, for $x\in D$, we say that $k$ is the \textit{proxy index} of $x$ if $k = \indexproxy(x)$. In this case, we call  $\topproxy{k}$ the \textit{proxy prefix} of $x$.}

To solve the \query Problem with guarantees, we examine two alternative assumptions:\\

\vspace*{-1ex} 
\noindent 
\textbf{Assumption 1 {\sc (Proxy Quality)}:} 
When the proxy quality w.r.t. the oracle can be quantified in a probabilistic manner, we aim to find high probability valid answers of maximal expected CRs with no oracle calls. We develop Algorithm \pqa to do that.
For $x \in D$, \pqa assumes \dujian{the conditional probability} of $dist^O(x)$, given $dist^P(x)$. \dujian{We can show that this assumption holds as long as data is i.i.d. (see \S~
\ref{subsubsec:proxy_assump}).}
This allows it to compute the success probability $\pos(S,\measure,\gamma)$ and expected CR   $\mathbb{E}[\measurecomp(S)]$ for any answer $S \subseteq D$.
 We prove that the optimal answer to any given query is $D_{k^*}$ for some  $0\leq k^* \leq |D|$. The optimal answer satisfies validity w.h.p. and has maximal expected CR. As $k^*$ is not known a priori, we explore the monotonicity of  $\pos(\topkproxy{k}, \measure, \gamma)$ and $\mathbb{E}[\measurecomp(\topkproxy{k})]$ w.r.t.  
$k$ in order to efficiently identify $\topkproxy{k^*}$. For RT queries, $\pos(\topkproxy{k}, \measure, \gamma)$ monotonically increases as $k$ increases. We use binary search to identify the smallest $k =\underline{k}$ such that  $\pos(\topkproxy{k}, \measure, \gamma) \geq 1-\delta$. Next, we find $\underline{k} \leq k = k^* \leq |D|$ which maximizes $\mathbb{E}[\measurecomp(\topkproxy{k})]$  and return $\topproxy{k^*}$ as the answer. 
For PT queries, $\mathbb{E}[\measurecomp(\topkproxy{k})]$ monotonically increases as $k$ increases. Thus, we incrementally compute $\pos(\topkproxy{k}, \measure, \gamma)$ for $0\leq k \leq |D|$ and set $k^*$ as the largest $k$ s.t.  $\pos(\topkproxy{k}, \measure, \gamma) \geq 1-\delta$. We return $\topproxy{k^*}$ as the answer. It is easy to see that $\topproxy{k^*}$ is the optimal answer and no oracle call is invoked in computing it. 

\begin{figure}[!tbp]
  \begin{subfigure}[b]{0.5\textwidth}
    \includegraphics[width=\textwidth]{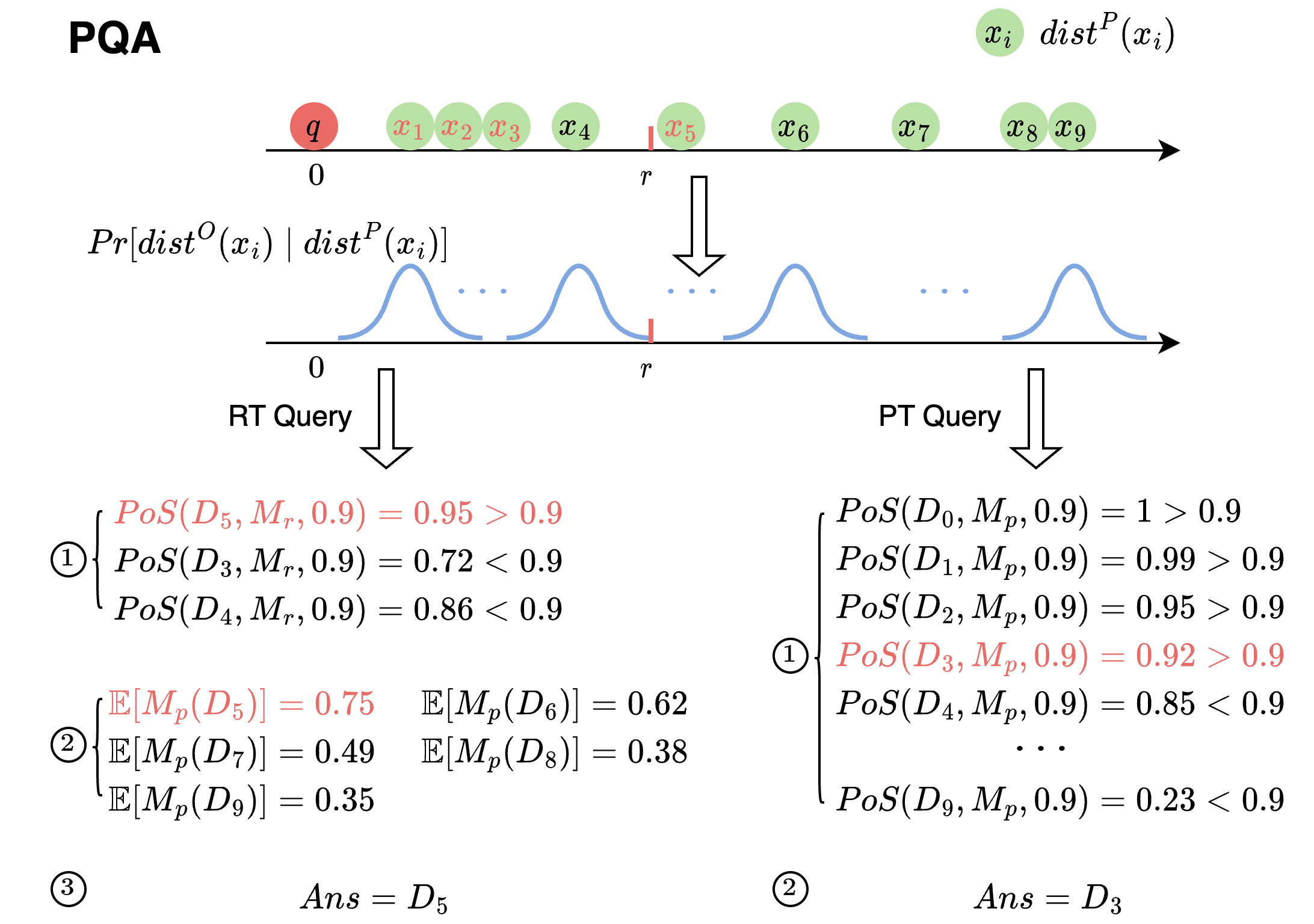}
  \end{subfigure}
    \vspace*{-7mm}
  \caption{\dujian{Example RT and PT query solved by \pqa with $\nearneighbor=\{x_1,x_2,x_3,x_5\}$, $\gamma=0.9$, and $\delta=0.1$. }}
  \label{fig:pqa_overview}

  \begin{subfigure}[b]{0.45\textwidth}
    \includegraphics[width=\textwidth]{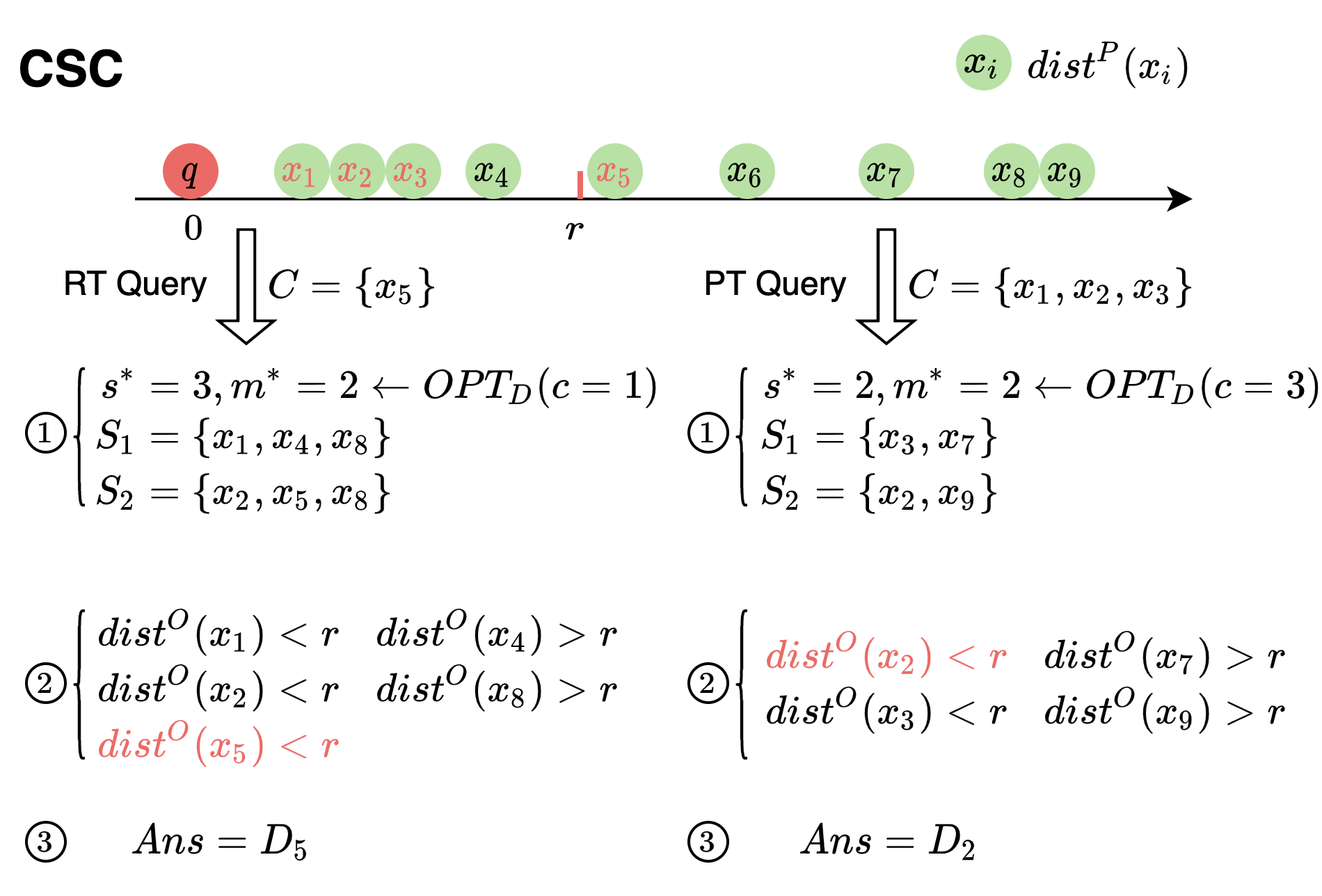}
  \end{subfigure}
    \vspace*{-4mm}
  \caption{Example RT and PT query solved by \csa with $\nearneighbor=\{x_1,x_2,x_3,x_5\}$, $\gamma=0.9$, and $\delta=0.1$. }
  \label{fig:csc_overview}
\end{figure}
\begin{example} 
A (synthetic) illustrative example is shown in Figure \ref{fig:pqa_overview}.\footnote{\dujian{All numbers are synthetic and are used to illustrate the operational workflow of our algorithms. We provide details of each computational step in \S~{\ref{sec:algos}}.}}
Consider a dataset $D=\{x_1,x_2,\cdots, x_9\}$. We show how to use \pqa to solve the example RT and PT queries with $\gamma=0.9$, $\delta=0.1$ and a ground truth $\nearneighbor=\{x_1,x_2,x_3,x_5\}$. We first compute proxy distance $dist^P(x_i)$ for each $x_i \in D$ and derive the oracle distance distribution $Pr[dist^O(x_i)|dist^P(x_i)]$ according to our assumption, which allows us to compute $\pos(S, \measure, \gamma)$ and $\mathbb{E}[\measurecomp(S)]$ for any $S \subseteq D$. We want to efficiently find the optimal answer $\topkproxy{k^*}$. In this example, $\indexproxy(x_i)=i$ and $\topproxy{k} = \{x_1,x_2,\cdots,x_k\}$.
For the RT query, we use binary search to find $\underline{k}=5$, i.e., the smallest $k$ satisfying $\pos(\topkproxy{k}, \recall, \gamma=0.9) \geq 1-\delta=0.9$. Next, we compute expected precision and return $\topkproxy{5}$ as the answer since $\mathbb{E}[\precis(\topkproxy{5})]=0.75 \geq \mathbb{E}[\precis(\topkproxy{k})]$ for any $\underline{k} \leq k \leq |D|$.
For the PT query, we compute $\pos(\topkproxy{k}, \precis, \gamma=0.9)$ for $0 \leq k \leq |D|$ and return $\topkproxy{3}$ as the answer since $k = 3$ is the largest $\topkproxy{k}$ satisfying $\pos(\topkproxy{k}, \precis, \gamma=0.9) \geq 1-\delta=0.9$.
\end{example} 

\vspace*{-1ex} 
\noindent 
\textbf{Assumption 2 {\sc (Core Set Closure)}:} 
When the proxy quality is hard to quantify, we aim to find  $\optdepth$ s.t. $\topkproxy{\optdepth}$ is the optimal answer. 
\eat{For a RT (resp. PT) query, $\optdepth$ should be the smallest (resp. largest) index s.t. $\topkproxy{\optdepth}$ is valid w.h.p. }  
Since computing $\optdepth$ exactly is expensive, we estimate it by sample and probe. Specifically, we uniformly draw $m$ samples of size $s$ each, from $D$ to estimate $k^*$ as  $\kunion$ where $\sunion$ is the union of samples, and return $\topkproxy{\kunion}$ as the answer.  
For RT (resp. PT) queries, we set $\kunion$ as the largest (resp. smallest) $\indexproxy(x)$, where $x$ is an oracle neighbor in $\sunion$. 
We seek the optimal values $s=s^*$ and $m=m^*$ which ensure  $\pos(\topkproxy{\kunion}, \measure, \gamma) \geq 1-\delta$ with a minimal expected number of oracle calls.
For that, we introduce the notion of $\textit{core set}$, denoted $C$. 
Given a query, the core set comprises all oracle neighbors $x \in \nearneighbor$ whose proxy prefix $\topproxy{\indexproxy(x)}$ is a valid answer. We say the core set is \textit{closed} w.r.t. a query $Q$ if one of the following holds: (i) $Q$ is a RT query and for every $x \in C$ any oracle neighbor whose proxy index is larger than that of $x$ is also in $C$; or (ii) $Q$ is a PT query and for every $x \in C$ any oracle neighbor whose proxy index is smaller than that of $x$ is also in $C$. 
\eat{ 
We say the core set $C$ is \textit{upward closed} (resp. \textit{downward closed}) provided for any $x \in C$, any oracle neighbor whose proxy index is larger (resp. smaller) than that of $x$ is also an element of $C$. 
} 
\eat{
Symmetrically,  we have \textit{downward closed}, a core requires closure over low proxy index. 
} 
Let $c$ denote the size of a given core set $C$.
We show that if the core set $C$ is closed w.r.t. a query and $c$ is known, $s^*$ and $m^*$ can be found by solving an optimization problem with $c$ as the input  (\S~\ref{subsec:alg_cs}). 
We develop Algorithm \csa to efficiently solve this problem and return $\topkproxy{\kunion}$.  
\csa returns valid answers w.h.p. with a minimal expected oracle usage and empirically good CR. 

\vspace*{-1ex}
\begin{example} 
The (synthetic) example  in Figure \ref{fig:csc_overview} illustrates  the idea behind Algorithm \csa. Consider the same setting as in Figure \ref{fig:pqa_overview}, where $D=\{x_1,x_2,\cdots, x_9\}$, and RT and PT queries with  $\gamma=0.9$, $\delta=0.1$, ground truth $\nearneighbor=\{x_1,x_2,x_3,x_5\}$, and $\topkproxy{k}=\{x_1, x_2, \cdots, x_k\}$. For the RT query, $x_5$ is the only oracle neighbor whose proxy prefix is a valid answer. Therefore, $C=\{x_5\}$ and $C$ is  closed. We can derive the optimal values $s^*=3$ and $m^*=2$, and uniformly draw samples $S_1$, $S_2$ from $D$. We then apply oracle  on each $x_i \in \mathcal{S} = S_1 \cup S_2$ and compute $dist^O(x_i)$ accordingly. At the end, we set $\kunion=5$ and return $\topkproxy{5}$ as the answer since $x_5$ has the largest proxy index among sampled oracle neighbors $x_1, x_2,x_5$. For the PT query, the core set is $C=\{x_1,x_2,x_3\}$, which is  closed. Similarly, we first derive the optimal values $s^*=2$ and $m^*=2$, and draw $S_1$, $S_2$ accordingly. Next, we apply the oracle on samples and compute the corresponding oracle distance. At the end, we set $\kunion=2$ and return $\topkproxy{2}$ as the answer since $x_2$ has the smallest proxy index among sampled oracle neighbors $x_2,x_3$.  
\end{example} 

\dujian{In case these assumptions do not hold,} we develop \pqe and \cse. \pqe is a heuristic extension to \pqa which calibrates oracle distance distribution by incurring some oracle calls. \cse complements \csa and ensures high success probability in general. The workflow and performance of all four approaches are summarized in Figure \ref{fig:use_case} and Table \ref{tab:guarantees}.

\begin{figure}
  \begin{subfigure}[b]{0.4\textwidth}
    \includegraphics[width=\textwidth]{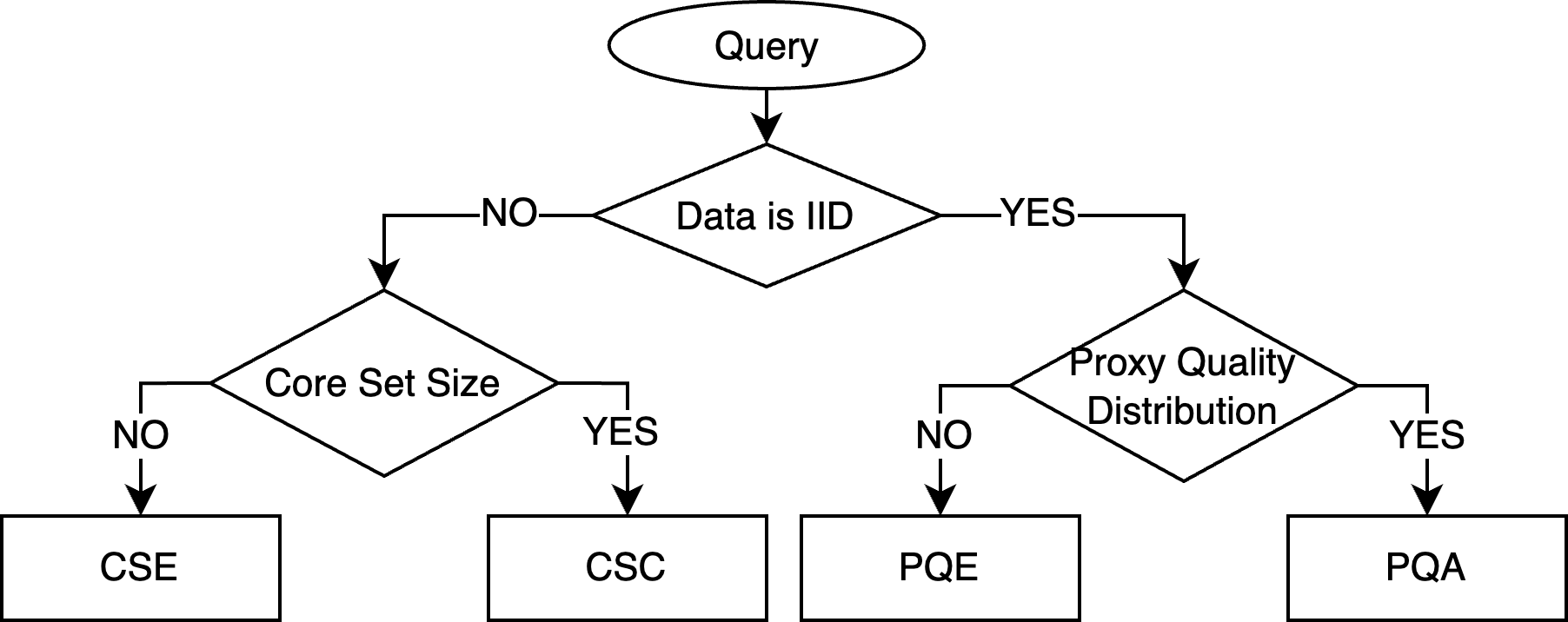}
  \end{subfigure}
    \vspace*{-3mm}
  \caption{\dujian{Workflow of different approaches.} }
  \label{fig:use_case}
  \vspace*{-3mm}
\end{figure}
\begin{table}
    \centering
\begin{small}
\begin{tabular}{|c|c|c|c|c|  }
\hline
 & Success Prob. & Oracle Usage & CR & Assumption \\
\hline
PQA & \cellcolor{blue!25}$\mathbf{\geq 1-\delta}$ & \cellcolor{blue!25}$\mathbf{0}$	& \cellcolor{blue!25}\textbf{MAX} 	&Yes \\
\hline
CSC & \cellcolor{blue!25}$\mathbf{\geq 1-\delta}$ & \cellcolor{blue!25}\textbf{MIN}	& good	& Yes \\
\hline
CSE & \cellcolor{blue!25}$\mathbf{\geq 1-\delta}$ & small	& good	& No\\
\hline
PQE & high	& small	& good	& No\\
\hline
\end{tabular}
\end{small}
\caption{Performance of different approaches for queries with specified $\delta$. Provable guarantees are highlighted. Empirical performance is described by ``high'', ``small'', and ``good''.}
\label{tab:guarantees}
\vspace{-5mm}
\end{table}

\begin{table}[ht]
  \caption{Notation Summary}
    \vspace*{-3mm}
  \label{tab:notation}
  \begin{small}
  \begin{tabular}{cccc}
    \toprule
    Symbol  & Description & Symbol  & Description \\
    \midrule
    $dist^O(x)$ & oracle distance & $r$ & radius threshold\\
    $dist^P(x)$ & proxy distance & $C$, $c$ & core set (size)\\
    $\nearneighbor$ & oracle neighbors in DB & $\delta$ & failure rate\\
   $N_S$  & \# oracle neighbors in $S$ & $I(x)$ & proxy index of $x$  \\
    $\measure$, $\measurecomp$  & main/comp. measure & $\precis$, $\recall$  & precision/recall \\
    $\topkproxy{k}$ & $k$ proxy-nearest neighbors & $\gamma$ & measure target \\
    \midrule
    $k^*$ &  \multicolumn{3}{c}{proxy index $I(x)$ of $x \in D$, s.t.  $\topkproxy{I(x)}$ is the optimal answer.   }\\
    $\kunion$ &  \multicolumn{3}{c}{max (resp. min) $I(x)$ of $x \in \sunion \cap \nearneighbor$ for RT (resp. PT) queries.  }\\
  \bottomrule
\end{tabular}
\end{small}
\end{table}

\eat{When $\measure = \precis$ (resp. $\measure=\recall$), our query is referred to as a \textit{precision target} (PT) (resp. \textit{recall target} (RT)) query. } 
We will use  $\measure$ and $\measurecomp$ when results hold for both PT and RT. We next describe our algorithms and provide a theoretical analysis. 

\section{Formal Analysis and Algorithms}
\label{sec:algos}

\subsection{Proxy Quality}
\label{subsec:alg_pq} 
In \S~\ref{subsubsec:proxy_assump}, we formally state the \textsc{Proxy Quality} assumption and show how the success probability of a set $S \subseteq D$ can be computed. Then, we develop Algorithm \pqa based on this assumption (\S~\ref{subsubsec:alg_pqa}) and analyze answer optimality (\S~\ref{subsubsec:pqa_opt}).
In \S~\ref{subsubsec:alg_pqe}, we develop Algorithm \pqe to extend \pqa to more general settings. 

\subsubsection{\textsc{Proxy Quality} Assumption}
\label{subsubsec:proxy_assump}

\dujian{In many real-world applications, data is collected in i.i.d. manner~\cite{10.1145/3448016.3452786,DBLP:conf/sigmod/LuCKC18}.
In our problem setting, the oracle and proxy are provided as input and serve as deterministic functions mapping a data object $x_i \in D$ to its prediction  $\oracle(x_i)$ or $\proxy(x_i)$. The difference between oracle and proxy distances to a given query object can be seen as i.i.d. random variables, whose i.i.d. property comes from the underlying data collection process.
}
Formally, the assumption states that given a query, the deviations between the proxy and oracle distances of different objects $x_i\in D$ are i.i.d. random variables: \dujian{for  $x_i \in D, \epsilon_i =dist^\oracle(x_i) - dist^\proxy(x_i)$}, where 
$\epsilon_i$ are i.i.d., $\epsilon_i \sim \mathcal{X}$.
\dujian{In Figure \ref{assmp_just:pqa} we report the distribution of $\epsilon_i$ on two real-world datasets, \textit{Mimic-III}~\cite{mimic_iii} and \textit{night-street}~\cite{jackson_dataset}. It is clear that, with high frequency, $\epsilon_i$ takes on values close to $0$, which indicates that the proxy is of good quality and can properly approximate the oracle predictions.}
\begin{figure}
  \begin{subfigure}[b]{0.23\textwidth}
    \includegraphics[width=\textwidth]{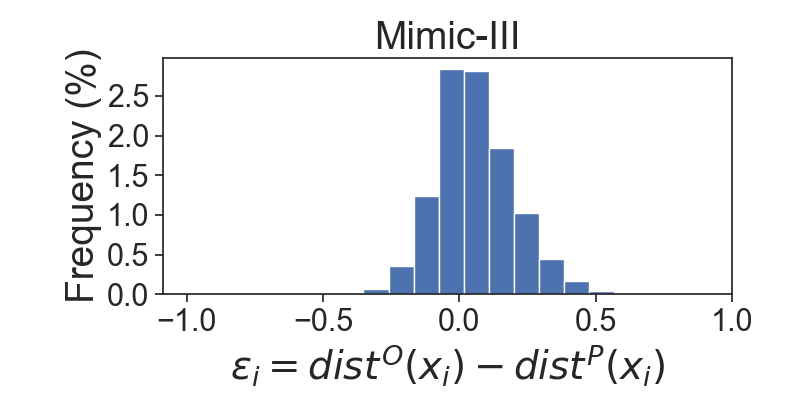}
  \end{subfigure}
 \hfill
    \begin{subfigure}[b]{0.23\textwidth}
    \includegraphics[width=\textwidth]{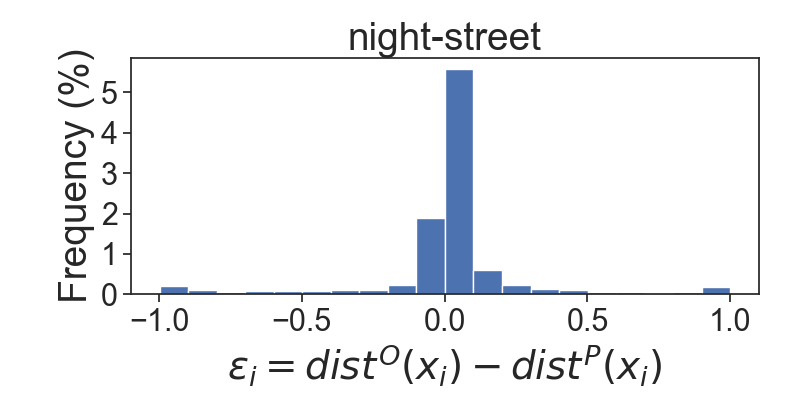}
  \end{subfigure}
  \vspace*{-4.5mm}
  \caption{\dujian{Distribution of $\epsilon_i =dist^\oracle(x_i) - dist^\proxy(x_i)$.}}
  \label{assmp_just:pqa}
\end{figure}

Under this assumption,
we can compute the oracle distance distribution for any $x_i \in D$,  after observing the proxy distance. 
The conditional probability of $x_i \in D$  being an oracle neighbor is:
\begin{small}
\begin{equation}
\label{fact:identity}
\begin{split}
Pr[x_i \in \nearneighbor \mid dist^\proxy(x_i)]&=Pr[dist^\oracle(x_i) \leq r \mid dist^\proxy(x_i)] \\
&=Pr[\epsilon_i \leq r-dist^\proxy(x_i)]
\end{split}
\end{equation}
\end{small}
The RHS of Eq. \ref{fact:identity} is the cdf of $\epsilon_i \sim \mathcal{X}$ evaluated at $r-dist^\proxy(x_i)$, i.e., $CDF_{{\footnotesize \mathcal{X}}}(r-dist^\proxy(x_i))$. For simplicity,  define $\phi(x_i) := CDF_{ \mathcal{X}}(r-dist^\proxy(x_i))$ and $\Phi(D):=\{\phi(x_i) \mid x_i \in D\}$. 
%
Notice, $\phi(x_i)$ provides the probability that $x_i$ is an oracle neighbor. 
\dujian{The overall success probability uses the \textit{possible world semantics}~\cite{moore1984possible}. The success probability of a subset $S \subseteq D$ equals the sum of  probabilities of all possible worlds in which $S$ has a high precision or recall w.r.t. the target $\gamma$.}
To  compute the success probability of $S$,
we seek the likelihood of any $S \subseteq D$ containing a certain number of oracle neighbors.

Recall that for any subset $S \subseteq D$, $N_S = |S \cap \nearneighbor|$  is the number of oracle neighbors in $S$.  
$N_S$ is thus a random variable equal to the sum of $|S|$ independent Bernoulli trials, 
each of which has a success probability $\phi(x_i)$, $x_i \in S$. 
Let $p_{N_S}(k) := Pr[N_S = k]$ be the probability mass function for any $S \subseteq D$ and $0\leq k \leq |S|$. We  next discuss how to compute it efficiently.


An important fact is that, given $S \subseteq D$, $x_i \notin S$, $p_{N_{S\cup \{x_i\}}}$ and $p_{N_{S}}$ satisfy the following recurrence relation:
\begin{small}
\begin{equation}
\label{eq:pndist_update}
\begin{split}
    p_{N_{S\cup \{x_i\}}}(k)=p_{N_{S}}(k-1)\cdot \phi(x_i) + p_{N_{S}}(k)\cdot (1-\phi(x_i))\\
\end{split}
\end{equation}
\end{small}
for $0\leq k \leq |S|+1$.
Eq. \ref{eq:pndist_update} says how to compute the probability mass function $p_{N_{S\cup \{x_i\}}}$ from $p_{N_{S}}$, for any $S \subseteq D$ and $x_i \notin S$. This recurrence relation directly suggests a way to compute $p_{N_{S}}$ for any $S$ with incremental updates, called  \textit{direct convolution} \cite{osti_1548776}.
We start from $S=\varnothing$ and apply Eq. \ref{eq:pndist_update} recursively to compute $p_{N_{S}}$ for any $S \subseteq D$. $p_{N_S}$ is implemented by an array (we abbreviate $\phi(x_i)$ as $\phi_i$). 
We initialize the array $p_{N_S}[0]=1$. 
We then iteratively update $p_{N_S}$ by including $x_i \in S$, $1\leq i \leq |S|$,  in \textit{any} order. The distribution updates are a direct implementation of Eq. \ref{eq:pndist_update}. 

We now discuss how to use $p_{N_S}$ to compute $\pos(S,\measure,\gamma)$, the success probability for $S$ to be a valid answer. 
We have the following fact, where $\overline{S} := D \setminus S$:
\begin{fact}
\label{fact:pos_pqa}
Given $S \subseteq D$ and $\gamma \in (0,1)$, 
\vspace*{-2ex} 
\begin{small}
\begin{equation}
\label{eq:precis_prob}
\begin{split}
    \pos(S, \precis, \gamma) &= Pr[\frac{N_S}{|S|} \geq \gamma]
    =\sum_{k=\lceil |S|\gamma \rceil}^{|S|}p_{N_S}(k)\\
\end{split}
\end{equation}
\end{small}
\vspace*{-1.5ex} 
\begin{small}
\begin{equation}
\label{eq:recall_prob}
\begin{split}
    \pos(S, \recall, \gamma) &= Pr[\frac{N_S}{|\nearneighbor|} \geq \gamma]
    =\sum_{j=0}^{|S|} p_{N_S}(j) \sum_{k=0}^{\lfloor j (1-\gamma)/\gamma \rfloor} p_{N_{\overline{S}}}(k)\\
\end{split}
\end{equation}
\end{small}
\end{fact}

For PT queries, the precision of $S \subseteq D$ increases  as $S$ contains more oracle neighbors.  
The probability of $S$ having a precision no less than $\gamma$ equals the probability of $S$ containing at least $|S|\gamma$ oracle neighbors, i.e., $Pr[N_S \geq |S|\gamma]$.
Eq. \ref{eq:precis_prob} gives this probability.

For RT queries, the recall of $S \subseteq D$ increases as $S$ contains more oracle neighbors \textit{relative to} the complement $\overline{S} = D \setminus S$: 
if $S$ contains $0 \leq j \leq |S|$ oracle neighbors, the conditional probability of $S$ having a recall no less than $\gamma$ equals  the probability that $\overline{S}$ contains no more than $j (1-\gamma)/\gamma$ oracle neighbors, i.e., $Pr[N_{\overline{S}} \leq j (1-\gamma)/\gamma]$. 
By the law of total probability \cite{Grimmett1986-GRIPAI}, the overall success probability $\pos(S, \recall, \gamma)$ equals  the summation of the product between the conditional success probability, $Pr[N_{\overline{S}} \leq j (1-\gamma)/\gamma]$, and the marginal probability, $Pr[N_S=j]$,  $0\leq j \leq |S|$. 
Using Eq. \ref{eq:recall_prob}, we use $p_{N_{S}}$ and $p_{N_{\overline{S}}}$ to compute this probability.


Fact \ref{fact:pos_pqa} 
gives a direct way to compute $\pos(S,M,\gamma)$ for any $S \subseteq D$ for a given query. 
We also leverage 
Eq. \ref{eq:precis_prob} and 
Eq. \ref{eq:recall_prob} iteratively.

  


\subsubsection{Algorithm \pqa}
\label{subsubsec:alg_pqa}
We develop  \pqa (Algorithm~\ref{alg:pqa})  which returns high probability valid answers with zero oracle calls and maximal expected CR, under the \textsc{Proxy Quality} assumption.
For PT queries, \pqa-PT computes the largest $k$ s.t.  $\pos(\topphi{k}, \precis, \gamma) \geq 1 - \delta$,  $0\leq k \leq |D|$, denoted  $k^*$. 
Notice that $\pos(S, \precis, \gamma)$ can be derived from $p_{N_S}$ in linear time, 
and $p_{N_S}$ can be computed from $p_{N_{S \setminus \{x_i\}}}$ in linear time, 
for any $x_i \in S \subseteq D$. \pqa-PT incrementally computes $p_{N_{D_k}}$ for each $0 \leq k \leq |D|$ and $\pos(\topphi{k}, \precis, \gamma)$ accordingly.
At the end, \pqa-PT returns $\topphi{k^*}$ where $k^* = \max\{0\leq k \leq |D| \mid \pos(\topphi{k}, M_p, \gamma) \geq 1 - \delta \}$. 
For RT queries, 
\pqa-RT uses binary search to identify the smallest $k =\underline{k}$ such that  $\pos(\topkproxy{k}, \measure, \gamma) \geq 1-\delta$. Next, \pqa-RT computes the expected CR of $\topphi{k}$ for each $\underline{k} \leq k \leq |D|$, and returns $\topphi{k^*}$ where $k^*=\argmax_{ \underline{k} \leq k \leq |D|} \mathbb{E}[\precis(\topphi{k})]$.

The algorithm is presented in Algorithm \ref{alg:pqa}. \pqa-PT is given in lines \ref{pqapt:s}-\ref{pqapt:e}. In lines \ref{pqapt:comp_pns_s}-\ref{pqapt:comp_pns_e}, we incrementally compute  $p_{N_{\topphi{k}}}$ for $0\leq k \leq |D|$. In lines \ref{pqapt:find_kmax_s}-\ref{pqapt:e}, we keep tracking the largest $k=k^*$,  $0\leq k\leq |D|$, such that  $\pos(\topphi{k}, \precis, \gamma) \geq 1 - \delta$, and return $\topphi{k^*}$ as the answer. The overall time complexity is $O(|D|^2)$.
\pqa-RT is given in lines \ref{pqart:s}-\ref{pqart:e}. In lines \ref{pqart:bs_s}-\ref{pqart:bs_e}, we use binary search to find the smallest $k=\underline{k}$ such that  $\pos(\topphi{k}, \recall, \gamma) \geq 1 - \delta$. Next, in lines \ref{pqart:maxexp_s}-\ref{pqart:e}, we compute $\mathbb{E}[\precis(\topphi{k})]$ for each $\underline{k} \leq k \leq |D|$ and return $\topphi{k^*}$ with the maximal expected CR. 
Binary search invokes $O(log(|D|))$ times $p_{N_S}$ computation, each of which is of $O(|D|^2)$. 
The overall time complexity is therefore $O(log(|D|)|D|^2)$. 
\begin{small}
\normalem
\begin{algorithm} 
\caption{PQA}
\label{alg:pqa}

\DontPrintSemicolon
  \SetKwFunction{pqapt}{PQA-PT}
  \SetKwProg{Fn}{Function}{:}{}
  \SetKwFunction{Sum}{Sum}

  \SetKwFunction{pqart}{PQA-RT}
  \SetKwFunction{fpnsupdate}{IncrementalUpdate}
\SetKwFunction{fprecis}{PoS-Mp}
  \SetKwFunction{frecall}{PoS-Mr}
  \SetKwFunction{fpns}{pNs}
  
\Fn{\pqapt{$\Phi_D=\Phi(D)$, $\gamma$, $\delta$} \label{pqapt:s}}
{$p_{N_S}[0] \gets 1$; $k^* \gets 0$\; \label{pqapt:comp_pns_s}
\For{$i \gets 1, 2, \cdots, |D|$}{
    $p_{N_S} \gets $ \fpnsupdate{$p_{N_S}$, $\Phi_D[i]$, $i$}\tcc*{Eq.3} \label{pqapt:comp_pns_e}
    \If{\fprecis{$p_{N_S}$, $\gamma$} $\geq 1-\delta$ \tcc*{Eq.4}\ \label{pqapt:find_kmax_s}}
    {$k^* \gets i$}}
\KwRet $\topphi{k^*}$
\label{pqapt:e}}
  
\Fn{\pqart{$\Phi_D=\Phi(D)$, $\gamma$, $\delta$} \label{pqart:s} }
{
$\underline{k} \gets 1$; $\overline{k} \gets |D|$\;\label{pqart:bs_s}
\While {$\underline{k} < \overline{k}$}{
$mid \gets \lfloor (\underline{k} + \overline{k}) / 2 \rfloor$\;
$p_{N_S} \gets$ \fpns{$\Phi(S)$};\ $p_{N_{\overline{S}}} \gets$ \fpns{$\Phi(\overline{S})$} \tcc*{Eq.3}
\eIf{\frecall{$p_{N_S}$, $p_{N_{\overline{S}}}$, $\gamma$} $< 1-\delta$ \tcc*{Eq.5}} 
{$\underline{k} \gets mid+1$}
{$\overline{k} \gets mid$}
} \label{pqart:bs_e}
$p_{N_S} \gets $ \fpns{$\Phi(\topphi{\underline{k}})$}\;
$E\measurecomp \gets \label{pqart:maxexp_s}$ \Sum{$\{p_{N_S}[i]\cdot i / {\underline{k}} \mid 1 \leq i \leq \underline{k}\}$} \tcc*{ $\mathbb{E}[\precis(S)]$}
\For{$i \gets \underline{k}+1, \underline{k}+2, \cdots, |D|$}{
$p_{N_S} \gets $ \fpnsupdate{$p_{N_S}$, $\Phi_D[i]$, $i$}\tcc*{Eq.3}
    $E\measurecomp' \gets $ \Sum{$\{p_{N_S}[j]\cdot j / {i} \mid 1 \leq j \leq i\}$}\;
    \If{$E\measurecomp' > E\measurecomp$}{
    $E\measurecomp \gets E\measurecomp'$;
    $k^* \gets i$\;}
}
\KwRet $\topphi{k^*}$
\label{pqart:e}}
\end{algorithm}
\end{small}
\ULforem

\subsubsection{\pqa Optimality}
\label{subsubsec:pqa_opt}
We first show that \eat{the optimal answer to a given query must be} there exists some $\topkproxy{k^*}$, \eat{for some $0 \leq k^* \leq |D|$} s.t. it is an optimal answer. We then explore the monotonicity relation between $\topkproxy{k}$ and $\topkproxy{k+1}$ w.r.t. success probability and expected CR, to efficiently find $\topkproxy{k^*}$. Finally, we show that  answers returned by \pqa are optimal for any query (proofs in the full version \cite{aquaprofull}). 

For $S \subseteq D$, we are interested in two operations to generate new answers: (i) \textit{replace} $x_i \in S$ with $x_j \notin S$, and (ii) \textit{append} $S$ with a new object $x \notin S$.
We first show that for any $S \subseteq D$, both success probability $\pos(S,M,\gamma)$ and expected CR $\mathbb{E}(\measurecomp(S))$  are monotone under the  replacement operation.

\begin{lemma} [Monotonicity of Replacement]
\label{lemma:monoto_replace}
Let $S \subseteq D$, $x_i \in S$, and $x_j \notin S$.  Denote $S'=S \cup \{x_j\} \setminus \{x_i\}$. For all $\gamma \in (0,1)$, if $\phi(x_i) \leq \phi(x_j)$, then
\begin{small}
\begin{equation}
\begin{split}
     \pos(S, M, \gamma) &\leq \pos(S', M, \gamma)\hspace{0.5cm} \mbox{and} \hspace{0.5cm} 
\mathbb{E}[\measurecomp(S)] \leq \mathbb{E}[\measurecomp(S')]\\
\end{split}
\end{equation}
\end{small}
\begin{proofsketch}
The proof leverages the notion of \textit{usual stochastic order} \cite{Nair2013}. 
\end{proofsketch}
\end{lemma}
 
\eat{Proof of Lemma \ref{lemma:monoto_replace} relies on the notion \textit{usual stochastic order} \cite{Nair2013}. Due to space limitation, we leave all the proofs to a complete version of this paper.} 

Lemma \ref{lemma:monoto_replace} says that, given $S \subseteq D$, if we replace $x_i \in S$ with $x_j \notin S$, where $x_j$ is more likely to be an oracle neighbor, both the success probability and the expected CR of $S$ will monotonically increase, for a given query.
Lemma \ref{lemma:monoto_replace} can be used to prune out a majority of unpromising solutions in the early stage of query processing. Specifically, given a query, we show that for any $0\leq k \leq |D|$, $\topphi{k}$ is optimal among all answers of size $k$. Recall $\topproxy{k}$ is the set of $k$ nearest neighbors of the query object w.r.t. the proxy distance. 
Formally,
\begin{theorem}
\label{theorem:topk}
For all $\gamma \in (0,1)$, $\forall 0\leq k \leq |D|$, $\topkproxy{k}$ has the highest success probability and expected CR among all $S \subseteq D$ with $|S|=k$.
\end{theorem}

\eat{Theorem \ref{theorem:topk} entails that, given a query, the optimal answer must be $\topphi{k^*}$ for some $1\leq k^* \leq |D|$. The next question is how to efficiently find the $\topphi{k^*}$.} 
Theorem \ref{theorem:topk} entails that, given a query, there exists  some  $0\leq k^* \leq |D|$  such that  $\topphi{k^*}$ is guaranteed to be an optimal answer. We study the append operation and have the following result. 
\begin{lemma} [Monotonicity of Append]
\label{lemma:monoto_append}
For all $\gamma \in (0,1)$ and  $0\leq k \leq |D|-1$,
\begin{small}
\begin{equation}
\begin{split}
\pos(\topphi{k}, M_r, \gamma) &\leq \pos(\topphi{k+1}, M_r, \gamma)\hspace{0.2cm}
\mathbb{E}[\recall(\topphi{k})] \leq \mathbb{E}[\recall(\topphi{k+1})]
\end{split}
\end{equation}
\end{small}
\end{lemma}

Lemma \ref{lemma:monoto_append} states that increasing $k$ leads to an increase both in the probability for $\topkproxy{k}$ to have a high recall and its expected recall.
In other words, the success probability of $\topphi{k}$ monotonically increases for RT queries, and the expected CR of $\topphi{k}$ monotonically increases for PT queries, as $k$ increases. 

By Theorem \ref{theorem:topk} and Lemma \ref{lemma:monoto_append}, for any given query, {\sl the answer  $\topkproxy{k^*}$ returned by Algorithm \pqa clearly has high success probability and the maximal expected CR, implying it   is an optimal answer. } 

\subsubsection{Algorithm \pqe}
\label{subsubsec:alg_pqe}


Recall that Algorithm \pqa requires $\Phi(D)$ as an input. In a general setting, when $\Phi(D)$ is unknown or \textsc{Proxy Quality} Assumption does not hold, we heuristically fit a normal distribution by sampling and probing on a limited number of objects, where the limit is controlled by a budget parameter. The resulting algorithm is \pqe (Algorithm~\ref{alg:pqe}).  
\eat{ 
\pqe heuristically estimates $\Phi(D)$, a required input of \pqa, by incurring more oracle calls and support queries in more general settings. } 

That is, in  \pqe, we employ $\epsilon_i \sim \mathcal{N}(\mu, \sigma)$ for all $x_i \in D$. 
\eat{$\mathcal{N}(\mu, \sigma)$ is a normal distribution of mean $\mu$ and standard deviation $\sigma$. } 
Specifically, we choose $\mu=0$, \eat{given the belief} which amounts to assuming that the proxy is an unbiased estimator of the  oracle. For $\sigma$, given a budget $b$, we sample and probe $b$ objects to estimate $\sigma$, denoted  $\hat{\sigma}$. 
We further introduce a hyper-parameter $\sigma_0$ to represent the deviation from the \textsc{Proxy Quality} assumption. In the ideal case where \textsc{Proxy Quality} holds, $\sigma_0=0$.  We heuristically choose $\sigma = \hat{\sigma} + \sigma_0$. We use $\mathcal{N}(\mu, \sigma)$ to compute $\Phi(D)$ and pass it to \pqa to find the  answers.
\begin{small}
\normalem
\begin{algorithm}
\caption{PQE}
\label{alg:pqe}

\DontPrintSemicolon
  \SetKwFunction{fpqe}{PQE}
  \SetKwProg{Fn}{Function}{:}{}
  \SetKwFunction{sample}{Sample}
  \SetKwFunction{std}{std}
  \SetKwFunction{cdf}{CDF$_{N(0, \sigma)}$}
  
\Fn{\fpqe{$D$, $\gamma$, $\delta$, $r$, $b$, $\sigma_0$}}
{
$S \gets$ \sample{$D$, $b$}\; \label{pqe:est_s}
$\sigma \gets \sigma_0 + $ \std{$\{ dist^O(x) - dist^P(x)\mid x \in S\}$}\; 
$\Phi_D \gets \{$ \cdf{$r - dist^P(x)$} $\mid x \in D\}$\; \label{pqe:est_e}
\eIf{RT query \label{pqe:pqa_s}}
{\KwRet \pqart{$\Phi_D$, $\gamma$, $\delta$}}
{\KwRet \pqapt{$\Phi_D$, $\gamma$, $\delta$}\label{pqe:pqa_e}}
}
\end{algorithm}
\ULforem
\end{small}

Algorithm \ref{alg:pqe} details the steps.
In lines \ref{pqe:est_s}-\ref{pqe:est_e}, we draw a sample $S \subseteq D$ of $|S|=b$ objects to estimate $\sigma$ and compute $\Phi(D)$. In lines \ref{pqe:pqa_s}-\ref{pqe:pqa_e}, we invoke Algorithm \pqa  with $\Phi(D)$ for PT or RT queries. 
The overall time complexity is dominated by Algorithm \pqa:
the additional time complexity on top of \pqa is $O(|D|)$. 


We now introduce \textsc{Core Set Closure} assumption, our second alternative assumption, and develop two algorithms $\csa$ and $\cse$. 

\vspace*{-1ex} 
\subsection{Core Set Closure}
\label{subsec:alg_cs}

In \S~\ref{subsubsec:csa_assumption}, we formally introduce the {\sc Core Set Closure} assumption and show how to find the optimal sample and probe strategy when core set size is known. We also analyze the case when core set size is unknown and show how to ensure high success probability. 
In \S~\ref{subsubsec:alg_csc}, we develop Algorithms \csa and \cse based on this. \dujian{In \S~\ref{subsubsec:progressive_qp}, we discuss how to support progressive query processing.}

\subsubsection{\textsc{Core Set Closure} Assumption}
\label{subsubsec:csa_assumption}

For a query, we define the \textit{core set} as the set of all oracle neighbors whose proxy prefix is a valid answer. 
We use $c$ to denote the size of a given core set $C$.
A core set $C$ is \textit{closed} w.r.t. an RT (resp. PT) query if for any $x \in C$, any oracle neighbor whose proxy index is larger (resp. smaller) than that of $x$ is also an element of $C$. {\sc Core Set Closure} assumption says that, \textit{for any given query, the core set is closed w.r.t. that query}. 

For RT queries, the core set is always closed, because as the proxy index of oracle neighbors increases, the recall of corresponding proxy prefix monotonically increases.
For PT queries, with a properly tuned proxy, the core set is likely to be closed in practice. 
\dujian{In Figure \ref{assmp_just:csc}, we report the average $Precision(\topkproxy{k})$ over $100$ random queries on two real datasets,  \textit{Mimic-III}~\cite{mimic_iii} and \textit{night-street}~\cite{jackson_dataset}. It is clear that the precision of proxy prefix $\topkproxy{k}$ monotonically decrease as $k$ increases on both datasets, which shows the core set closure property for PT queries.}
\begin{figure}
 \begin{subfigure}[b]{0.23\textwidth}
    \includegraphics[width=\textwidth]{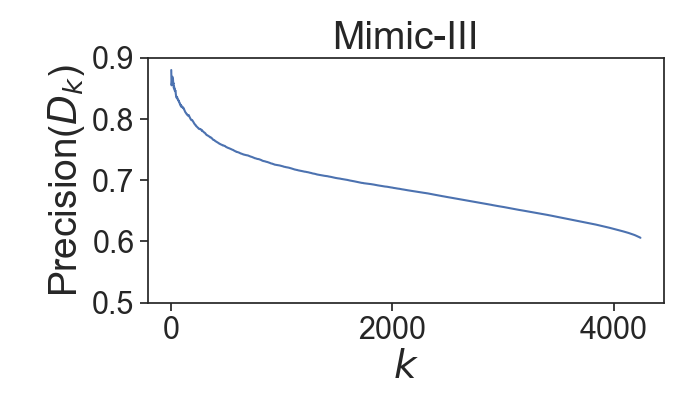}
 \end{subfigure}
\hfill
    \begin{subfigure}[b]{0.23\textwidth}
    \includegraphics[width=\textwidth]{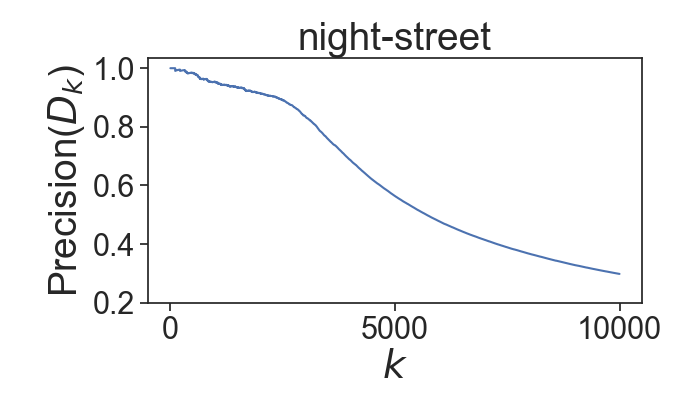}
 \end{subfigure}
 \vspace{-5mm}
 \caption{\dujian{Precision of proxy prefixes $\topkproxy{k}$.}}
 \label{assmp_just:csc}
\end{figure} 

We uniformly draw $m$ samples of size $s$ from $D$ to derive $\kunion$ where $\sunion$ is the union of samples, and return $\topkproxy{\kunion}$ as the answer.
Recall that $\kunion$ is the largest (resp. smallest)  $\indexproxy(x)$ for RT (resp. PT) queries, where $x$ is determined to be an oracle neighbor by probing $\sunion$ (see \S~\ref{sec:approach}, Assumption 2). 
If the core set $C$ is closed w.r.t. a given query, the success probability of $\topproxy{\kunion}$ is the likelihood of $\sunion$ intersecting with $C$, i.e., $\pos(\topproxy{\kunion}, \measure, \gamma) = Pr[\sunion \cap C \neq \varnothing]$.
Since samples are drawn uniformly, we have $Pr[\sunion \cap \closure \neq \varnothing] = 1 - (\binom{|D|-c}{s} / \binom{|D|}{s})^m = 1 - (\prod_{i=0}^{c-1} \frac{|D|-s-i}{|D|-i})^{m}$, where $s$ is the sample size and $m$ is the number of samples. We denote $\poi(|D|,s,m,c) := 1 - (\prod_{i=0}^{c-1} \frac{|D|-s-i}{|D|-i})^{m}$.


\paragraph{\textbf{When} $\mathbf{c}$ \textbf{is known}}

Given $s$ and $m$, the expected number of oracle calls made by the sample and probe strategy is  $\eou(s,m)=\mathbb{E}[|\sunion|]= |D|(1 - (1-\frac{s}{|D|})^{m})$. When $c$ is known, we can determine $s=s^*$ and $m=m^*$, which minimizes $\eou(s,m)$ while ensuring  $\poi(|D|,s,m,c) \geq 1 - \delta$, by solving the following equation:  
\begin{small}
\begin{equation}
\label{problem:find_opt_s_m}
\begin{split} 
     \min_{s,m} \quad & \eou(s,m)=|D|(1 - (1-\frac{s}{|D|})^{m})\\
    \textrm{s.t.} \quad &\poi(|D|,s,m,c) \geq 1 - \delta   \\
\end{split}
\end{equation}
\end{small}

By plugging in the expression for  $\poi(|D|,s,m,c)$, the constraint can be simplified to  $m \geq \lceil \frac{log(\delta)}{log(\prod_{i=0}^{c-1} \frac{|D|-s-i}{|D|-i})} \rceil$. By denoting the RHS as $\mlow(s)$, we can rewrite the constraint as $m \geq \mlow(s)$ for simplicity. 
Note, for a given $s$, $\eou(s,m)$ monotonically increases as $m$ increases.
For a fixed $s$, the optimal $m$ which ensures a high success probability (i.e., $\geq 1-\delta$) and minimizes $\eou(s,m)$ is clearly, $m=\mlow(s)$. 
As a special case, we have $m^*=\mlow(s^*)$.
Thus, a naive approach for finding  $s^*$ and $m^*$ is to compute $\eou(s,m)$ for each $1\leq s \leq |D|$ and $m=\mlow(s)$ and picking the best. 

Such exhaustive search for the exact value of $s^*$ and $m^*$, however, can be expensive in a large DB. Instead, we are interested in approximation solutions with good guarantees, which we develop next. Given a query, let $(s,m)$ denote the sample size and number of samples used by a strategy. Then $|D| - \eou(s,m)$ denotes the expected number of saved oracle calls compared with the exhaustive approach of probing every object in the DB. Define the \textit{savings ratio}  as $\xi(s,m) = \frac{|D|-\eou(s,m)}{|D|-\eou(s^*, m^*)}$. It denotes the fraction of oracle calls saved by  strategy $(s,m)$ compared to the optimal strategy $(s^*, m^*)$. 
\eat{Given a query, 
we define $\xi = \frac{|D|-\eou(s,m)}{|D|-\eou(s^*, m^*)}$ as the approximation ratio for any $s$ and $m$, which accounts for saved oracle usage in comparison to calling oracle on the whole DB.} 
A larger $\xi$ indicates a better approximation, and the optimal strategy $(s^*, m^*)$ yields $\xi(s^*, m^*) = 1$.  
 
Let us examine the special cases where either $s=1$ or $m=1$. 
For $s=1$, we let $m=\mlow(1)$, and 
\begin{small}
\begin{equation}
\label{eq:csa_s1}
     \xi_{s=1} := \xi(1, \mlow(1)) \geq \delta^{\frac{-1}{c}(\frac{1}{|D|} - \frac{|D|}{|D|-1)}} \cdot (1-1/|D|)
\end{equation}
\end{small}
For $m=1$, we set $s=s_1:=\lceil \frac{-log(\delta)}{\sum_{i=0}^{c-1} \frac{1}{|D|-i}} \rceil$ to ensure high success probability, and 
\begin{small}
\begin{equation}
\label{eq:csa_m1}
    \xi_{m=1} := \xi(s_1, 1)  \geq  \delta ^{\frac{-1}{|D|c}} \cdot (1 - 1/|D| + log(\delta)/c)
\end{equation}
\end{small}

In practice where e.g., $\delta=0.1$, $|D|=10,000$, and  $c=100$, we have both $\xi_{s=1}$ and $\xi_{m=1}$ being no less than $97.7\%$, that is, if we fix either $s=1$ or $m=1$ as above, the saved oracle usage is at least $97.7\%$ of what the optimal strategy $(s^*, m^*)$ achieves. Thus, either of them can be used as an approximation to the optimal strategy.  

\paragraph{\textbf{When} $c$ \textbf{is unknown} } 


We incur extra oracle calls and apply Hoeffding Bounds \cite{vershynin_2018} to ensure high success probability. 
\begin{proposition} [Hoeffding Bounds] 
Let $\{X_i\}_{i=1}^{n}$ be independent random variables, with  $X_i \in \{0,1\}$ and let $\mathbb{E}[X_i] = \mu$. Let $\hat{\mu} = \frac{1}{n} \sum_{i=1}^n X_i$. Then for $\forall \epsilon \geq 0$, we have the concentration bound
\begin{small}
\begin{equation}
    Pr[\hat{\mu}-\epsilon \leq \mu] \geq 1-exp(-2n\epsilon^2)
\end{equation}
\end{small}
\end{proposition}

For an RT query with target $\gamma$, we can derive a probabilistic lower bound for $c$ as follows. For an RT query, the core set $C$ consists of the top $(1-\gamma)\times 100$ \% oracle neighbors of largest proxy indices. That is, we can write $c=\lfloor |\nearneighbor|(1-\gamma)\rfloor+1$. When $s$ and $m$ are fixed, $\poi(|D|,s,m,c)$ monotonically increases as $c$ increases. Given $\deltar \in (0,1)$, let $\clow$ denote a probabilistic lower bound of $c$, i.e., $Pr[c \geq \clow] \geq 1-\deltar$. 
We can solve Eq. \ref{problem:find_opt_s_m}, either exactly or approximately as needed, subject to a more stringent constraint $\poi(|D|,s,m,\clow) \geq \frac{1-\delta}{1-\deltar}$ to find $s$ and $m$, which ensures an overall success probability no less than $1 - \delta$. 
We show how to derive such probabilistic lower bound $\clow$ using Hoeffding Bounds. 

Randomly draw $x_i \in D$. Define $X_i = 1$ iff $dist^O(x_i) \leq r$. We have $\mud := \mathbb{E}[X_i] = \frac{|\nearneighbor|}{|D|}$. Randomly draw $\{x_i\}_{i=1}^{n}$ with replacement. Denote $\hat{\mud} = \frac{1}{n} \sum_{i=1}^n X_i$.
For any $\epsr, \deltar \in (0,1)$, by Hoeffding Bounds, we have $Pr[\hat{\mud}-\epsr \leq \mud] \geq 1-\deltar $ if $n \geq \frac{log(\deltar)}{-2\epsr^2}$. 
For an RT query with target $\gamma$, 
since $c=\lfloor |\nearneighbor|(1-\gamma)\rfloor+1$,
and we have $Pr[c\geq \lfloor |D|(\hat{\mud}-\epsr)(1-\gamma)\rfloor+1]\geq 1-\deltar$, if $n \geq \frac{log(\deltar)}{-2\epsr^2}$. We denote the probabilistic lower bound $\clow := \lfloor |D|(\hat{\mud}-\epsr)(1-\gamma)\rfloor+1$.

For PT queries, such $\clow$ is hard to obtain. Given $\delta$ and $k$, we apply Hoeffding Bounds in a similar way to derive a probabilistic lower bound for the precision $\precis(\topkproxy{k})$, denoted as $\underline{\mu_{\topkproxy{k}}}$. That is, $Pr[\precis(\topkproxy{k}) \geq \underline{\mu_{\topkproxy{k}}}] \geq 1- \delta$. For a PT query with target $\gamma$,  $\topkproxy{k}$ is a high probability valid answer if $\underline{\mu_{\topkproxy{k}}} \geq \gamma$. We use heuristics to identify $\topkproxy{k}$ of high $\underline{\mu_{\topkproxy{k}}}$ and good CR (details in \S~\ref{subsubsec:alg_csc}). 

\subsubsection{Algorithms with \textsc{Core Set Closure}}
\label{subsubsec:alg_csc}
\paragraph{Algorithm \textbf{\csa}} Algorithm \csa returns high probability valid answers with a minimal expected number of oracle calls and empirically good CR, under \textsc{Core Set Closure} and taking $c$ as input.
\csa is presented in Algorithm \ref{alg:csa}. We compute $s^*$ and $m^*$ in line \ref{csa:getsm} either exactly or approximately, and draw samples in line \ref{csa:sample}. In lines \ref{csa:comp_k} to \ref{csa:e}, we compute $\kunion$ according to the query type and return $\topproxy{\kunion}$ as the answer. The exact solution to Eq.\ref{problem:find_opt_s_m} requires $O(c|D|)$ operations while approximate solutions take $O(c)$ operations. The  time complexity is dominated by accessing proxy prefixes, which requires sorting all objects w.r.t. proxy distance taking $O(|D|log(|D|))$. 
\normalem
\begin{algorithm}
\caption{\csa}
\label{alg:csa}

\DontPrintSemicolon
  \SetKwFunction{fcsa}{CSC}
  \SetKwProg{Fn}{Function}{:}{}
  \SetKwFunction{sample}{UniformSample}
  \SetKwFunction{fsm}{getsm}
  
\Fn{\fcsa{$D$, $c$, $\delta$}}{
$s^*, m^* \gets$ \fsm{$|D|$, $c$, $\delta$} \tcc*{Solve Eq.\ref{problem:find_opt_s_m}} \label{csa:getsm}
$\sunion \gets$ \sample{$D$, $s^*$, $m^*$}\; \label{csa:sample}
\eIf{RT query \label{csa:comp_k}}
{$\kunion \gets \max \{ I(x) \mid x \in \sunion \land dist^O(x) \leq r  \}$}
{$\kunion \gets \min \{ I(x) \mid x \in \sunion \land dist^O(x) \leq r  \}$}
\KwRet $\topproxy{\kunion}$
\label{csa:e}}
\end{algorithm}
\ULforem

\paragraph{Algorithm \textbf{\cse}}
Algorithm \cse incurs more oracle calls and returns high probability valid answers in general settings where $c$ is unknown or \textsc{Core Set Closure} assumption does not hold.

\cse is presented in Algorithm \ref{alg:cse}.
\cse-RT is given in lines \ref{csert:s} to \ref{csert:e}. Given $\epsr$ and $\deltar$, we sample and probe $n=\lceil \frac{log(\deltar)}{-2\epsr^2} \rceil$ objects to derive $\clow$. Then, we invoke Algorithm \csa to process the query subject to $\poi(|D|,s,m,\clow) \geq \frac{1-\delta}{1-\deltar}$. 
\cse-PT is given in lines \ref{csept:s} to \ref{csept:e}.
We use \csa to find good answer candidates, and apply Hoeffding Bounds to return high probability valid answers.
Specifically, given a query and budget $b'$, we sample and probe $b'$ objects to estimate $c$, then invoke \csa with the estimation to compute $\topkproxy{k_1}$ (lines \ref{csept:sample} to \ref{csept:est_topk}). 
We also use the same sample to estimate the largest $k=k_2$ such that $\topkproxy{k}$ has a sampled precision no less than $\gamma$ (line \ref{csept:khat}).
To improve CR (i.e., recall for PT queries), we set  $\hat{k}=\max\{k_1, k_2\}$ and consider $\topkproxy{\hat{k}}$ as the answer candidate (line \ref{csept:kstar}).
In lines \ref{csept:hb} to \ref{csept:e}, given $\epsilon_p$, we draw samples and estimate a probabilistic lower bound for $\precis(\topkproxy{\hat{k}})$ by applying Hoeffding Bounds.
For a PT query with target $\gamma$, we return $\topproxy{\hat{k}}$ if the probabilistic lower bound is no less than $\gamma$. O/w, we return all oracle neighbors identified from samples. The overall time complexity is dominated by \csa and is also $O(|D|log(|D|))$.

\vspace*{-1ex} 
\dujian{
\subsubsection{Progressive Query Processing}
\label{subsubsec:progressive_qp}
We observe that though we minimize the oracle usage, for some challenging queries  the \textit{bare minimum} of oracle calls can still be too high. We propose \textit{progressive query processing} for that. Recall that, our \csa and \cse approaches draw $m$ samples of size $s$ to compute $\kunion$ and return $\topkproxy{\kunion}$ as the answer. Instead of computing $\kunion$ after seeing all the samples, we can derive $\kunion'$ after seeing each sample and use $\kunion'$ to select answers with adaptive success probability bounds that are progressively better and better. We can keep refining $\kunion'$ when we see more samples, eventually approaching $\kunion$, but the user can terminate the evaluation at any time based on the oracle cost incurred thus far.
}

\normalem
\begin{algorithm}
\caption{CSE}
\label{alg:cse}
\DontPrintSemicolon
  \SetKwFunction{fcsert}{CSE-RT}
\SetKwFunction{fcsept}{CSE-PT}
  \SetKwProg{Fn}{Function}{:}{}
  \SetKwFunction{sample}{UniformSample}
  \SetKwFunction{fhb}{HoeffdingEst}

\Fn{\fcsert{$D$, $\delta$} \label{csert:s}}
{
$\hat{\mud} \gets$ \fhb{$D$, $\deltar$, $\epsilon_r$} \;
$\clow \gets \lfloor |D|(\hat{\mud}-\epsilon_r)(1-\gamma)\rfloor+1$\;

\KwRet \fcsa{$D$, $\clow$, $\frac{1-\delta}{1-\deltar}$}
\label{csert:e}}
  
\Fn{\fcsept{$D$, $\delta$}\label{csept:s}}
{
$S \gets$ \sample{$D$, $b'$}\; \label{csept:sample}
$\hat{c} \gets \frac{|D|}{|S|} \cdot \text{the size of core set w.r.t. } S$\;
$\topproxy{k_1} \gets$ \fcsa{$D$, $\hat{c}$, $1-\delta$}\;\label{csept:est_topk}
$k_2 \gets \max\{ I(x) \mid x \in S \land M_p(\topkproxy{I(x)} \cap S) \geq \gamma \}$\; \label{csept:khat}
$\hat{k} \gets \max\{ k_1, k_2\}$\;\label{csept:kstar}
$\underline{\mu_{\topproxy{\hat{k}}}} \gets$ \fhb{$\topproxy{\hat{k}}$, $\delta$, $\epsilon_p$} $- \epsilon_p$ \label{csept:hb}

\eIf{$\underline{\mu_{\topproxy{\hat{k}}}} \geq \gamma$}
{\KwRet $\topproxy{\hat{k}}$}
{\KwRet $\{x \in S \mid dist^O(x) \leq r\}$}
\label{csept:e}}

\Fn{\fhb{$D$, $\delta$, $\epsilon$}}{
$S \gets$ \sample{$D$, $\lceil \frac{log(\delta)}{-2\epsilon^2} \rceil$}\;
\KwRet $\hat{\mu} \gets \frac{|\{x \in S \mid dist^O(x) \leq r\}|}{|S|}$\;
}

\end{algorithm}
\ULforem

\vspace*{-1ex} 
\section{Experiments}
\label{sec:experiment}
Our extensive experiments
(1) assess the performance of \pqa to demonstrate its optimality w.r.t. CR and success probability under {\sc Proxy Quality} assumption (\S~\ref{subsec:pqa}), 
(2) assess the performance of  \csa to demonstrate its minimal oracle usage and success probability under {\sc Core Set Closure} assumption (\S~\ref{subsec:csa}), 
(3) compare \pqe, \cse 
with the baselines on CR and success probability under the \textit{same oracle usage} (\S~\ref{subsec:cse_pqe}) and under \textit{varied oracle budgets} (\S~\ref{subsec:oracle_efficiency}). 
(4) Compare \pqe, \cse 
with the baselines w.r.t. query time (\S~\ref{subsec:time_efficiency}). 
(5) Compare \pqe, \cse 
with the baselines w.r.t. CPU overhead, CR, and success probability on datasets of various sizes and domains  (\S~\ref{subsec:scalability}).

\subsection{Experimental Setup}
\subsubsection{Datasets and Proxy Models}

\vspace*{-1.1ex}
\paragraph{\textbf{Multi-label Image Recognition}}
\textit{VOC} \cite{voc2007} and \textit{COCO} \cite{coco2014} are widely used benchmarks in multi-label recognition tasks. 
The validation set of COCO consists of $40,504$ images from $80$ classes, and VOC contains $4,952$ images from $20$ object categories. We also uniformly sample a $8000$-image subset from COCO, denoted as COCO (small). We use COCO and COCO (small) in different experiments.
\vspace*{-1ex}
\paragraph{\textbf{Medical}}
\textit{Mimic-III} \cite{mimic_iii} and \textit{eICU} \cite{eICU2018} are two publicly available clinical datasets, that include patient trajectories, demographics collected by daily ICU admissions, and clinical measurements. 
After pruning records with only one admission, we obtain a Mimic-III subset of $4,243$ records and an eICU subset of $8,235$ records.

\paragraph{\textbf{Video}}
We use the \textit{night-street} dataset \cite{jackson_dataset} to support queries over classification tasks. Each video frame has a Boolean label indicating whether or not it contains a car. 
We uniformly draw a subset of  $10,000$ frames from the original dataset for evaluation.  
\begin{table}
    \centering
    \begin{small}
\begin{tabular}{ |c|c|c|c|  }
\hline
Datasets & Oracle & Proxy & Query targets \\
\hline
VOC\&COCO &\multirow{1}{*}{Human labeler}	&\multirow{1}{*}{ML-GCN\cite{chen2019mlgcn}}& \multirow{1}{*}{Similar images}\\
Mimic-III\&eICU &\multirow{1}{*}{Physicians}	&\multirow{1}{*}{LIG-Doctor\cite{jose2021lig}}& \multirow{1}{*}{Similar patients}\\
night-street  & Mask R-CNN\cite{he2017mask}	& ResNet-50\cite{he2015res}	& Car frames \\
\hline
\end{tabular}
\end{small}
\label{data:oracle_proxy}
\vspace{-6mm}
\end{table}
\vspace*{-1ex} 
\subsubsection{Baselines}
We consider the following baselines. 

\textbf{SUPG}
The closest work to ours is \supg \cite{DBLP:journals/pvldb/KangGBHZ20}.
\dujian{
\supg uses oracle $\oracle'$ with a Boolean output and a proxy model $\proxy'$ which outputs a score in $[0,1]$.
Given a query object $q$ and a radius $r$, our problem can be mapped to a binary classification problem: \textit{for each object $x$, is it a near neighbor to $q$ w.r.t. $r$?} {Given oracle $O$ and proxy $P$ for our problem,} a natural oracle predicate for SUPG should output $1$ when the given object is a near neighbor and $0$ otherwise. This translates to $O'(x) = 1$ iff $dist^O(x)<=r$.
Similarly, a natural proxy model for SUPG should give high scores when the corresponding object is more probable to be a near neighbor. As illustrated in Figure \ref{assmp_just:pqa}, with a properly chosen proxy, proxy distance $dist^P(.)$ is a good approximation for oracle distance $dist^O(.)$. Given that $dist^P(x) \in [0,1]$ in our problem, we choose $\proxy'(x)=1-dist^P(x)$ as the proxy model for SUPG. Intuitively, if object $x$ has a small proxy distance, $x$ is more likely to have a small oracle distance as well and hence more probable to be classified as $O'(x)=1$, which is properly reflected by a high value of $\proxy'(x)$. 
}


\dujian{
\textbf{Probabilistic Top-K~\cite{10.1145/3448016.3452786}}. 
This baseline studies approximate Top-K queries and delivers solutions with statistical guarantees. Given a query, there exists a direct mapping from our FRNN query to a Top-K query. For example, given a query object $q$ and radius $r$, an FRNN query asks for an answer $Ans$ which comprises all near-neighbours within the radius $r$ to $q$. Naturally, we can re-write this query in Top-K semantics: given query object $q$, return the Top-$K$ nearest neighbors to $q$ where $K=|Ans|$ according to the aforementioned FRNN query. Furthermore, this Top-K baseline relies on distribution over oracle predictions, which can be obtained from our \textsc{Proxy Quality} assumption in \pqa.
}

\dujian{
\textbf{Sample2Test} Given an FRNN RT (resp. PT) query, this baseline first probes samples w.r.t. a given oracle budget, and then selects the optimal proxy prefix as the answer according to sample precision (resp. recall). This is the approach used in probabilistic predicates (PP) \cite{DBLP:conf/sigmod/LuCKC18}, NoScope \cite{DBLP:journals/pvldb/KangEABZ17}, and also serves as a baseline in SUPG~\cite{DBLP:journals/pvldb/KangGBHZ20}.
Given a sample $S \subset D$ and a proxy index $k$, denote $S^k=S \cap \topkproxy{k}$.  The sample precision at $k$ is $Precision_S(k) = \frac{|S^k \cap \nearneighbor|}{|S^k|}$ and the sample recall is $Recall_S(k) = \frac{|S^k \cap \nearneighbor|}{|S \cap \nearneighbor|}$. Given dataset $D$ and the target $\gamma$, this baseline returns $\topkproxy{k'}$ where $k' = max\{1 \leq k \leq |D| \mid Precision_S(k) \geq \gamma\}$ for PT queries, and $k' = min\{1 \leq k \leq |D| \mid Recall_S(k) \geq \gamma\}$ for RT queries. We select the largest (resp. smallest) proxy prefix for PT (resp. RT) to improve CR.
}

\dujian{
\textbf{Scan2Test} We also consider the naive approach which probes all objects with the oracle and selects the correct answer set for a given query. This approach is used as the baseline in~\cite{10.1145/3448016.3452786}.
}

\subsubsection{Evaluation Measures} 
For both RT and PT queries, we are interested in three measures: (i)  \textit{empirical} success probability in relation  to the \textit{required} success  probabilities; 
(ii)  average CR of answers returned by different methods; and 
(iii)  query processing time including (a) CPU overhead and (b) number of oracle calls. We do not compare proxy time since it is identical for all approaches and is only a fraction of the overall query processing time. 

\subsubsection{Protocol} 
Our evaluation protocol randomly chooses several query objects from a dataset and aggregates our measures for those query objects. 
In Section~\ref{subsec:pqa}, we randomly choose $200$ query objects and aggregate their results. In other experiments, we randomly choose $50$ query objects and execute each query $10$ times and aggregate their results. This is because \pqa is deterministic  while other algorithms are subject to randomness, so we average  over multiple trials. 
We use cosine distance whenever the model outputs are  multi-dimensional vectors:  $dist_{cos}(\mathbf{y_1},\mathbf{y_2})=1-\frac{\mathbf{y_1 \cdot y_2}}{\|\mathbf{y_1}\|\cdot\|\mathbf{y_2}\|}$, \dujian{given its wide application in proximity query processing~\cite{7349760,mingdong2018methods,7577578}}. When the output is scalar (e.g., Boolean labels), we use the absolute difference $dist_{abs}(\mathbf{y_1},\mathbf{y_2})=|\mathbf{y_1}-\mathbf{y_2}|$ as the distance function, \dujian{which allows us to generalize SUPG query with boolean oracle predicates}. In all cases, the radius threshold is $r=0.9$. The choice of distances and thresholds has no impact on our statistical guarantees. 

{\bf Default values.} Unless otherwise stated, we set $\gamma$ (recall and precision targets) to $0.95$ and $\delta$ to $0.1$ in all our experiments. \dujian{We add a black dashed line ($-\cdot-$) at the level of $1-\delta$ in figures to help visually track the success probability of each approach.} 
We empirically choose $\sigma_0=0.3$ for \pqe, $\epsilon_p=0.1\%$ and $b'=100$ for \cse-PT, $\epsilon_r=10\%$ and $\deltar = 0.05$ for \cse-RT according to our experiment results.
We choose a small $\epsilon_p$ for \cse-PT to improve the probabilistic lower bound for precision, and a relatively large $\epsilon_r$ for \cse-RT to reduce the oracle usage incurred by applying Hoeffding Bounds. 
In addition, we only report results of \cse when $m=1$ (see Eq.~\ref{eq:csa_m1}) given its dominating performance over other $m$ settings.  

Our algorithms are implemented in Python 3.7 and experiments are conducted on a M1 Pro chip @ $3.22$GHz with a 16GB RAM. 

\vspace*{-1ex} 
\subsection{\pqa Success at Maximal CR}
\label{subsec:pqa}
\pqa finds high probability valid answers of maximal expected CR with zero oracle calls, whenever {\sc Proxy Quality} assumption holds (\S~\ref{subsec:alg_pq}). We test it on two semi-synthetic datasets. Specifically, we use real proxy distances from VOC and eICU, and synthesize oracle distances with a normal distribution $\mathcal{N}(0,\sigma=0.1)$.
We clip the normal distribution to $[0,1]$ to agree with the output range of our distance measures.  We demonstrate CR maximality and success probability guarantees of \pqa by comparing it with a series of variants. 
Recall that \pqa returns the top-$k^*$ objects of smallest proxy distances as the answer. We measure the success probability and CR when using a perturbed $k^*$. We try perturbations ranging from $-20\%$ to $20\%$ by returning top-$(1+perturb) \cdot k^*$ for $-20\% \leq perturb \leq 20\%$. 

The results are shown in Figure \ref{exp:pqa}. The two top plots summarize RT queries. Recall $\delta=0.1$, which requires success probability being no less than $90\%$. On VOC, with zero perturbation, \pqa achieves a $92\%$ empirical success probability and $39\%$ CR for RT queries. On eICU, the empirical success probability and CR are $90\%$ and $38\%$ respectively. With negative perturbation, empirical success probability quickly shrinks to nearly zero; with positive perturbation, CR starts to drop. This observation clearly demonstrates that \pqa gives the highest answer CR while respecting the success probability constraint. The two bottom plots are for PT queries, and are similar to RT. The unperturbed \pqa achieves $94\%$ empirical success probability on both datasets, a $53\%$ CR on VOC, and a $42\%$ CR on eICU. Any perturbation to $k^*$ either fails the success probability constraint or degrades CR.

\dujian{We compare \pqa and \bltopk on the same semi-synthetic VOC dataset (see Figure \ref{exp:pqa_topk}). Both approaches achieve desired success probability targets.
However, \bltopk suffers from huge oracle usage while \pqa needs no oracle calls,  indicating  \pqa is capable of efficient query processing when proxy quality distribution is known. Furthermore, we investigate the sensitivity of \pqa  to $\sigma$\footnote{We assume a normal distribution $\epsilon_i \sim \mathcal{N}(0, \sigma)$ for \pqa to compute $\Phi(D)$.},  shown in Figure \ref{exp:pqa_pqe}. We test $\pqa$ on VOC  with various $\sigma$ values and report \pqe performance\footnote{Budget of \pqe set equal to the oracle cost incurred by \cse for the given query.} for comparison purposes. As $\sigma$ increases, for both query types, the success probability of \pqa increases while CR decreases. This agrees with the intuition that, \textit{as the proxy quality gets worse, \pqa becomes more conservative to improve success probability at the cost of CR degradation.} Since \pqe does not rely on external $\sigma$, both success probability and CR are constant and higher than \pqa, given that \pqe has the flexibility to probe samples with the oracle. 
}
\begin{figure}
  \begin{subfigure}[b]{0.23\textwidth}
    \includegraphics[width=\textwidth]{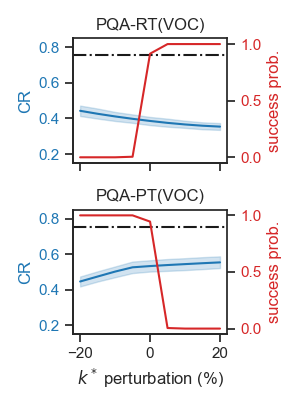}
  \end{subfigure}
 \hfill
    \begin{subfigure}[b]{0.23\textwidth}
    \includegraphics[width=\textwidth]{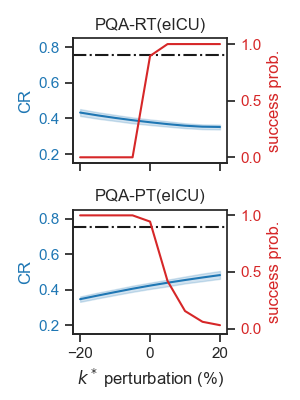}
  \end{subfigure}
  \vspace*{-5mm}
  \caption{PQA with perturbed $k^*$ on VOC and eICU datasets.}
  \label{exp:pqa}
  \begin{subfigure}[b]{0.23\textwidth}
    \includegraphics[width=\textwidth]{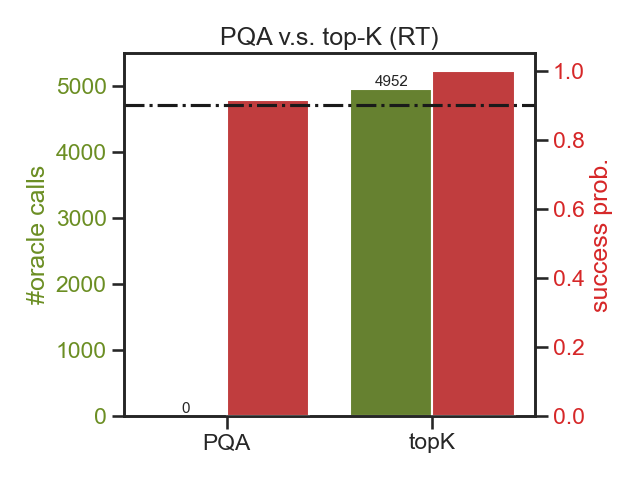}
  \end{subfigure}
 \hfill
    \begin{subfigure}[b]{0.23\textwidth}
    \includegraphics[width=\textwidth]{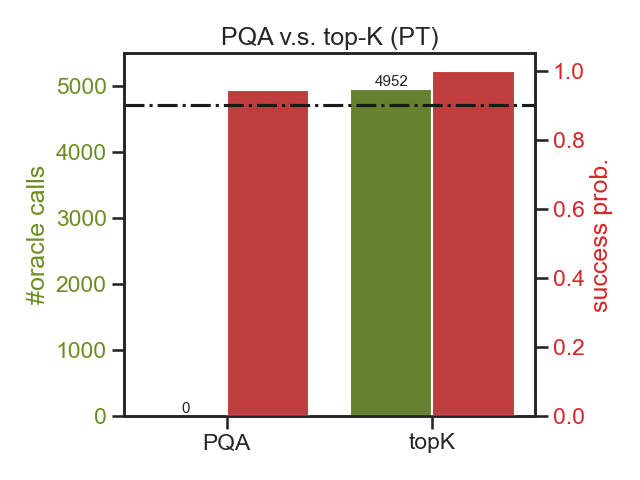}
  \end{subfigure}
  \vspace*{-5mm}
  \caption{\dujian{Comparison of \pqa and \bltopk  on VOC.}}
  \label{exp:pqa_topk}

  \begin{subfigure}[b]{0.23\textwidth}
    \includegraphics[width=\textwidth]{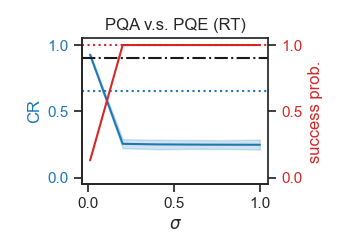}
  \end{subfigure}
 \hfill
    \begin{subfigure}[b]{0.23\textwidth}
    \includegraphics[width=\textwidth]{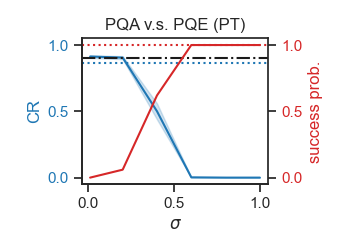}
  \end{subfigure}
  \vspace*{-6mm}
  \caption{\dujian{\pqa (solid line) v.s. \pqe (dotted line) on VOC.}}
  \label{exp:pqa_pqe}
\end{figure}


\vspace*{-1ex} 
\subsection{\csa with Minimal Oracle Usage}
\label{subsec:csa}
\csa ensures high success probabilities with minimal oracle usage under {\sc Core Set Closure} assumption (\S~\ref{subsec:alg_cs}). 
We implement
an \textit{exact} algorithm and two approximation algorithms to compute $s^*$ and $m^*$ (\S~\ref{subsubsec:csa_assumption}), 
\textit{Approx-s1} and \textit{Approx-m1}. We compare  these algorithms to two baselines, \textit{Rand-s} and \textit{Rand-sm}. Specifically, \textit{Rand-s} randomly chooses $s$ and sets $m=\mlow(s)$, whereas \textit{Rand-sm} chooses both $s, m$ at random. For each query, we precompute the core set size and feed it to all approaches. We study the empirical success probability and oracle usage  on VOC and eICU. The results for RT and PT queries are reported in Figure \ref{exp:csa}. \dujian{Especially, we report standard deviation of oracle usage for both query types on both datasets.} 
We also report CPU overheads  in Figure \ref{exp:csa_overheads}.

\begin{figure*}
  \begin{subfigure}[b]{0.24\textwidth}
    \includegraphics[width=\textwidth]{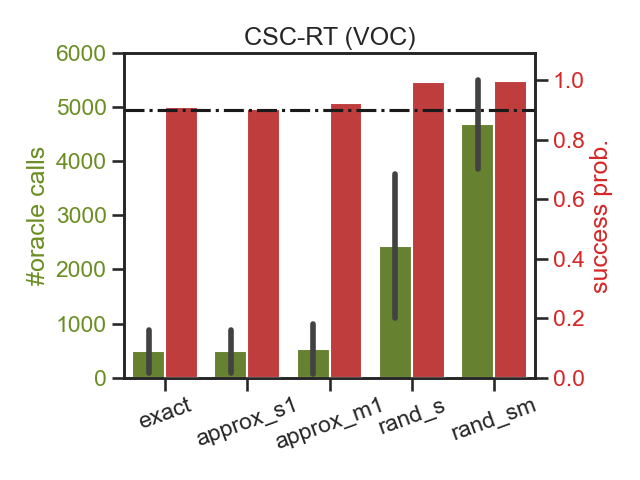}
  \end{subfigure}
  \hfill
  \begin{subfigure}[b]{0.24\textwidth}
    \includegraphics[width=\textwidth]{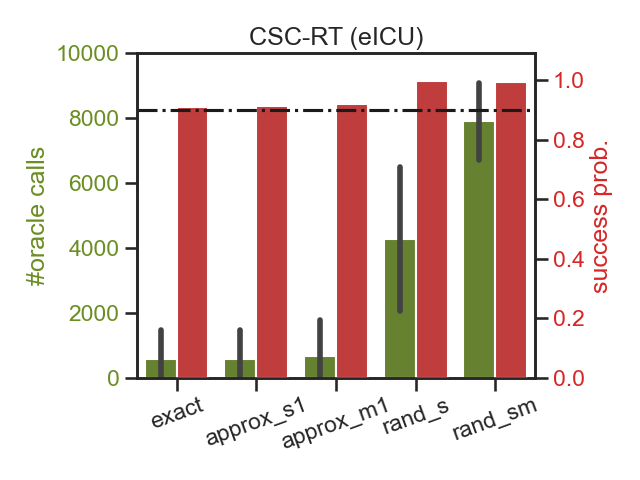}
  \end{subfigure}
\hfill
  \begin{subfigure}[b]{0.24\textwidth}
    \includegraphics[width=\textwidth]{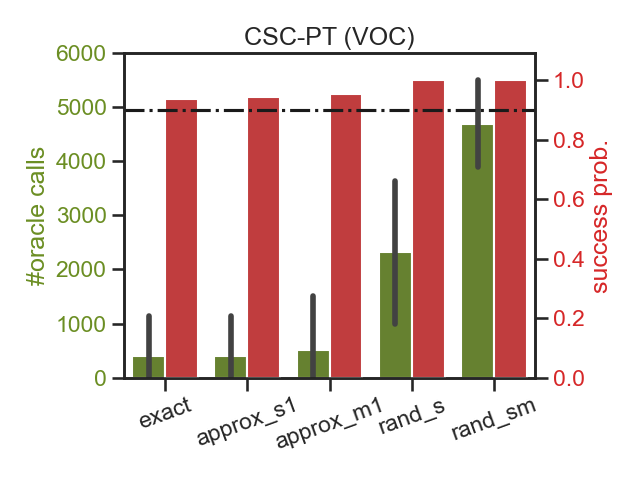}
  \end{subfigure}
  \hfill
  \begin{subfigure}[b]{0.24\textwidth}
    \includegraphics[width=\textwidth]{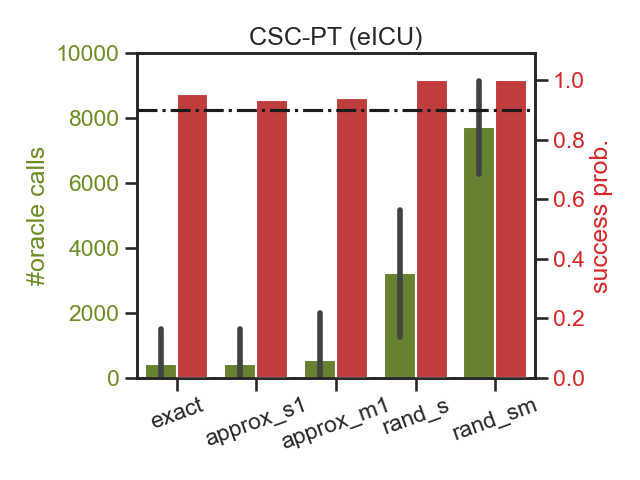}
  \end{subfigure}  
    \vspace*{-6mm}
  \caption{\dujian{Oracle usage and empirical success probability by \csa-RT and \csa-PT.}}
  \label{exp:csa}
  \vspace*{-3mm}
\end{figure*}

For RT queries, all approaches achieve high  empirical success probability. The exact algorithm invokes the oracle on only $9.8\%$ objects in VOC and $7.1\%$ objects in eICU. The oracle usage of approximation algorithms is just up to $1.1\%$ more than the exact algorithm.
However, the baseline \textit{Rand-s} applies the oracle on at least $49.3\%$ objects and \textit{Rand-sm} makes oracles calls on at least $94.6\%$ objects. 

\begin{figure}[htpb]
  \begin{subfigure}[b]{0.23\textwidth}
    \includegraphics[width=\textwidth]{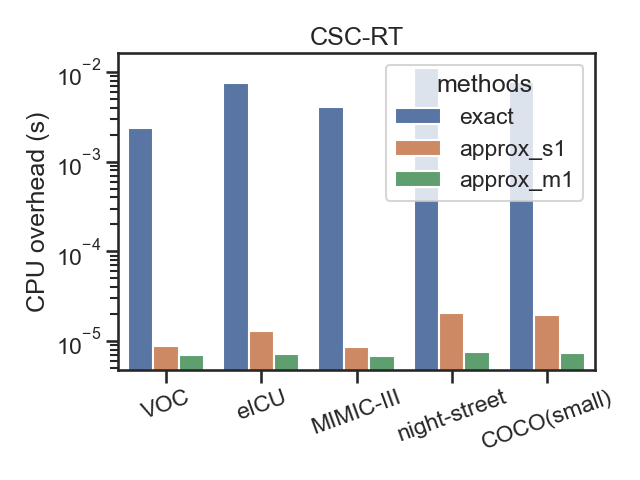}
  \end{subfigure}
  \hfill
  \begin{subfigure}[b]{0.23\textwidth}
    \includegraphics[width=\textwidth]{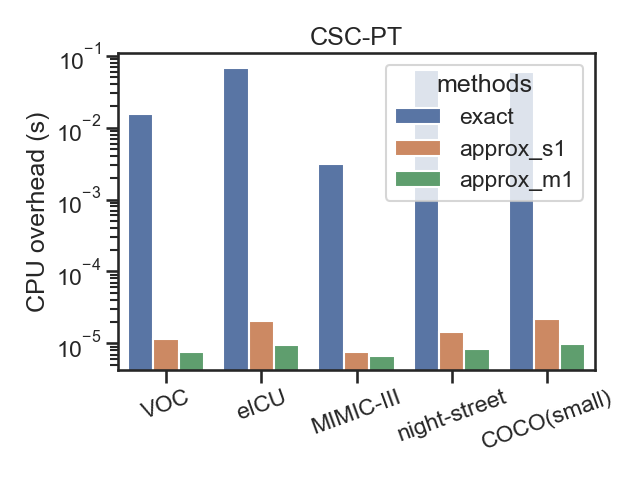}
  \end{subfigure}
    \vspace*{-6mm}
  \caption{\csa CPU overheads}
  \label{exp:csa_overheads}

  \begin{subfigure}[b]{0.23\textwidth}
    \includegraphics[width=\textwidth]{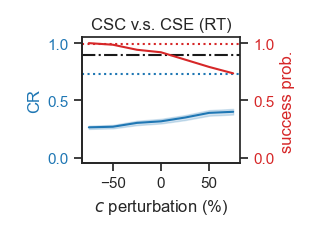}
  \end{subfigure}
  \hfill
  \begin{subfigure}[b]{0.23\textwidth}
    \includegraphics[width=\textwidth]{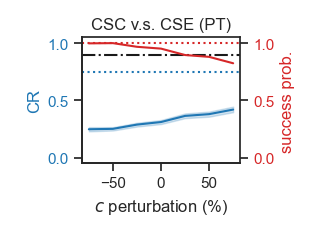}
  \end{subfigure}
    \vspace*{-5mm}
  \caption{\dujian{\csa (solid line) v.s. \cse (dotted line) on VOC.}}
  \label{exp:csc_cse}
  \vspace*{-5mm}
\end{figure}


\begin{figure*}[htpb]
  \begin{subfigure}[b]{0.24\textwidth}
    \includegraphics[width=\textwidth]{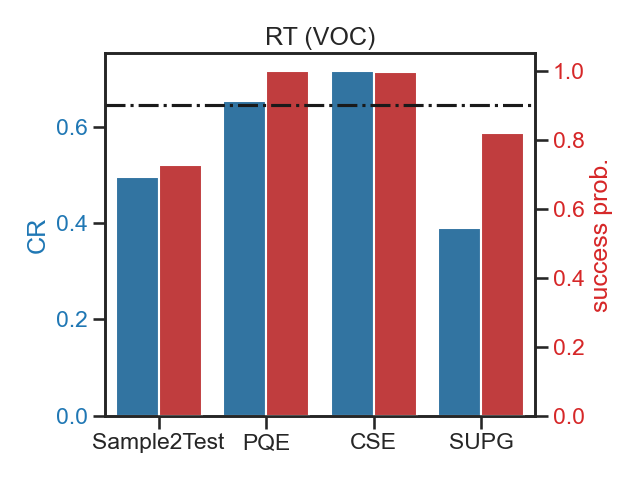}
  \end{subfigure}
  \hfill
  \begin{subfigure}[b]{0.24\textwidth}
    \includegraphics[width=\textwidth]{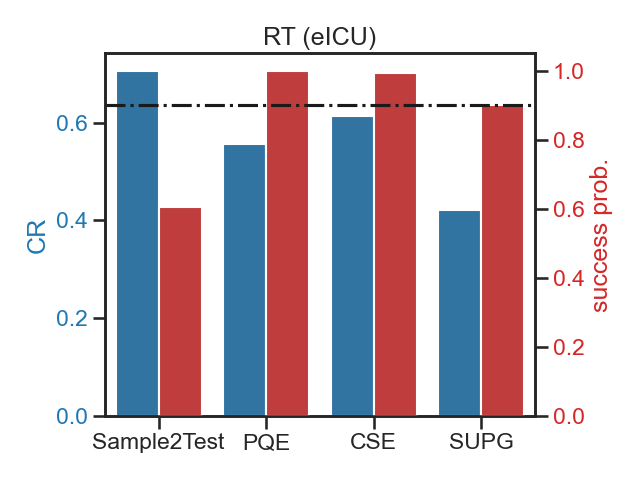}
  \end{subfigure}
\hfill
  \begin{subfigure}[b]{0.24\textwidth}
    \includegraphics[width=\textwidth]{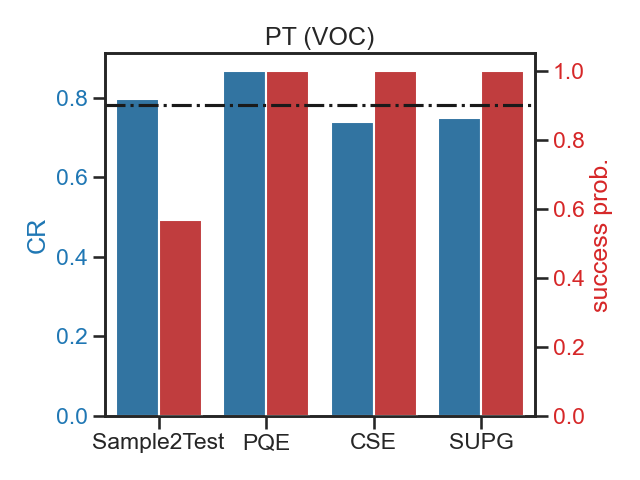}
  \end{subfigure}
  \hfill
  \begin{subfigure}[b]{0.24\textwidth}
    \includegraphics[width=\textwidth]{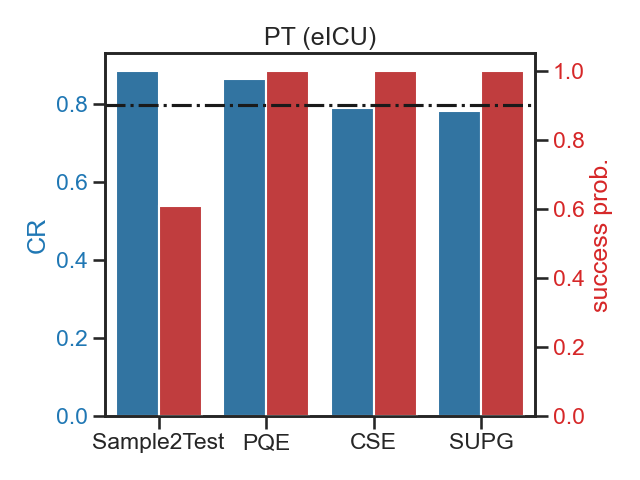}
  \end{subfigure}
    \vspace*{-5mm}
  \caption{\dujian{CR and  empirical success probability by \pqe, \cse, \supg, \sampletest.}}
  \label{exp:cse}

  \begin{subfigure}[b]{0.24\textwidth}
    \includegraphics[width=\textwidth]{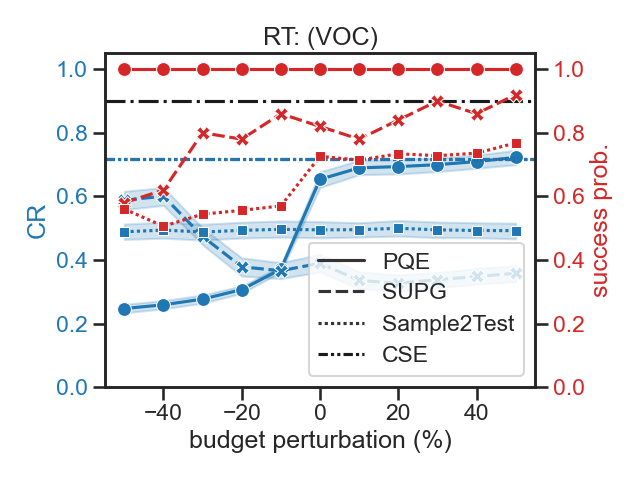}
  \end{subfigure}
  \hfill
  \begin{subfigure}[b]{0.24\textwidth}
    \includegraphics[width=\textwidth]{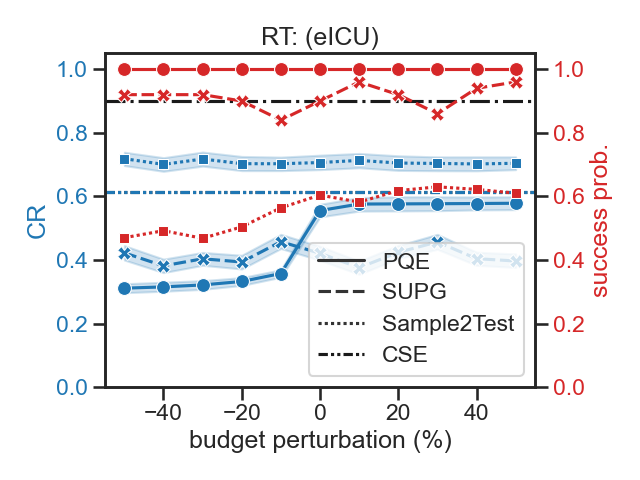}
  \end{subfigure}
\hfill
  \begin{subfigure}[b]{0.24\textwidth}
    \includegraphics[width=\textwidth]{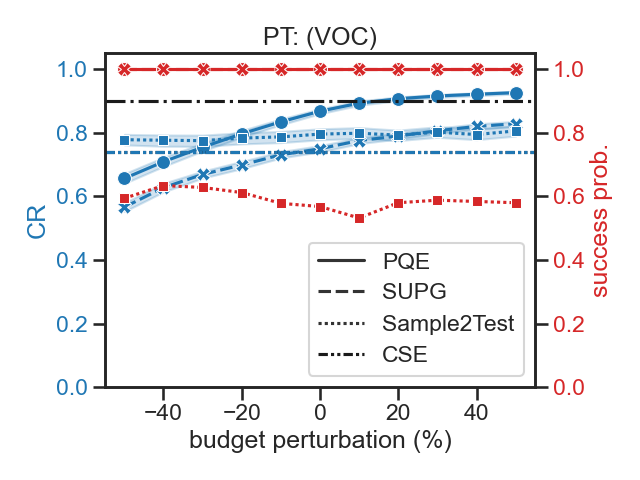}
  \end{subfigure}
  \hfill
  \begin{subfigure}[b]{0.24\textwidth}
    \includegraphics[width=\textwidth]{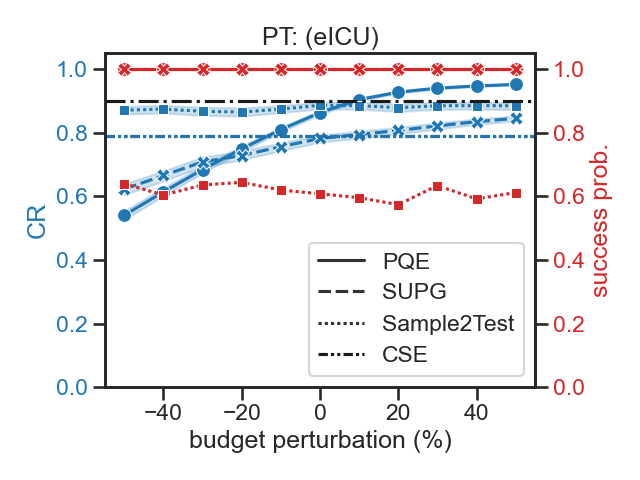}
  \end{subfigure}
    \vspace*{-5mm}
  \caption{\dujian{CR and  empirical success probability by \pqe, \supg, \sampletest with perturbed budget from \cse}}
  \label{exp:oracle_efficiency}
\end{figure*}
\vspace*{-1ex}
\begin{figure}[!tbp]
\vspace*{-1ex}
  \begin{subfigure}[b]{0.23\textwidth}
    \includegraphics[width=\textwidth]{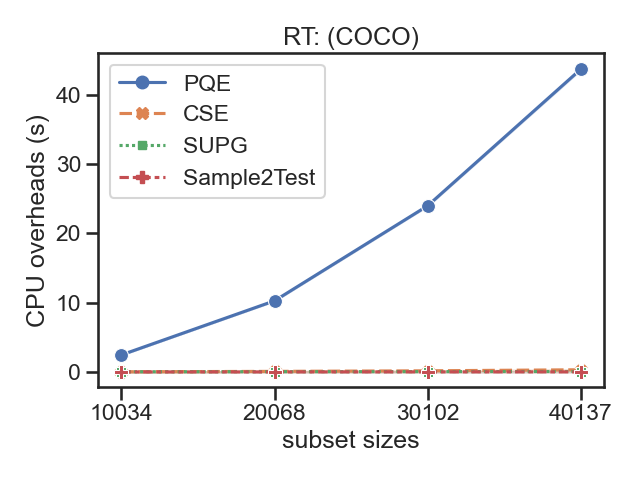}
  \end{subfigure}
  \hfill
  \begin{subfigure}[b]{0.23\textwidth}
    \includegraphics[width=\textwidth]{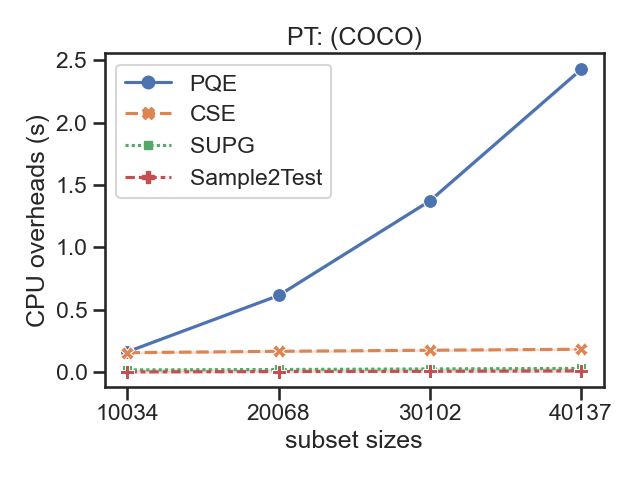}
  \end{subfigure}
    \vspace*{-5mm}
  \caption{\dujian{Scalability test: CPU overheads}}
  \label{exp:scalability_overhead}
  \vspace*{-5mm}
\end{figure}

\vspace*{1ex}
Results of PT queries are similar to RT. All approaches achieve high  empirical success probability. The exact algorithm has the smallest oracle usage, which accounts for $8.1\%$ objects in VOC and $5.2\%$ objects in eICU. The approximation algorithms incur an oracle usage which is just up to $2.1\%$ higher than the exact algorithm. The baseline method \textit{Rand-s} calls the oracle on at least $39\%$ objects, and \textit{Rand-sm} makes oracle calls for at least $93.7\%$ objects. 


We are also interested in the CPU overhead of the exact algorithm and the two approximation algorithms. Results on our five real-world datasets are summarized in Figure \ref{exp:csa_overheads}. Clearly, the exact algorithm has a larger CPU overhead in comparison to the  two approximation algorithms. {\sl Specifically, the approximation algorithms achieve a speedup up to $1466 \times$ for RT queries and $7391 \times$ for PT queries on CPU overheads, in comparison to the exact method. } 

\dujian{On VOC, we investigate how sensitive \csa is to the input core set size $c$ and include \cse performance for the same query for comparison (Figure \ref{exp:csc_cse}). As $c$ increases, for both query types, success probability of \csa decreases and CR increases. Since \cse estimates $c$ internally, both success probability and CR are agnostic to external $c$ changes. Note that the \cse performance is generally better than \csa, which is attributed to the fact that \cse has more flexibility to probe objects with oracle for the additional $c$ estimation.
}

\vspace*{-1ex} 
\subsection{\cse and \pqe vs \supg and \sampletest}
\label{subsec:cse_pqe}
We implement \cse and \pqe and compare them to \supg \dujian{and \sampletest}. For a fair comparison, we use the oracle usage incurred by \cse as the budget for \pqe, \supg, \dujian{and \sampletest}. We measure empirical success probability and CR on VOC and eICU (Figure \ref{exp:cse}). 
For RT, \cse and \pqe achieve high  empirical success probability on both datasets, while \supg fails the success threshold empirically by a margin of $10\%$ on VOC. On CR, \pqe outperforms \supg by a margin up to $26\%$, and \cse outperforms \supg up to $33\%$. For PT, all approaches achieve high  empirical success probability. On CR, \pqe outperforms \supg by up to $12\%$ and \cse achieves comparable CR 
to \supg.
\dujian{\sampletest continuously fails the success probability for both query types on both datasets.}

Failures of \supg on RT queries stem from the sample mean and variance it uses \textit{without error bounds}, which introduces uncontrolled uncertainty and degrades statistical guarantees. 
\vspace*{-1ex} 
\subsection{Oracle Efficiency}
\label{subsec:oracle_efficiency}
To measure Oracle efficiency, we perturb the oracle usage incurred by \cse and use it as the budget for \pqe, \supg, \dujian{\sampletest}. The CR of \cse is also plotted as a baseline. Results are in Figure \ref{exp:oracle_efficiency}. For RT queries, \pqe achieves high  empirical success probability on both datasets, while \supg \dujian{and \sampletest} fails frequently on VOC especially with small budgets. This indicates that \cse is the most oracle efficient approach for RT queries. For PT queries, all approaches \dujian{except \sampletest} achieve high  empirical success probability on both datasets.

\vspace*{-1ex} 
\subsection{Time Efficiency}
\label{subsec:time_efficiency}
The running time of a query is composed of CPU overhead and model usage including proxy and oracle calls. We measure CPU overhead for each approach locally and approximate model usage by timing the number of model calls and average time taken by each call. For instance, on medical datasets (MIMIC-III \& eICU), the oracle is a human physician whose average diagnosis time is $15$ minutes \cite{tai2007time}, while the proxy is a recurrent neural network taking roughly $1$ millisecond for each call \cite{jose2021lig}. Results are reported in Table \ref{exp:time_efficiency_rt}, \ref{exp:time_efficiency_pt} with the best results in bold and saving ratios w.r.t. \supg. \dujian{For both query types, \sampletest and \scantest are the two most time-consuming approaches.}
\begin{table}
    \centering
    \begin{small}
\begin{tabular}{ |c|c|c|c|c|c|c|  }
\hline
Time / Hours & \pqe & \cse & \supg & \dujian{\scantest} &\dujian{\sampletest}\\
\hline
VOC &21.22	&\textbf{15.88}	&24.12 & 27.51 & 26.84\\
COCO(small) &32.21	&\textbf{19.66}	&30.46 & 44.44 & 44.06\\
\hline
MIMIC-III &124.3	&\textbf{108.1}	&865 &1061 & 1045.08\\
eICU &1337.3	&\textbf{679.2}	&925 &2059 & 2009.16\\
\hline
night-street  &0.38	&\textbf{0.11}	&0.19 &1.11 &1.15\\
\hline
\end{tabular}
\caption{\dujian{RT queries:  query  time by \cse, \pqe, and baselines}}
\label{exp:time_efficiency_rt}
\vspace{-3mm}
\begin{tabular}{ |c|c|c|c|c|c|  }
\hline
Time / Hours & \pqe & \cse & \supg &\dujian{\scantest} &\dujian{\sampletest}\\ 
\hline
VOC & \textbf{4.71}	& 6.3	& 7.52 &27.51 & 26.41\\ 
COCO(small) & \textbf{11.93} & 13.18 & 19.45 &44.44 &39.97\\ 
\hline
MIMIC-III &997.3 & \textbf{132.7} &158.4 &1061 &1060.69\\
eICU &\textbf{569.4}	&617.3	&871.1 &2059 & 1973.64\\
\hline
night-street  &0.28	&\textbf{0.25}	&0.35 & 1.11 & 1.17\\
\hline
\end{tabular}
\caption{\dujian{PT queries:  query time by \cse, \pqe, and baselines.}}
\label{exp:time_efficiency_pt}
\vspace{-3mm}
\begin{tabular}{ |c||c|c|  }
\hline
|D| & \multicolumn{1}{|c|}{Success Prob. } & \multicolumn{1}{|c|}{CR} \\
\hline
	&\pqe/\cse/\supg/\dujian{\sampletest}	&\pqe/\cse/\supg/\dujian{\sampletest}\\
\hline
10034 & 1/0.99/\textcolor{red}{0.78}/\textcolor{red}{0.69} 	&0.72/\textbf{0.74}/0.46/0.53 \\
20068 & 1/0.99/\textcolor{red}{0.76}/\textcolor{red}{0.59}	&0.65/\textbf{0.66}/0.46/0.58 \\
30102 & 1/0.99/\textcolor{red}{0.76}/\textcolor{red}{0.65}	&0.63/\textbf{0.65}/0.44/0.5 \\
40137 & 1/0.98/\textcolor{red}{0.75}/\textcolor{red}{0.66}	&0.68/\textbf{0.69}/0.5/0.54 \\
\hline
\end{tabular}
\caption{\dujian{Scalability test for RT queries}}
\label{exp:scalability_rt}
\vspace{-4mm}
\begin{tabular}{ |c|c|c|  }
\hline
|D| & \multicolumn{1}{|c|}{Success Prob. } & \multicolumn{1}{|c|}{CR} \\
\hline
	&\pqe/\cse/\supg/\dujian{\sampletest}	&\pqe/\cse/\supg/\dujian{\sampletest}\\
\hline
10034		&1/1/1/\textcolor{red}{0.52} & \textbf{0.73}/0.67/0.59/\textbf{0.73} \\
20068		&1/1/1/\textcolor{red}{0.48} & {0.69}/0.65/0.59/\textbf{0.71} \\
30102		&1/1/1/\textcolor{red}{0.5} & {0.66}/0.64/0.62/\textbf{0.68} \\
40137		&1/1/1/\textcolor{red}{0.51} & {0.67}/0.64/0.66/\textbf{0.74}\\
\hline
\end{tabular}
\end{small}
\caption{\dujian{Scalability test for PT queries}}
\label{exp:scalability_pt}
\vspace*{-10mm}
\end{table}

\vspace*{-1.5ex} 
\subsection{Scalability}
\label{subsec:scalability}
We measure CPU overhead, success probability, and CR of \pqe, \cse, \supg, \dujian{and \sampletest}. We  uniformly draw subsets of the original COCO dataset ($25\%$, $50\%$, $75\%$, and $100\%$). To make a fair comparison, we use the oracle usage incurred by \cse as the budget for \pqe and \supg. Results are reported in Figure \ref{exp:scalability_overhead} and Table \ref{exp:scalability_rt}, \ref{exp:scalability_pt}. For RT queries, \cse, \supg, \dujian{and \sampletest} have a reasonably low CPU overhead. For PT queries, \pqe has the highest CPU overhead, $2.5$ seconds per query, while the overhead of \cse, \supg, \dujian{and \sampletest} is less than $0.2$ seconds.

\vspace*{-1ex} 
\section{Related Work}
\label{sec:related}

\noindent {\bf Query approximation.}
Query approximation techniques~\cite{DBLP:journals/dase/LiL18} can be categorized into (1) online aggregation: select samples online and use them to answer OLAP queries, and (2) offline synopses generation 
\eat{: generate synopses based on a-priori knowledge (e.g., data statistics or query workload) and use them} to facilitate OLAP queries. Our work adopts a probabilistic top-k approach~\cite{DBLP:conf/vldb/TheobaldWS04} and is significantly different from these.

\noindent {\bf FRNN query.}
FRNN query answering systems \cite{bentley1975frnn,BENTLEY1977209} build spatial indexes on the whole DB, which requires oracle calls on every single object. Our work focuses on reducing the oracle usage and is clearly distinguished from this line of work.

\noindent {\bf Optimizing ML inference.}
Several recent approaches were proposed to speed up the application of an ML model. 
Existing approaches follow either an in-database~\cite{DBLP:conf/nsdi/CrankshawWZFGS17} or in-application approach~\cite{DBLP:conf/kdd/AhmedABCCDDEFFG19}. Amazon Aurora is an example of an in-database containerized solution that enables external calls from SQL queries to ML models in SageMaker\footnote{\url{https://aws.amazon.com/fr/sagemaker/}}. Containerized execution introduces overhead in prediction latency. To mitigate that, Google's BigQuery ML\footnote{\url{https://cloud.google.com/bigquery-ml/docs}} and Microsoft's Raven were developed~\cite{DBLP:conf/cidr/KaranasosIPSPPX20}. Compared to Raven, BigQuery ML relies mostly on hard-coded models and targets batch predictions, since it inherits a relatively high startup cost. Raven and its runtime environment ONNX~\cite{DBLP:conf/osdi/ChenMJZYSCWHCGK18} offer the additional ability to make tuple-level inference. 

\noindent{\bf Combining queries and ML inference.}
Bolukbasi et al.~\cite{DBLP:conf/icml/BolukbasiWDS17} enable  incremental predictions for neural networks.  
\eat{The authors formulate a global objective for 
by learning a policy and solve it by reducing it to a layer-by-layer weighted binary classification problem.}  Computation time is reduced by pruning examples that are classified in earlier layers, selected adaptively. \eat{This is further extended to learn how to adaptively select the network to be evaluated for each example. This was shown to achieve a good trade-off between classification accuracy and evaluation time on ImageNet.} 
Kang et al.~\cite{DBLP:journals/pvldb/KangEABZ17} present NOSCOPE, a system for querying videos that can reduce the cost of neural network video analysis by up to three orders of magnitude via inference-optimized model search.\eat{Given a target video, object to detect, and reference neural network, NOSCOPE automatically searches for and trains a sequence, or cascade, of models that preserve the accuracy of the reference network but are specialized to the target video and are therefore far less computationally expensive.}
Lu et al.~\cite{DBLP:conf/sigmod/LuCKC18} \dujian{and Yang et al.~\cite{10.14778/3547305.3547310}} use probabilistic predicates to filter data blobs that do not satisfy the query \dujian{and empirically increase data reduction rates}. 
\eat{To support complex predicates and avoid per-query training,  They augment a cost-based query optimizer to choose plans with appropriate combinations of simpler probabilistic predicates. }
\eat{While introducing a configurable amount of error, They show that their solution boosts  the performance of ML queries.}
Anderson et al.~\cite{DBLP:conf/icde/AndersonCRW19} use a hierarchical model to reduce the runtime cost of queries over visual content. \eat{Prior work used cascades, e.g., smaller CNNs, to optimize the computational cost of inference. By treating transformations of physical representation of the input image  as  part  of  query  optimization, that  is,  by  including image  transformations  such  as  resolution  scaling within  the  cascade, they   optimize  data  handling costs. }
\eat{The authors propose Tahoma, which generates and evaluates many potential classifier cascades. }
Gao et al.~\cite{DBLP:conf/sigmod/GaoXAY21} introduce a Multi-Level Splitting Sampling to let one "promising" sample path prefix generate multiple "offspring" paths, and direct Monte-Carlo based simulations toward more promising paths. \eat{This frees users from manually tuning models.} 
\dujian{Lai et al.~\cite{10.1145/3448016.3452786} studies approximate Top-K query with light-weight proxy models that generate oracle label distribution.}
Recent work that proposed to use cheap proxy models, such as image classifiers,  to identify  an  approximate  set  of  data points satisfying a query~\cite{DBLP:journals/pvldb/KangGBHZ20}, is by far the  closest to our work, albeit they require a budget. \eat{. As mentioned in \S~\ref{sec:intro}, they studied approximate answering of recall- and precision-target queries, with a user-specified cost budget. Their setting is restricted to classification, whereas we consider more general similarity queries. The limitations of using a user-specified cost budget were already discussed in \S~\ref{sec:intro}. Our experiments clearly demonstrate the superiority of our approach.

}

\dujian{\section{Conclusion and Discussion}}
\label{sec:discussion}

We formalize and solve precision-target and recall-target queries, two paradigms that are well-suited for querying the results of ML predictions.
We propose two assumptions and  develop four algorithms.
Our extensive experiments on five real-world datasets show that our approach enjoys statistical guarantees with a small cost and a good complementary rate, i.e., a good balance between recall and precision rates. 
\eat{
We are currently investigating how to express precision-target and recall-target queries as a foundational extension of the relational algebra. This would allow us to answer queries within the relational query runtime without exfiltrating data into memory. This will also open new directions in DB/ML cross-optimizations such as pushing predicates and model inlining. Existing frameworks such as Raven and its runtime environment ONNX~\cite{DBLP:conf/cidr/KaranasosIPSPPX20} appear as good starting points.}

\dujian{Our framework can be extended to optimize a query workload using metric  properties like triangle inequality \cite{10.1145/3448016.3457303}. 
Consider the objects $\{q_1, q_2, x\}$. Suppose that we first choose $q_1$ as the query object, and compute the proxy distances $dist^\proxy(q_1, q_2)$ and $dist^\proxy(q_1, x)$ to find answers using our approaches. Next, when we choose $q_2$ as the query object, by leveraging triangle inequality, we can lower bound $dist^\proxy(q_2,x)$ as $dist^\proxy(q_2,x) \geq |dist^\proxy(q_1,q_2) - dist^\proxy(q_1,x)|$. If this bound is high, we can safely avoid applying the probe to $x$ for query $q_2$.
We can extend our framework to multiple proxies at different accuracy and cost levels. This represents real-world scenarios where proxies are derived from huge neural models by activating specific subnetworks \cite{DBLP:conf/icml/BolukbasiWDS17}. It is clear that, with multiple proxy models, the search space for our optimization problem will exponentially increase. 
Secondly, by introducing proxies with various cost and accuracy levels, optimizing efficiency would go beyond simply counting oracle calls, and would yield a linear programming problem. We are currently exploring   possible solutions. }

\clearpage

\normalem 
 
\bibliographystyle{ACM-Reference-Format}
\bibliography{sample-base}


\begin{thebibliography}{51}


\ifx \showCODEN    \undefined \def \showCODEN     #1{\unskip}     \fi
\ifx \showDOI      \undefined \def \showDOI       #1{#1}\fi
\ifx \showISBNx    \undefined \def \showISBNx     #1{\unskip}     \fi
\ifx \showISBNxiii \undefined \def \showISBNxiii  #1{\unskip}     \fi
\ifx \showISSN     \undefined \def \showISSN      #1{\unskip}     \fi
\ifx \showLCCN     \undefined \def \showLCCN      #1{\unskip}     \fi
\ifx \shownote     \undefined \def \shownote      #1{#1}          \fi
\ifx \showarticletitle \undefined \def \showarticletitle #1{#1}   \fi
\ifx \showURL      \undefined \def \showURL       {\relax}        \fi
\providecommand\bibfield[2]{#2}
\providecommand\bibinfo[2]{#2}
\providecommand\natexlab[1]{#1}
\providecommand\showeprint[2][]{arXiv:#2}

\bibitem[\protect\citeauthoryear{??}{fra}{[n.d.]}]%
        {franklin1959preventive}
 \bibinfo{year}{[n.d.]}\natexlab{}.
\newblock  (\bibinfo{year}{[n.\,d.]}).
\newblock


\bibitem[\protect\citeauthoryear{Ahmed, Amizadeh, Bilenko, Carr, Chin, Dekel,
  Dupr{\'{e}}, Eksarevskiy, Filipi, Finley, Goswami, Hoover, Inglis,
  Interlandi, Kazmi, Krivosheev, Luferenko, Matantsev, Matusevych, Moradi,
  Nazirov, Ormont, Oshri, Pagnoni, Parmar, Roy, Siddiqui, Weimer, Zahirazami,
  and Zhu}{Ahmed et~al\mbox{.}}{2019}]%
        {DBLP:conf/kdd/AhmedABCCDDEFFG19}
\bibfield{author}{\bibinfo{person}{Zeeshan Ahmed}, \bibinfo{person}{Saeed
  Amizadeh}, \bibinfo{person}{Mikhail Bilenko}, \bibinfo{person}{Rogan Carr},
  \bibinfo{person}{Wei{-}Sheng Chin}, \bibinfo{person}{Yael Dekel},
  \bibinfo{person}{Xavier Dupr{\'{e}}}, \bibinfo{person}{Vadim Eksarevskiy},
  \bibinfo{person}{Senja Filipi}, \bibinfo{person}{Tom Finley},
  \bibinfo{person}{Abhishek Goswami}, \bibinfo{person}{Monte Hoover},
  \bibinfo{person}{Scott Inglis}, \bibinfo{person}{Matteo Interlandi},
  \bibinfo{person}{Najeeb Kazmi}, \bibinfo{person}{Gleb Krivosheev},
  \bibinfo{person}{Pete Luferenko}, \bibinfo{person}{Ivan Matantsev},
  \bibinfo{person}{Sergiy Matusevych}, \bibinfo{person}{Shahab Moradi},
  \bibinfo{person}{Gani Nazirov}, \bibinfo{person}{Justin Ormont},
  \bibinfo{person}{Gal Oshri}, \bibinfo{person}{Artidoro Pagnoni},
  \bibinfo{person}{Jignesh Parmar}, \bibinfo{person}{Prabhat Roy},
  \bibinfo{person}{Mohammad~Zeeshan Siddiqui}, \bibinfo{person}{Markus Weimer},
  \bibinfo{person}{Shauheen Zahirazami}, {and} \bibinfo{person}{Yiwen Zhu}.}
  \bibinfo{year}{2019}\natexlab{}.
\newblock \showarticletitle{Machine Learning at Microsoft with {ML.NET}}. In
  \bibinfo{booktitle}{\emph{Proceedings of the 25th {ACM} {SIGKDD}
  International Conference on Knowledge Discovery {\&} Data Mining, {KDD} 2019,
  Anchorage, AK, USA, August 4-8, 2019}},
  \bibfield{editor}{\bibinfo{person}{Ankur Teredesai}, \bibinfo{person}{Vipin
  Kumar}, \bibinfo{person}{Ying Li}, \bibinfo{person}{R{\'{o}}mer Rosales},
  \bibinfo{person}{Evimaria Terzi}, {and} \bibinfo{person}{George Karypis}}
  (Eds.). \bibinfo{publisher}{{ACM}}, \bibinfo{pages}{2448--2458}.
\newblock


\bibitem[\protect\citeauthoryear{Alodadi and Janeja}{Alodadi and
  Janeja}{2015}]%
        {7349760}
\bibfield{author}{\bibinfo{person}{Mohammad Alodadi} {and}
  \bibinfo{person}{Vandana~P. Janeja}.} \bibinfo{year}{2015}\natexlab{}.
\newblock \showarticletitle{Similarity in Patient Support Forums Using TF-IDF
  and Cosine Similarity Metrics}. In \bibinfo{booktitle}{\emph{2015
  International Conference on Healthcare Informatics}}.
  \bibinfo{pages}{521--522}.
\newblock
\urldef\tempurl%
\url{https://doi.org/10.1109/ICHI.2015.99}
\showDOI{\tempurl}


\bibitem[\protect\citeauthoryear{Anderson, Cafarella, Ros, and
  Wenisch}{Anderson et~al\mbox{.}}{2019}]%
        {DBLP:conf/icde/AndersonCRW19}
\bibfield{author}{\bibinfo{person}{Michael~R. Anderson},
  \bibinfo{person}{Michael~J. Cafarella}, \bibinfo{person}{German Ros}, {and}
  \bibinfo{person}{Thomas~F. Wenisch}.} \bibinfo{year}{2019}\natexlab{}.
\newblock \showarticletitle{Physical Representation-Based Predicate
  Optimization for a Visual Analytics Database}. In
  \bibinfo{booktitle}{\emph{35th {IEEE} International Conference on Data
  Engineering, {ICDE} 2019, Macao, China, April 8-11, 2019}}.
  \bibinfo{pages}{1466--1477}.
\newblock


\bibitem[\protect\citeauthoryear{Augustine, Shetiya, Esfandiari, Basu~Roy, and
  Das}{Augustine et~al\mbox{.}}{2021}]%
        {10.1145/3448016.3457303}
\bibfield{author}{\bibinfo{person}{Jees Augustine}, \bibinfo{person}{Suraj
  Shetiya}, \bibinfo{person}{Mohammadreza Esfandiari}, \bibinfo{person}{Senjuti
  Basu~Roy}, {and} \bibinfo{person}{Gautam Das}.}
  \bibinfo{year}{2021}\natexlab{}.
\newblock \showarticletitle{A Generalized Approach for Reducing Expensive
  Distance Calls for A Broad Class of Proximity Problems}. In
  \bibinfo{booktitle}{\emph{Proceedings of the 2021 International Conference on
  Management of Data}} (Virtual Event, China) \emph{(\bibinfo{series}{SIGMOD
  '21})}. \bibinfo{publisher}{Association for Computing Machinery},
  \bibinfo{address}{New York, NY, USA}, \bibinfo{pages}{142–154}.
\newblock
\showISBNx{9781450383431}
\urldef\tempurl%
\url{https://doi.org/10.1145/3448016.3457303}
\showDOI{\tempurl}


\bibitem[\protect\citeauthoryear{Bentley}{Bentley}{1975}]%
        {bentley1975frnn}
\bibfield{author}{\bibinfo{person}{Jon~L Bentley}.}
  \bibinfo{year}{1975}\natexlab{}.
\newblock \bibinfo{booktitle}{\emph{A Survey of Techniques for Fixed Radius
  near Neighbor Searching.}}
\newblock \bibinfo{type}{{T}echnical {R}eport}. \bibinfo{address}{Stanford, CA,
  USA}.
\newblock


\bibitem[\protect\citeauthoryear{Bentley, Stanat, and Williams}{Bentley
  et~al\mbox{.}}{1977}]%
        {BENTLEY1977209}
\bibfield{author}{\bibinfo{person}{Jon~L. Bentley}, \bibinfo{person}{Donald~F.
  Stanat}, {and} \bibinfo{person}{E.Hollins Williams}.}
  \bibinfo{year}{1977}\natexlab{}.
\newblock \showarticletitle{The complexity of finding fixed-radius near
  neighbors}.
\newblock \bibinfo{journal}{\emph{Inform. Process. Lett.}} \bibinfo{volume}{6},
  \bibinfo{number}{6} (\bibinfo{year}{1977}), \bibinfo{pages}{209--212}.
\newblock
\showISSN{0020-0190}
\urldef\tempurl%
\url{https://doi.org/10.1016/0020-0190(77)90070-9}
\showDOI{\tempurl}


\bibitem[\protect\citeauthoryear{Biscarri, Zhao, and Brunner}{Biscarri
  et~al\mbox{.}}{2018}]%
        {osti_1548776}
\bibfield{author}{\bibinfo{person}{William Biscarri},
  \bibinfo{person}{Sihai~Dave Zhao}, {and} \bibinfo{person}{Robert~J.
  Brunner}.} \bibinfo{year}{2018}\natexlab{}.
\newblock \showarticletitle{A simple and fast method for computing the Poisson
  binomial distribution function}.
\newblock \bibinfo{journal}{\emph{Computational Statistics and Data Analysis
  (Print)}}  \bibinfo{volume}{122} (\bibinfo{date}{6} \bibinfo{year}{2018}).
\newblock
\urldef\tempurl%
\url{https://doi.org/10.1016/j.csda.2018.01.007}
\showDOI{\tempurl}


\bibitem[\protect\citeauthoryear{Bolukbasi, Wang, Dekel, and
  Saligrama}{Bolukbasi et~al\mbox{.}}{2017}]%
        {DBLP:conf/icml/BolukbasiWDS17}
\bibfield{author}{\bibinfo{person}{Tolga Bolukbasi}, \bibinfo{person}{Joseph
  Wang}, \bibinfo{person}{Ofer Dekel}, {and} \bibinfo{person}{Venkatesh
  Saligrama}.} \bibinfo{year}{2017}\natexlab{}.
\newblock \showarticletitle{Adaptive Neural Networks for Efficient Inference}.
  In \bibinfo{booktitle}{\emph{Proceedings of the 34th International Conference
  on Machine Learning, {ICML} 2017, Sydney, NSW, Australia, 6-11 August 2017}}
  \emph{(\bibinfo{series}{Proceedings of Machine Learning Research})},
  \bibfield{editor}{\bibinfo{person}{Doina Precup} {and}
  \bibinfo{person}{Yee~Whye Teh}} (Eds.), Vol.~\bibinfo{volume}{70}.
  \bibinfo{publisher}{{PMLR}}, \bibinfo{pages}{527--536}.
\newblock


\bibitem[\protect\citeauthoryear{Brull, Ghali, and Quan}{Brull
  et~al\mbox{.}}{1999}]%
        {brull1999missed}
\bibfield{author}{\bibinfo{person}{R Brull}, \bibinfo{person}{W~A Ghali}, {and}
  \bibinfo{person}{H Quan}.} \bibinfo{year}{1999}\natexlab{}.
\newblock \showarticletitle{Missed opportunities for prevention in general
  internal medicine.}
\newblock \bibinfo{journal}{\emph{CMAJ}} \bibinfo{volume}{160},
  \bibinfo{number}{8} (\bibinfo{date}{Apr} \bibinfo{year}{1999}),
  \bibinfo{pages}{1137--1140}.
\newblock
\showISSN{0820-3946 (Print); 1488-2329 (Electronic); 0820-3946 (Linking)}


\bibitem[\protect\citeauthoryear{Canel, Kim, Zhou, Li, Lim, Andersen, Kaminsky,
  and Dulloor}{Canel et~al\mbox{.}}{2019}]%
        {jackson_dataset}
\bibfield{author}{\bibinfo{person}{Christopher Canel}, \bibinfo{person}{Thomas
  Kim}, \bibinfo{person}{Giulio Zhou}, \bibinfo{person}{Conglong Li},
  \bibinfo{person}{Hyeontaek Lim}, \bibinfo{person}{David~G. Andersen},
  \bibinfo{person}{Michael Kaminsky}, {and} \bibinfo{person}{Subramanya~R.
  Dulloor}.} \bibinfo{year}{2019}\natexlab{}.
\newblock \showarticletitle{Scaling Video Analytics on Constrained Edge Nodes}.
\newblock \bibinfo{journal}{\emph{CoRR}}  \bibinfo{volume}{abs/1905.13536}
  (\bibinfo{year}{2019}).
\newblock
\showeprint[arXiv]{1905.13536}
\urldef\tempurl%
\url{http://arxiv.org/abs/1905.13536}
\showURL{%
\tempurl}


\bibitem[\protect\citeauthoryear{Cao, Long, Wang, Zhu, and Wen}{Cao
  et~al\mbox{.}}{2016}]%
        {yue2016deep}
\bibfield{author}{\bibinfo{person}{Yue Cao}, \bibinfo{person}{Mingsheng Long},
  \bibinfo{person}{Jianmin Wang}, \bibinfo{person}{Han Zhu}, {and}
  \bibinfo{person}{Qingfu Wen}.} \bibinfo{year}{2016}\natexlab{}.
\newblock \showarticletitle{Deep Quantization Network for Efficient Image
  Retrieval}. In \bibinfo{booktitle}{\emph{Proceedings of the Thirtieth AAAI
  Conference on Artificial Intelligence}} (Phoenix, Arizona)
  \emph{(\bibinfo{series}{AAAI'16})}. \bibinfo{publisher}{AAAI Press},
  \bibinfo{pages}{3457–3463}.
\newblock


\bibitem[\protect\citeauthoryear{Chen, Moreau, Jiang, Zheng, Yan, Shen, Cowan,
  Wang, Hu, Ceze, Guestrin, and Krishnamurthy}{Chen et~al\mbox{.}}{2018}]%
        {DBLP:conf/osdi/ChenMJZYSCWHCGK18}
\bibfield{author}{\bibinfo{person}{Tianqi Chen}, \bibinfo{person}{Thierry
  Moreau}, \bibinfo{person}{Ziheng Jiang}, \bibinfo{person}{Lianmin Zheng},
  \bibinfo{person}{Eddie~Q. Yan}, \bibinfo{person}{Haichen Shen},
  \bibinfo{person}{Meghan Cowan}, \bibinfo{person}{Leyuan Wang},
  \bibinfo{person}{Yuwei Hu}, \bibinfo{person}{Luis Ceze},
  \bibinfo{person}{Carlos Guestrin}, {and} \bibinfo{person}{Arvind
  Krishnamurthy}.} \bibinfo{year}{2018}\natexlab{}.
\newblock \showarticletitle{{TVM:} An Automated End-to-End Optimizing Compiler
  for Deep Learning}. In \bibinfo{booktitle}{\emph{13th {USENIX} Symposium on
  Operating Systems Design and Implementation, {OSDI} 2018, Carlsbad, CA, USA,
  October 8-10, 2018}}, \bibfield{editor}{\bibinfo{person}{Andrea~C.
  Arpaci{-}Dusseau} {and} \bibinfo{person}{Geoff Voelker}} (Eds.).
  \bibinfo{publisher}{{USENIX} Association}, \bibinfo{pages}{578--594}.
\newblock


\bibitem[\protect\citeauthoryear{Chen, Liu, Wang, Bakker, Georgiou, Fieguth,
  Liu, and Lew}{Chen et~al\mbox{.}}{2021}]%
        {chen2021deep}
\bibfield{author}{\bibinfo{person}{Wei Chen}, \bibinfo{person}{Yu Liu},
  \bibinfo{person}{Weiping Wang}, \bibinfo{person}{Erwin Bakker},
  \bibinfo{person}{Theodoros Georgiou}, \bibinfo{person}{Paul Fieguth},
  \bibinfo{person}{Li Liu}, {and} \bibinfo{person}{Michael~S. Lew}.}
  \bibinfo{year}{2021}\natexlab{}.
\newblock \bibinfo{title}{Deep Image Retrieval: A Survey}.
\newblock
\newblock
\showeprint[arxiv]{2101.11282}~[cs.CV]


\bibitem[\protect\citeauthoryear{Chen, Wei, Wang, and Guo}{Chen
  et~al\mbox{.}}{2019}]%
        {chen2019mlgcn}
\bibfield{author}{\bibinfo{person}{Zhao-Min Chen}, \bibinfo{person}{Xiu-Shen
  Wei}, \bibinfo{person}{Peng Wang}, {and} \bibinfo{person}{Yanwen Guo}.}
  \bibinfo{year}{2019}\natexlab{}.
\newblock \showarticletitle{Multi-Label Image Recognition With Graph
  Convolutional Networks}. In \bibinfo{booktitle}{\emph{2019 IEEE/CVF
  Conference on Computer Vision and Pattern Recognition (CVPR)}}.
  \bibinfo{pages}{5172--5181}.
\newblock
\urldef\tempurl%
\url{https://doi.org/10.1109/CVPR.2019.00532}
\showDOI{\tempurl}


\bibitem[\protect\citeauthoryear{Choi, Bahadori, Schuetz, Stewart, and
  Sun}{Choi et~al\mbox{.}}{2016}]%
        {choi2016doctor}
\bibfield{author}{\bibinfo{person}{Edward Choi}, \bibinfo{person}{Mohammad~Taha
  Bahadori}, \bibinfo{person}{Andy Schuetz}, \bibinfo{person}{Walter~F.
  Stewart}, {and} \bibinfo{person}{Jimeng Sun}.}
  \bibinfo{year}{2016}\natexlab{}.
\newblock \bibinfo{title}{Doctor AI: Predicting Clinical Events via Recurrent
  Neural Networks}.
\newblock
\newblock
\showeprint[arxiv]{1511.05942}~[cs.LG]


\bibitem[\protect\citeauthoryear{Crankshaw, Wang, Zhou, Franklin, Gonzalez, and
  Stoica}{Crankshaw et~al\mbox{.}}{2017}]%
        {DBLP:conf/nsdi/CrankshawWZFGS17}
\bibfield{author}{\bibinfo{person}{Daniel Crankshaw}, \bibinfo{person}{Xin
  Wang}, \bibinfo{person}{Giulio Zhou}, \bibinfo{person}{Michael~J. Franklin},
  \bibinfo{person}{Joseph~E. Gonzalez}, {and} \bibinfo{person}{Ion Stoica}.}
  \bibinfo{year}{2017}\natexlab{}.
\newblock \showarticletitle{Clipper: {A} Low-Latency Online Prediction Serving
  System}. In \bibinfo{booktitle}{\emph{14th {USENIX} Symposium on Networked
  Systems Design and Implementation, {NSDI} 2017, Boston, MA, USA, March 27-29,
  2017}}, \bibfield{editor}{\bibinfo{person}{Aditya Akella} {and}
  \bibinfo{person}{Jon Howell}} (Eds.). \bibinfo{publisher}{{USENIX}
  Association}, \bibinfo{pages}{613--627}.
\newblock


\bibitem[\protect\citeauthoryear{Deniziak and Michno}{Deniziak and
  Michno}{2016}]%
        {deniziak2016content}
\bibfield{author}{\bibinfo{person}{Stanislaw Deniziak} {and}
  \bibinfo{person}{Tomasz Michno}.} \bibinfo{year}{2016}\natexlab{}.
\newblock \showarticletitle{Content based image retrieval using query by
  approximate shape}. In \bibinfo{booktitle}{\emph{2016 Federated Conference on
  Computer Science and Information Systems (FedCSIS)}}.
  \bibinfo{pages}{807--816}.
\newblock


\bibitem[\protect\citeauthoryear{Ding, Amer-Yahia, and Lakshmanan}{Ding
  et~al\mbox{.}}{2022}]%
        {aquaprofull}
\bibfield{author}{\bibinfo{person}{Dujian Ding}, \bibinfo{person}{Sihem
  Amer-Yahia}, {and} \bibinfo{person}{Laks~VS Lakshmanan}.}
  \bibinfo{year}{2022}\natexlab{}.
\newblock \bibinfo{booktitle}{\emph{On Efficient Approximate Queries over
  Machine Learning Models (Full Version)}}.
\newblock
\urldef\tempurl%
\url{https://sites.google.com/view/dujian/publications}
\showURL{%
\tempurl}


\bibitem[\protect\citeauthoryear{Everingham, Van~Gool, Williams, Winn, and
  Zisserman}{Everingham et~al\mbox{.}}{2010}]%
        {voc2007}
\bibfield{author}{\bibinfo{person}{Mark Everingham}, \bibinfo{person}{Luc
  Van~Gool}, \bibinfo{person}{Christopher K.~I. Williams},
  \bibinfo{person}{John Winn}, {and} \bibinfo{person}{Andrew Zisserman}.}
  \bibinfo{year}{2010}\natexlab{}.
\newblock \showarticletitle{The Pascal Visual Object Classes (VOC) Challenge}.
\newblock \bibinfo{journal}{\emph{International Journal of Computer Vision}}
  \bibinfo{volume}{88}, \bibinfo{number}{2} (\bibinfo{year}{2010}),
  \bibinfo{pages}{303--338}.
\newblock


\bibitem[\protect\citeauthoryear{Gao, Xu, Agarwal, and Yang}{Gao
  et~al\mbox{.}}{2021}]%
        {DBLP:conf/sigmod/GaoXAY21}
\bibfield{author}{\bibinfo{person}{Junyang Gao}, \bibinfo{person}{Yifan Xu},
  \bibinfo{person}{Pankaj~K. Agarwal}, {and} \bibinfo{person}{Jun Yang}.}
  \bibinfo{year}{2021}\natexlab{}.
\newblock \showarticletitle{Efficiently Answering Durability Prediction
  Queries}. In \bibinfo{booktitle}{\emph{{SIGMOD} '21: International Conference
  on Management of Data, Virtual Event, China, June 20-25, 2021}}.
  \bibinfo{pages}{591--604}.
\newblock


\bibitem[\protect\citeauthoryear{Grimmett}{Grimmett}{1986}]%
        {Grimmett1986-GRIPAI}
\bibfield{author}{\bibinfo{person}{Geoffrey~R. Grimmett}.}
  \bibinfo{year}{1986}\natexlab{}.
\newblock \bibinfo{booktitle}{\emph{Probability: An Introduction}}.
\newblock \bibinfo{publisher}{Oxford University Press}.
\newblock


\bibitem[\protect\citeauthoryear{He, Gkioxari, Doll{\'{a}}r, and Girshick}{He
  et~al\mbox{.}}{2017}]%
        {he2017mask}
\bibfield{author}{\bibinfo{person}{Kaiming He}, \bibinfo{person}{Georgia
  Gkioxari}, \bibinfo{person}{Piotr Doll{\'{a}}r}, {and}
  \bibinfo{person}{Ross~B. Girshick}.} \bibinfo{year}{2017}\natexlab{}.
\newblock \showarticletitle{Mask {R-CNN}}.
\newblock \bibinfo{journal}{\emph{CoRR}}  \bibinfo{volume}{abs/1703.06870}
  (\bibinfo{year}{2017}).
\newblock
\showeprint[arXiv]{1703.06870}
\urldef\tempurl%
\url{http://arxiv.org/abs/1703.06870}
\showURL{%
\tempurl}


\bibitem[\protect\citeauthoryear{He, Zhang, Ren, and Sun}{He
  et~al\mbox{.}}{2015}]%
        {he2015res}
\bibfield{author}{\bibinfo{person}{Kaiming He}, \bibinfo{person}{Xiangyu
  Zhang}, \bibinfo{person}{Shaoqing Ren}, {and} \bibinfo{person}{Jian Sun}.}
  \bibinfo{year}{2015}\natexlab{}.
\newblock \showarticletitle{Deep Residual Learning for Image Recognition}.
\newblock \bibinfo{journal}{\emph{CoRR}}  \bibinfo{volume}{abs/1512.03385}
  (\bibinfo{year}{2015}).
\newblock
\showeprint[arXiv]{1512.03385}
\urldef\tempurl%
\url{http://arxiv.org/abs/1512.03385}
\showURL{%
\tempurl}


\bibitem[\protect\citeauthoryear{Hsieh, Ananthanarayanan, Bodik, Bahl,
  Philipose, Gibbons, and Mutlu}{Hsieh et~al\mbox{.}}{2018}]%
        {hsieh2018focus}
\bibfield{author}{\bibinfo{person}{Kevin Hsieh}, \bibinfo{person}{Ganesh
  Ananthanarayanan}, \bibinfo{person}{Peter Bodik}, \bibinfo{person}{Paramvir
  Bahl}, \bibinfo{person}{Matthai Philipose}, \bibinfo{person}{Phillip~B.
  Gibbons}, {and} \bibinfo{person}{Onur Mutlu}.}
  \bibinfo{year}{2018}\natexlab{}.
\newblock \bibinfo{title}{Focus: Querying Large Video Datasets with Low Latency
  and Low Cost}.
\newblock
\newblock
\showeprint[arxiv]{1801.03493}~[cs.DB]


\bibitem[\protect\citeauthoryear{Johnson, Pollard, Shen, Lehman, Feng,
  Ghassemi, Moody, Szolovits, Anthony~Celi, and Mark}{Johnson
  et~al\mbox{.}}{2016}]%
        {mimic_iii}
\bibfield{author}{\bibinfo{person}{Alistair E.~W. Johnson},
  \bibinfo{person}{Tom~J. Pollard}, \bibinfo{person}{Lu Shen},
  \bibinfo{person}{Li-wei~H. Lehman}, \bibinfo{person}{Mengling Feng},
  \bibinfo{person}{Mohammad Ghassemi}, \bibinfo{person}{Benjamin Moody},
  \bibinfo{person}{Peter Szolovits}, \bibinfo{person}{Leo Anthony~Celi}, {and}
  \bibinfo{person}{Roger~G. Mark}.} \bibinfo{year}{2016}\natexlab{}.
\newblock \showarticletitle{MIMIC-III, a freely accessible critical care
  database}.
\newblock \bibinfo{journal}{\emph{Scientific Data}} \bibinfo{volume}{3},
  \bibinfo{number}{1} (\bibinfo{year}{2016}), \bibinfo{pages}{160035}.
\newblock


\bibitem[\protect\citeauthoryear{Jr., Gutierrez, Spadon, Brandoli, and
  Amer{-}Yahia}{Jr. et~al\mbox{.}}{2021}]%
        {DBLP:journals/isci/RodriguesGSBA21}
\bibfield{author}{\bibinfo{person}{Jos{\'{e}} F.~Rodrigues Jr.},
  \bibinfo{person}{Marco~Antonio Gutierrez}, \bibinfo{person}{Gabriel Spadon},
  \bibinfo{person}{Bruno Brandoli}, {and} \bibinfo{person}{Sihem
  Amer{-}Yahia}.} \bibinfo{year}{2021}\natexlab{}.
\newblock \showarticletitle{LIG-Doctor: Efficient patient trajectory prediction
  using bidirectional minimal gated-recurrent networks}.
\newblock \bibinfo{journal}{\emph{Inf. Sci.}}  \bibinfo{volume}{545}
  (\bibinfo{year}{2021}), \bibinfo{pages}{813--827}.
\newblock


\bibitem[\protect\citeauthoryear{Kang, Emmons, Abuzaid, Bailis, and
  Zaharia}{Kang et~al\mbox{.}}{2017}]%
        {DBLP:journals/pvldb/KangEABZ17}
\bibfield{author}{\bibinfo{person}{Daniel Kang}, \bibinfo{person}{John Emmons},
  \bibinfo{person}{Firas Abuzaid}, \bibinfo{person}{Peter Bailis}, {and}
  \bibinfo{person}{Matei Zaharia}.} \bibinfo{year}{2017}\natexlab{}.
\newblock \showarticletitle{NoScope: Optimizing Deep CNN-Based Queries over
  Video Streams at Scale}.
\newblock \bibinfo{journal}{\emph{Proc. {VLDB} Endow.}} \bibinfo{volume}{10},
  \bibinfo{number}{11} (\bibinfo{year}{2017}), \bibinfo{pages}{1586--1597}.
\newblock


\bibitem[\protect\citeauthoryear{Kang, Gan, Bailis, Hashimoto, and
  Zaharia}{Kang et~al\mbox{.}}{2020}]%
        {DBLP:journals/pvldb/KangGBHZ20}
\bibfield{author}{\bibinfo{person}{Daniel Kang}, \bibinfo{person}{Edward Gan},
  \bibinfo{person}{Peter Bailis}, \bibinfo{person}{Tatsunori Hashimoto}, {and}
  \bibinfo{person}{Matei Zaharia}.} \bibinfo{year}{2020}\natexlab{}.
\newblock \showarticletitle{Approximate Selection with Guarantees using
  Proxies}.
\newblock \bibinfo{journal}{\emph{Proc. {VLDB} Endow.}} \bibinfo{volume}{13},
  \bibinfo{number}{11} (\bibinfo{year}{2020}), \bibinfo{pages}{1990--2003}.
\newblock


\bibitem[\protect\citeauthoryear{Karanasos, Interlandi, Psallidas, Sen, Park,
  Popivanov, Xin, Nakandala, Krishnan, Weimer, Yu, Ramakrishnan, and
  Curino}{Karanasos et~al\mbox{.}}{2020}]%
        {DBLP:conf/cidr/KaranasosIPSPPX20}
\bibfield{author}{\bibinfo{person}{Konstantinos Karanasos},
  \bibinfo{person}{Matteo Interlandi}, \bibinfo{person}{Fotis Psallidas},
  \bibinfo{person}{Rathijit Sen}, \bibinfo{person}{Kwanghyun Park},
  \bibinfo{person}{Ivan Popivanov}, \bibinfo{person}{Doris Xin},
  \bibinfo{person}{Supun Nakandala}, \bibinfo{person}{Subru Krishnan},
  \bibinfo{person}{Markus Weimer}, \bibinfo{person}{Yuan Yu},
  \bibinfo{person}{Raghu Ramakrishnan}, {and} \bibinfo{person}{Carlo Curino}.}
  \bibinfo{year}{2020}\natexlab{}.
\newblock \showarticletitle{Extending Relational Query Processing with {ML}
  Inference}. In \bibinfo{booktitle}{\emph{10th Conference on Innovative Data
  Systems Research, {CIDR} 2020, Amsterdam, The Netherlands, January 12-15,
  2020, Online Proceedings}}. \bibinfo{publisher}{www.cidrdb.org}.
\newblock


\bibitem[\protect\citeauthoryear{Lahitani, Permanasari, and Setiawan}{Lahitani
  et~al\mbox{.}}{2016}]%
        {7577578}
\bibfield{author}{\bibinfo{person}{Alfirna~Rizqi Lahitani},
  \bibinfo{person}{Adhistya~Erna Permanasari}, {and}
  \bibinfo{person}{Noor~Akhmad Setiawan}.} \bibinfo{year}{2016}\natexlab{}.
\newblock \showarticletitle{Cosine similarity to determine similarity measure:
  Study case in online essay assessment}. In \bibinfo{booktitle}{\emph{2016 4th
  International Conference on Cyber and IT Service Management}}.
  \bibinfo{pages}{1--6}.
\newblock
\urldef\tempurl%
\url{https://doi.org/10.1109/CITSM.2016.7577578}
\showDOI{\tempurl}


\bibitem[\protect\citeauthoryear{Lai, Han, Liu, Zhang, Lo, and Kao}{Lai
  et~al\mbox{.}}{2021}]%
        {10.1145/3448016.3452786}
\bibfield{author}{\bibinfo{person}{Ziliang Lai}, \bibinfo{person}{Chenxia Han},
  \bibinfo{person}{Chris Liu}, \bibinfo{person}{Pengfei Zhang},
  \bibinfo{person}{Eric Lo}, {and} \bibinfo{person}{Ben Kao}.}
  \bibinfo{year}{2021}\natexlab{}.
\newblock \showarticletitle{Top-K Deep Video Analytics: A Probabilistic
  Approach}. In \bibinfo{booktitle}{\emph{Proceedings of the 2021 International
  Conference on Management of Data}} (Virtual Event, China)
  \emph{(\bibinfo{series}{SIGMOD '21})}. \bibinfo{publisher}{Association for
  Computing Machinery}, \bibinfo{address}{New York, NY, USA},
  \bibinfo{pages}{1037–1050}.
\newblock
\showISBNx{9781450383431}
\urldef\tempurl%
\url{https://doi.org/10.1145/3448016.3452786}
\showDOI{\tempurl}


\bibitem[\protect\citeauthoryear{Li and Li}{Li and Li}{2018}]%
        {DBLP:journals/dase/LiL18}
\bibfield{author}{\bibinfo{person}{Kaiyu Li} {and} \bibinfo{person}{Guoliang
  Li}.} \bibinfo{year}{2018}\natexlab{}.
\newblock \showarticletitle{Approximate Query Processing: What is New and Where
  to Go? - {A} Survey on Approximate Query Processing}.
\newblock \bibinfo{journal}{\emph{Data Sci. Eng.}} \bibinfo{volume}{3},
  \bibinfo{number}{4} (\bibinfo{year}{2018}), \bibinfo{pages}{379--397}.
\newblock


\bibitem[\protect\citeauthoryear{Li, Rao, Solares, Hassaine, Ramakrishnan,
  Canoy, Zhu, Rahimi, and Salimi-Khorshidi}{Li et~al\mbox{.}}{2020}]%
        {li2020behrt}
\bibfield{author}{\bibinfo{person}{Yikuan Li}, \bibinfo{person}{Shishir Rao},
  \bibinfo{person}{Jos{\'e}Roberto~Ayala Solares}, \bibinfo{person}{Abdelaali
  Hassaine}, \bibinfo{person}{Rema Ramakrishnan}, \bibinfo{person}{Dexter
  Canoy}, \bibinfo{person}{Yajie Zhu}, \bibinfo{person}{Kazem Rahimi}, {and}
  \bibinfo{person}{Gholamreza Salimi-Khorshidi}.}
  \bibinfo{year}{2020}\natexlab{}.
\newblock \showarticletitle{BEHRT: Transformer for Electronic Health Records}.
\newblock \bibinfo{journal}{\emph{Scientific Reports}} \bibinfo{volume}{10},
  \bibinfo{number}{1} (\bibinfo{year}{2020}), \bibinfo{pages}{7155}.
\newblock


\bibitem[\protect\citeauthoryear{Lin, Maire, Belongie, Hays, Perona, Ramanan,
  Doll{\'a}r, and Zitnick}{Lin et~al\mbox{.}}{2014}]%
        {coco2014}
\bibfield{author}{\bibinfo{person}{Tsung-Yi Lin}, \bibinfo{person}{Michael
  Maire}, \bibinfo{person}{Serge Belongie}, \bibinfo{person}{James Hays},
  \bibinfo{person}{Pietro Perona}, \bibinfo{person}{Deva Ramanan},
  \bibinfo{person}{Piotr Doll{\'a}r}, {and} \bibinfo{person}{C.~Lawrence
  Zitnick}.} \bibinfo{year}{2014}\natexlab{}.
\newblock \showarticletitle{Microsoft COCO: Common Objects in Context}. In
  \bibinfo{booktitle}{\emph{Computer Vision -- ECCV 2014}},
  \bibfield{editor}{\bibinfo{person}{David Fleet}, \bibinfo{person}{Tomas
  Pajdla}, \bibinfo{person}{Bernt Schiele}, {and} \bibinfo{person}{Tinne
  Tuytelaars}} (Eds.). \bibinfo{publisher}{Springer International Publishing},
  \bibinfo{address}{Cham}, \bibinfo{pages}{740--755}.
\newblock
\showISBNx{978-3-319-10602-1}


\bibitem[\protect\citeauthoryear{Love}{Love}{1994}]%
        {love1994cancer}
\bibfield{author}{\bibinfo{person}{R~R Love}.} \bibinfo{year}{1994}\natexlab{}.
\newblock \showarticletitle{Cancer prevention through health promotion.
  Defining the role of physicians in public health.}
\newblock \bibinfo{journal}{\emph{Cancer}} \bibinfo{volume}{74},
  \bibinfo{number}{4 Suppl} (\bibinfo{date}{Aug} \bibinfo{year}{1994}),
  \bibinfo{pages}{1418--1422}.
\newblock
\showISSN{0008-543X (Print); 0008-543X (Linking)}
\urldef\tempurl%
\url{https://doi.org/10.1002/1097-0142(19940815)74:4+<1418::aid-cncr2820741604>3.0.co;2-5}
\showDOI{\tempurl}


\bibitem[\protect\citeauthoryear{Lu, Chowdhery, Kandula, and Chaudhuri}{Lu
  et~al\mbox{.}}{2018}]%
        {DBLP:conf/sigmod/LuCKC18}
\bibfield{author}{\bibinfo{person}{Yao Lu}, \bibinfo{person}{Aakanksha
  Chowdhery}, \bibinfo{person}{Srikanth Kandula}, {and}
  \bibinfo{person}{Surajit Chaudhuri}.} \bibinfo{year}{2018}\natexlab{}.
\newblock \showarticletitle{Accelerating Machine Learning Inference with
  Probabilistic Predicates}. In \bibinfo{booktitle}{\emph{Proceedings of the
  2018 International Conference on Management of Data, {SIGMOD} Conference
  2018, Houston, TX, USA, June 10-15, 2018}}. \bibinfo{pages}{1493--1508}.
\newblock


\bibitem[\protect\citeauthoryear{Mingdong, Lixin, Derong, Yue, and
  Tiezheng}{Mingdong et~al\mbox{.}}{2018}]%
        {mingdong2018methods}
\bibfield{author}{\bibinfo{person}{ZHU Mingdong}, \bibinfo{person}{XU Lixin},
  \bibinfo{person}{SHEN Derong}, \bibinfo{person}{KOU Yue}, {and}
  \bibinfo{person}{NIE Tiezheng}.} \bibinfo{year}{2018}\natexlab{}.
\newblock \showarticletitle{Methods for Similarity Query on Uncertain Data with
  Cosine Similarity Constraints}.
\newblock \bibinfo{journal}{\emph{Journal of Frontiers of Computer Science \&
  Technology}} \bibinfo{volume}{12}, \bibinfo{number}{1}
  (\bibinfo{year}{2018}), \bibinfo{pages}{49}.
\newblock


\bibitem[\protect\citeauthoryear{Moore}{Moore}{1984}]%
        {moore1984possible}
\bibfield{author}{\bibinfo{person}{Robert~C Moore}.}
  \bibinfo{year}{1984}\natexlab{}.
\newblock \bibinfo{booktitle}{\emph{Possible-world semantics for autoepistemic
  logic}}.
\newblock \bibinfo{type}{{T}echnical {R}eport}. \bibinfo{institution}{SRI
  INTERNATIONAL MENLO PARK CA ARTIFICIAL INTELLIGENCE CENTER}.
\newblock


\bibitem[\protect\citeauthoryear{Nair, Sankaran, and Balakrishnan}{Nair
  et~al\mbox{.}}{2013}]%
        {Nair2013}
\bibfield{author}{\bibinfo{person}{N.~Unnikrishnan Nair},
  \bibinfo{person}{P.~G. Sankaran}, {and} \bibinfo{person}{N. Balakrishnan}.}
  \bibinfo{year}{2013}\natexlab{}.
\newblock \bibinfo{booktitle}{\emph{Stochastic Orders in Reliability}}.
\newblock \bibinfo{publisher}{Springer New York}, \bibinfo{address}{New York,
  NY}, \bibinfo{pages}{281--326}.
\newblock
\showISBNx{978-0-8176-8361-0}
\urldef\tempurl%
\url{https://doi.org/10.1007/978-0-8176-8361-0_8}
\showDOI{\tempurl}


\bibitem[\protect\citeauthoryear{Pollard, Johnson, Raffa, Celi, Mark, and
  Badawi}{Pollard et~al\mbox{.}}{2018}]%
        {eICU2018}
\bibfield{author}{\bibinfo{person}{Tom~J. Pollard}, \bibinfo{person}{Alistair
  E.~W. Johnson}, \bibinfo{person}{Jesse~D. Raffa}, \bibinfo{person}{Leo~A.
  Celi}, \bibinfo{person}{Roger~G. Mark}, {and} \bibinfo{person}{Omar Badawi}.}
  \bibinfo{year}{2018}\natexlab{}.
\newblock \showarticletitle{The eICU Collaborative Research Database, a freely
  available multi-center database for critical care research}.
\newblock \bibinfo{journal}{\emph{Scientific Data}} \bibinfo{volume}{5},
  \bibinfo{number}{1} (\bibinfo{year}{2018}), \bibinfo{pages}{180178}.
\newblock


\bibitem[\protect\citeauthoryear{Redmon, Divvala, Girshick, and Farhadi}{Redmon
  et~al\mbox{.}}{2016}]%
        {redmon2016look}
\bibfield{author}{\bibinfo{person}{Joseph Redmon}, \bibinfo{person}{Santosh
  Divvala}, \bibinfo{person}{Ross Girshick}, {and} \bibinfo{person}{Ali
  Farhadi}.} \bibinfo{year}{2016}\natexlab{}.
\newblock \bibinfo{title}{You Only Look Once: Unified, Real-Time Object
  Detection}.
\newblock
\newblock
\showeprint[arxiv]{1506.02640}~[cs.CV]


\bibitem[\protect\citeauthoryear{Redmon and Farhadi}{Redmon and
  Farhadi}{2016}]%
        {redmon2016yolo9000}
\bibfield{author}{\bibinfo{person}{Joseph Redmon} {and} \bibinfo{person}{Ali
  Farhadi}.} \bibinfo{year}{2016}\natexlab{}.
\newblock \bibinfo{title}{YOLO9000: Better, Faster, Stronger}.
\newblock
\newblock
\showeprint[arxiv]{1612.08242}~[cs.CV]


\bibitem[\protect\citeauthoryear{Rodrigues, P{\'{e}}pin, Goeuriot, and
  Amer{-}Yahia}{Rodrigues et~al\mbox{.}}{2020}]%
        {DBLP:conf/cikm/RodriguesPGA20}
\bibfield{author}{\bibinfo{person}{Jose~F. Rodrigues},
  \bibinfo{person}{Jean~Louis P{\'{e}}pin}, \bibinfo{person}{Lorraine
  Goeuriot}, {and} \bibinfo{person}{Sihem Amer{-}Yahia}.}
  \bibinfo{year}{2020}\natexlab{}.
\newblock \showarticletitle{An Extensive Investigation of Machine Learning
  Techniques for Sleep Apnea Screening}. In \bibinfo{booktitle}{\emph{{CIKM}
  '20: The 29th {ACM} International Conference on Information and Knowledge
  Management, Virtual Event, Ireland, October 19-23, 2020}}.
  \bibinfo{pages}{2709--2716}.
\newblock


\bibitem[\protect\citeauthoryear{Rodrigues-Jr, Gutierrez, Spadon, Brandoli, and
  Amer-Yahia}{Rodrigues-Jr et~al\mbox{.}}{2021}]%
        {jose2021lig}
\bibfield{author}{\bibinfo{person}{Jose~F. Rodrigues-Jr},
  \bibinfo{person}{Marco~A. Gutierrez}, \bibinfo{person}{Gabriel Spadon},
  \bibinfo{person}{Bruno Brandoli}, {and} \bibinfo{person}{Sihem Amer-Yahia}.}
  \bibinfo{year}{2021}\natexlab{}.
\newblock \showarticletitle{LIG-Doctor: Efficient patient trajectory prediction
  using bidirectional minimal gated-recurrent networks}.
\newblock \bibinfo{journal}{\emph{Information Sciences}}  \bibinfo{volume}{545}
  (\bibinfo{year}{2021}), \bibinfo{pages}{813--827}.
\newblock
\showISSN{0020-0255}
\urldef\tempurl%
\url{https://doi.org/10.1016/j.ins.2020.09.024}
\showDOI{\tempurl}


\bibitem[\protect\citeauthoryear{Tai-Seale, McGuire, and Zhang}{Tai-Seale
  et~al\mbox{.}}{2007}]%
        {tai2007time}
\bibfield{author}{\bibinfo{person}{Ming Tai-Seale}, \bibinfo{person}{Thomas~G
  McGuire}, {and} \bibinfo{person}{Weimin Zhang}.}
  \bibinfo{year}{2007}\natexlab{}.
\newblock \showarticletitle{Time allocation in primary care office visits.}
\newblock \bibinfo{journal}{\emph{Health Serv Res}} \bibinfo{volume}{42},
  \bibinfo{number}{5} (\bibinfo{date}{Oct} \bibinfo{year}{2007}),
  \bibinfo{pages}{1871--1894}.
\newblock


\bibitem[\protect\citeauthoryear{Theobald, Weikum, and Schenkel}{Theobald
  et~al\mbox{.}}{2004}]%
        {DBLP:conf/vldb/TheobaldWS04}
\bibfield{author}{\bibinfo{person}{Martin Theobald}, \bibinfo{person}{Gerhard
  Weikum}, {and} \bibinfo{person}{Ralf Schenkel}.}
  \bibinfo{year}{2004}\natexlab{}.
\newblock \showarticletitle{Top-k Query Evaluation with Probabilistic
  Guarantees}. In \bibinfo{booktitle}{\emph{(e)Proceedings of the Thirtieth
  International Conference on Very Large Data Bases, {VLDB} 2004, Toronto,
  Canada, August 31 - September 3 2004}},
  \bibfield{editor}{\bibinfo{person}{Mario~A. Nascimento},
  \bibinfo{person}{M.~Tamer {\"{O}}zsu}, \bibinfo{person}{Donald Kossmann},
  \bibinfo{person}{Ren{\'{e}}e~J. Miller}, \bibinfo{person}{Jos{\'{e}}~A.
  Blakeley}, {and} \bibinfo{person}{K.~Bernhard Schiefer}} (Eds.).
  \bibinfo{publisher}{Morgan Kaufmann}, \bibinfo{pages}{648--659}.
\newblock


\bibitem[\protect\citeauthoryear{Vershynin}{Vershynin}{2018}]%
        {vershynin_2018}
\bibfield{author}{\bibinfo{person}{Roman Vershynin}.}
  \bibinfo{year}{2018}\natexlab{}.
\newblock \bibinfo{booktitle}{\emph{High-Dimensional Probability: An
  Introduction with Applications in Data Science}}.
\newblock \bibinfo{publisher}{Cambridge University Press}.
\newblock
\urldef\tempurl%
\url{https://doi.org/10.1017/9781108231596}
\showDOI{\tempurl}


\bibitem[\protect\citeauthoryear{Yang, Wang, Huang, Lu, Li, and Wang}{Yang
  et~al\mbox{.}}{2022}]%
        {10.14778/3547305.3547310}
\bibfield{author}{\bibinfo{person}{Zhihui Yang}, \bibinfo{person}{Zuozhi Wang},
  \bibinfo{person}{Yicong Huang}, \bibinfo{person}{Yao Lu},
  \bibinfo{person}{Chen Li}, {and} \bibinfo{person}{X.~Sean Wang}.}
  \bibinfo{year}{2022}\natexlab{}.
\newblock \showarticletitle{Optimizing Machine Learning Inference Queries with
  Correlative Proxy Models}.
\newblock \bibinfo{journal}{\emph{Proc. VLDB Endow.}} \bibinfo{volume}{15},
  \bibinfo{number}{10} (\bibinfo{date}{jun} \bibinfo{year}{2022}),
  \bibinfo{pages}{2032–2044}.
\newblock
\showISSN{2150-8097}
\urldef\tempurl%
\url{https://doi.org/10.14778/3547305.3547310}
\showDOI{\tempurl}


\bibitem[\protect\citeauthoryear{Zheng, Yang, and Tian}{Zheng
  et~al\mbox{.}}{2018}]%
        {zheng2018SIFT}
\bibfield{author}{\bibinfo{person}{Liang Zheng}, \bibinfo{person}{Yi Yang},
  {and} \bibinfo{person}{Qi Tian}.} \bibinfo{year}{2018}\natexlab{}.
\newblock \showarticletitle{SIFT Meets CNN: A Decade Survey of Instance
  Retrieval}.
\newblock \bibinfo{journal}{\emph{IEEE Transactions on Pattern Analysis and
  Machine Intelligence}} \bibinfo{volume}{40}, \bibinfo{number}{5}
  (\bibinfo{year}{2018}), \bibinfo{pages}{1224--1244}.
\newblock
\urldef\tempurl%
\url{https://doi.org/10.1109/TPAMI.2017.2709749}
\showDOI{\tempurl}


\bibitem[\protect\citeauthoryear{Zhou, Chai, Li, and SUN}{Zhou
  et~al\mbox{.}}{2020}]%
        {9094012}
\bibfield{author}{\bibinfo{person}{Xuanhe Zhou}, \bibinfo{person}{Chengliang
  Chai}, \bibinfo{person}{Guoliang Li}, {and} \bibinfo{person}{JI SUN}.}
  \bibinfo{year}{2020}\natexlab{}.
\newblock \showarticletitle{Database Meets Artificial Intelligence: A Survey}.
\newblock \bibinfo{journal}{\emph{IEEE Transactions on Knowledge and Data
  Engineering}} \bibinfo{volume}{1}, \bibinfo{number}{1}
  (\bibinfo{year}{2020}), \bibinfo{pages}{1--18}.
\newblock


\end{thebibliography}

\appendix
\section{Appendix}
\label{sec:appendix}

\subsection{Proof of Lemma \ref{lemma:monoto_replace}}

In order to prove Lemma \ref{lemma:monoto_replace}, we need to introduce the notion of the \textit{usual stochastic order} \cite{Nair2013}
, $\leq_{st}$.
\begin{definition} [Usual Stochastic Order]
\label{def:usual_stocha_order}
Let $X$ and $Y$ be two random variables such that
    $Pr[X \geq x] \leq Pr[Y \geq x], \, \forall x\in (-\infty, \infty)$.
Then $X$ is said to be smaller than $Y$ in the usual stochastic order (denoted by $X \leq_{st} Y$).
\end{definition}
One important property for usual stochastic order is as follows,
\begin{proposition} 
\label{prop:st_expectation}
Let $X$ and $Y$ be two random variables. If $X \leq_{st} Y$, then 
    $\mathbb{E}[\psi(X)] \leq \mathbb{E}[\psi(Y)]$
for all \textit{increasing function} $\psi$ for which the expectations exist.
\end{proposition}

The proof of Proposition \ref{prop:st_expectation} relies on constructing upper sets on the domain of $X$ and $Y$, which is beyond the scope of this paper. We refer interested readers to the literature \cite{Nair2013} for more details.

Now, we can prove Lemma \ref{lemma:monoto_replace}.
\begin{proof}
Given $\gamma$, we first show $\pos(S, M, \gamma) \leq \pos(S', M, \gamma)$ and then $\mathbb{E}[\measurecomp(S)] \leq \mathbb{E}[\measurecomp(S')]$.

Recall $\pos(S, M, \gamma) = Pr[\measure(S) \geq \gamma]$. A sufficient condition for $\pos(S, M, \gamma) \leq \pos(S', M, \gamma)$ is $\precis(S) \leq_{st} \precis(S')$ \textbf{and} $\recall(S) \leq_{st} \recall(S')$.
We first discuss $\measure=\precis$, and then $\measure=\recall$. We abbreviate $\phi(x_i)$, $\phi(x_j)$ as $\phi_i$, $\phi_j$ for brevity.

When $\measure = \precis$, define random variables $X=N_{S \setminus \{x_i\}}$, $Y=N_S$, and $Z=N_{S'}$. By equation \ref{eq:precis_prob}, we have $Pr[\precis(S)\geq \gamma]=Pr[Y \geq \lceil |S|\gamma \rceil]$. The following relation holds,
\begin{small}
\begin{equation}
\label{eq:prob_expand_trick}
\begin{split}
Pr[Y \geq \lceil |S|\gamma \rceil]
    &= Pr[X \geq \lceil |S|\gamma \rceil](1-\phi_i) + Pr[X \geq \lceil |S|\gamma \rceil-1]\phi_i\\
    &= \phi_i \cdot Pr[X=\lceil |S|\gamma \rceil-1] + Pr[X \geq \lceil |S|\gamma \rceil]
\end{split}
\end{equation}
\end{small}
where the last step is due to $Pr[X \geq \lceil |S|\gamma \rceil-1] - Pr[X \geq \lceil |S|\gamma \rceil] = Pr[X=\lceil |S|\gamma \rceil-1]$. Similarly, we have
\begin{small}
\begin{equation}
    Pr[Z \geq \lceil |S|\gamma \rceil] = \phi_j \cdot Pr[X=\lceil |S|\gamma \rceil-1] + Pr[X \geq \lceil |S|\gamma \rceil]
\end{equation}
\end{small}
Since $\phi_i \leq \phi_j$, we have $Pr[Y \geq \lceil s\gamma \rceil] \leq Pr[Z \geq \lceil s\gamma \rceil]$ for $\gamma \in \mathbb{R}$, and therefore $Pr[\precis(S)\geq \gamma] \leq Pr[\precis(S')\geq \gamma]$ for $\gamma \in \mathbb{R}$. By definition, we can conclude
$Y \leq_{st} Z$ and $\precis(S) \leq_{st} \precis(S')$. 

Next, we show $\recall(S) \leq_{st} \recall(S')$. 
When $\gamma = 0$, we have $Pr[\recall(S)\geq 0] = 1 \leq Pr[\recall(S')\geq 0] = 1$.
When $\gamma \in \mathbb{R} \setminus \{0\}$, denote random variables $X_C=N_{(D \setminus S) \cup \{x_i\}}$, $Y_C=N_{D \setminus S}$, and $Z_C=N_{D \setminus S'}$. By equation \ref{eq:recall_prob}, we have,
\begin{small}
\begin{equation}
\begin{split}
    Pr[\recall(S) \geq \gamma] &= 
    \sum_{k=0}^{|S|} Pr[Y=k]\cdot Pr[Y_C \leq \lfloor \frac{k(1-\gamma)}{\gamma} \rfloor],\\
    Pr[\recall(S') \geq \gamma] &= 
    \sum_{k=0}^{|S|} Pr[Z=k]\cdot Pr[Z_C \leq \lfloor \frac{k(1-\gamma)}{\gamma} \rfloor].\\
\end{split}
\end{equation}
\end{small}
Similar to Eq. \ref{eq:prob_expand_trick}, we have
\begin{small}
\begin{equation}
\begin{split}
Pr[Y_C \leq \lfloor \frac{k(1-\gamma)}{\gamma} \rfloor]
    &= \phi_i \cdot Pr[X_C = \lfloor \frac{k(1-\gamma)}{\gamma} \rfloor+1] + Pr[X_C \leq \lfloor \frac{k(1-\gamma)}{\gamma} \rfloor]\\
Pr[Z_C \leq \lfloor \frac{k(1-\gamma)}{\gamma} \rfloor]
    &= \phi_j \cdot Pr[X_C = \lfloor \frac{k(1-\gamma)}{\gamma} \rfloor+1] + Pr[X_C \leq \lfloor \frac{k(1-\gamma)}{\gamma} \rfloor]\\
\end{split}
\end{equation}
\end{small}
Since $\phi_i \leq \phi_j$, we conclude $Pr[Y_C \leq \lfloor \frac{k(1-\gamma)}{\gamma} \rfloor]  \leq Pr[Z_C \leq \lfloor \frac{k(1-\gamma)}{\gamma} \rfloor]$ for any $0 \leq k \leq |S|$. 
Denote $\psi(x) := Pr[Z_C \leq \lfloor \frac{x(1-\gamma)}{\gamma} \rfloor]$, we have,
\begin{small}
\begin{equation}
    Pr[\recall(S) \geq \gamma] \leq 
    \sum_{k=0}^{|S|} Pr[Y=k]\cdot \psi(k) =\mathbb{E}[\psi(Y)] .
\end{equation}
\end{small}
Because $Pr[Z_C \leq \lfloor \frac{x(1-\gamma)}{\gamma} \rfloor] = Pr[\frac{\gamma}{1-\gamma} \cdot Z_C \leq x]$, which is the cdf for the random variable $\frac{\gamma}{1-\gamma} Z_C$ evaluated at $x$, we know $\psi(x)$ is an increasing function. By Proposition \ref{prop:st_expectation} and the result $Y \leq_{st} Z$ which we have proven above, we have
    $Pr[\recall(S) \geq \gamma] \leq \mathbb{E}[\psi(Y)] \leq \mathbb{E}[\psi(Z)] = Pr[\recall(S') \geq \gamma]$
for $\gamma \in \mathbb{R} \setminus \{0\}$. 

By definition, we conclude $Pr[\recall(S) \geq \gamma] \leq Pr[\recall(S') \geq \gamma]$ for $\gamma \in \mathbb{R}$ and therefore $\recall(S) \leq_{st} \recall(S')$.
Because $\precis(S) \leq_{st} \precis(S')$ and $\recall(S) \leq_{st} \recall(S')$, we have $\pos(S, M, \gamma) \leq \pos(S', M, \gamma)$ for any given $\gamma$. 

Next, we show $\mathbb{E}[\measurecomp(S)] \leq \mathbb{E}[\measurecomp(S')]$.
Since $\precis(S) \leq_{st} \precis(S')$ and $\recall(S) \leq_{st} \recall(S')$, by Proposition \ref{prop:st_expectation}, we have
$\mathbb{E}[\psi(\precis(S))] \leq \mathbb{E}[\psi(\precis(S'))]$ and $\mathbb{E}[\psi(\recall(S))] \leq \mathbb{E}[\psi(\recall(S'))]$
where $\psi(x)$ is an increasing function.
Let $\psi(x) := x$, we can conclude $\mathbb{E}[\measurecomp(S)] \leq \mathbb{E}[\measurecomp(S')]$ for both PT and RT queries.
\end{proof}

\subsection{Proof of Theorem \ref{theorem:topk}}

\begin{proof}

%
When $k=0$, the case is trivial.
When $1\leq k \leq |D|$,
consider $S \subseteq D$ of $|S|=k$. For $1\leq i \leq k$, let $x_i$ and $x'_i$ denote the $i$-th object of the smallest proxy distance from $\topkproxy{k}$ and $S$, separately. Since $\topkproxy{k}$ is the collection of $k$ nearest proxy neighbors, we have $dist^P(x_i) \leq dist^P(x'_i)$ and, therefore, $\phi(x_i) \geq \phi(x'_i)$. By replacing each $x'_i$ by $x_i$ for $1 \leq i \leq k$, we construct $\topkproxy{k}$ from $S$. After each replacement operation, the success probability and expected CR monotonically increase according to Lemma \ref{lemma:monoto_replace}. As a result, we have $\pos(S, M, \gamma) \leq \pos(\topkproxy{k}, M, \gamma)$ and $\mathbb{E}[\measurecomp(S)] \leq \mathbb{E}[\measurecomp(\topkproxy{k})]$ for any $S \subseteq D$ of $|S|=k$.

\end{proof}

\subsection{Proof of Lemma \ref{lemma:monoto_append}}

\begin{proof}





Given $\gamma$, we first prove $\pos(\topphi{k}, M_r, \gamma) \leq \pos(\topphi{k+1}, M_r, \gamma)$ by showing $\recall(\topkproxy{k}) \leq_{st} \recall(\topkproxy{k+1})$, then
$\mathbb{E}[\recall(\topphi{k})] \leq \mathbb{E}[\recall(\topphi{k+1})]$. The proof is similar to the proof of Lemma \ref{lemma:monoto_replace}, and we only present critical steps for brevity. 

We first show $\recall(\topkproxy{k}) \leq_{st} \recall(\topkproxy{k+1})$.
When $\gamma = 0$, we have $Pr[\recall(\topkproxy{k})\geq 0] = 1 \leq Pr[\recall(\topkproxy{k+1})\geq 0] = 1$.
When $\gamma \in \mathbb{R} \setminus \{0\}$, denote $X = N_{\topkproxy{k}}$, $X_C = N_{D\setminus \topkproxy{k}}$, $Y = N_{\topkproxy{k+1}}$, $Y_C = N_{D\setminus \topkproxy{k+1}}$.
We have,
\begin{small}
\begin{equation}
\begin{split}
    Pr[\recall(\topkproxy{k}) \geq \gamma] &= 
    \sum_{j=0}^{k} Pr[X=j]\cdot Pr[X_C \leq \lfloor \frac{j(1-\gamma)}{\gamma} \rfloor],\\
    Pr[\recall(\topkproxy{k+1}) \geq \gamma] &=
    \sum_{j=0}^{k+1} Pr[Y=j]\cdot Pr[Y_C \leq \lfloor \frac{j(1-\gamma)}{\gamma} \rfloor].\\
\end{split}
\end{equation}
\end{small}
For $x' \in \topkproxy{k+1} \setminus \topkproxy{k}$, similar to Eq. \ref{eq:prob_expand_trick}, we have, 
\begin{small}
\begin{equation}
\begin{split}
Pr[Y_C \leq \lfloor \frac{j(1-\gamma)}{\gamma} \rfloor]
    &= \phi(x') \cdot Pr[X_C = \lfloor \frac{j(1-\gamma)}{\gamma} \rfloor+1] \\
    & \hspace{3mm} +Pr[X_C \leq \lfloor \frac{j(1-\gamma)}{\gamma} \rfloor] \\
    &\geq Pr[X_C \leq \lfloor \frac{j(1-\gamma)}{\gamma} \rfloor]\\
\end{split}
\end{equation}
\end{small}
Denote $\psi(x) := Pr[Y_C \leq \lfloor \frac{x(1-\gamma)}{\gamma} \rfloor]$, which is an increasing function, we have,
\begin{small}
\begin{equation}
    Pr[\recall(\topkproxy{k}) \geq \gamma] \leq 
    \sum_{j=0}^{k} Pr[X=j]\cdot \psi(j) =\mathbb{E}[\psi(X)] .
\end{equation}
\end{small}
It is easy to examine that $X \leq_{st} Y$.
By Proposition \ref{prop:st_expectation}, we have
$
    Pr[\recall(\topkproxy{k}) \geq \gamma] \leq \mathbb{E}[\psi(X)] \leq \mathbb{E}[\psi(Y)] = Pr[\recall(\topkproxy{k+1}) \geq \gamma]
$
for $\gamma \in \mathbb{R} \setminus \{0\}$. 
By definition, we conclude $Pr[\recall(\topkproxy{k}) \geq \gamma] \leq Pr[\recall(\topkproxy{k+1}) \geq \gamma]$ for $\gamma \in \mathbb{R}$ and therefore
$\recall(\topkproxy{k}) \leq_{st} \recall(\topkproxy{k+1})$. 

Denote $\psi(x):=x$. By Proposition \ref{prop:st_expectation} and the result $\recall(\topkproxy{k}) \leq_{st} \recall(\topkproxy{k+1})$, we have $\mathbb{E}[\recall(\topkproxy{k})] \leq \mathbb{E}[\recall(\topkproxy{k+1})]$, same as the proof of Lemma \ref{lemma:monoto_replace}.

\end{proof}

\subsection{Proof of Eq. \ref{eq:csa_s1} \& \ref{eq:csa_m1}}

\begin{proof}
We first give a lower bound for $\eou(s^*, m^*)$, upon which we develop Eq. \ref{eq:csa_s1} and \ref{eq:csa_m1} accordingly.

Recall $\mlow(s) = \lceil \frac{log(\delta)}{log(\prod_{i=0}^{c-1} \frac{|D|-s-i}{|D|-i})} \rceil$ and $\eou(s^*, m^*) = \eou(s^*, \mlow(s^*))$, for any given $c$ and $\delta$.  When $s$ is a constant, $\eou(s,m)$ monotonically increases as $m$ increases. 
Denote $\mlowlow(s) := \frac{log(\delta)}{log(\prod_{i=0}^{c-1} \frac{|D|-s-i}{|D|-i})} \leq \mlow(s)$ for $1\leq s \leq |D|-c$.
Clearly, $\eou(s^*, \mlow(s^*)) \geq \eou(s^*, \mlowlow(s^*))$. $\eou(s, \mlowlow(s))$ is a monotonically decreasing function of $s$ \footnote{This can be seen by showing gradients of $\eou(s, \mlowlow(s))$ w.r.t. $s$ are constantly less or equal to zero for $1 \leq s \leq |D|-c$. }, whose minimal value is taken on $s = |D|-c$. 
We conclude $\eou(s^*, m^*) \geq \eou(|D|-c, \mlowlow(|D|-c))$ and $\xi(s,m) \geq \frac{|D|-\eou(s,m)}{|D|-\eou(|D|-c, \mlowlow(|D|-c))}$ for any $s$, $m$ settings.

Next, we prove Eq. \ref{eq:csa_s1}. When $s=1$ and $m=\mlow(1)$, we have $\xi(1,\mlow(1)) \geq \frac{|D|-\eou(1,\mlow(1))}{|D|-\eou(|D|-c, \mlowlow(|D|-c))}\geq \frac{|D|-\eou(1,\mlowlow(1)+1)}{|D|-\eou(|D|-c, \mlowlow(|D|-c))}$ due to $\mlow(1) \leq \mlowlow(1)+1$. By taking logarithm on both sides and cancelling redundant terms, we have,

\eat{
Now, we can give the approximation rate if we use $s=1$,
\begin{equation}
\begin{split}
\gamma = \frac{|D|-c(1,m(1))}{|D|-OPT} &\geq \frac{|D|-c(1,m(1))}{|D|-c(D-K,\underline{m}(D-K))} \\
&= \frac{(1-\frac{1}{|D|})^{\lceil \frac{log(1-prob)}{log(\prod_{i=0}^{K-1} \frac{|D|-1-i}{|D|-i})} \rceil}}{ (\frac{K}{|D|})^{ \frac{log(1-prob)}{log(\prod_{i=0}^{K-1} \frac{K-i}{|D|-i})} }} \\
&\geq \frac{(1-\frac{1}{|D|})^{ \frac{log(1-prob)}{log(\frac{|D|-K}{|D|})} + 1}}{(\frac{K}{|D|})^{ \frac{log(1-prob)}{\sum_{i=0}^{K-1}log(\frac{K-i}{|D|-i})} }} \\
\end{split}
\end{equation}

The expression is quite complicated, and we would like to simplify it. We first take logarithm on both sides,}
\begin{small}
\begin{equation}
\label{eq:lb_csa_s1}
\begin{split}
log(\xi(1,\mlow(1))) &\geq log(1-\frac{1}{|D|}) -log(\delta) \Big( \frac{log(1-\frac{1}{|D|})}{log(\frac{|D|}{|D|-c})} + \frac{log(\frac{|D|}{c})}{\sum_{i=0}^{c-1}log(\frac{|D|-i}{c-i})} \Big)\\
\end{split}
\end{equation}
\end{small}
Denote $g(c) := \frac{log(1-\frac{1}{|D|})}{log(\frac{|D|}{|D|-c})}$ and $h(c) := \frac{log(\frac{|D|}{c})}{\sum_{i=0}^{c-1}log(\frac{|D|-i}{c-i})}$. 
Because $1-\frac{1}{x} \leq log(x) \leq x-1$ for $x>0$, we have
\begin{small}
\begin{equation}
\begin{split}
    g(c) &\geq \frac{log(1-\frac{1}{|D|})}{1 - \frac{|D|-c}{|D|}} = \frac{log(1-\frac{1}{|D|})|D|}{c} 
    \geq -\frac{|D|}{c(|D|-1)}\\
 h(c) &\geq \frac{(|D|-c)/|D|}{\sum_{i=0}^{c-1}\frac{|D|-i}{c-i}-1}
 \geq \frac{(|D|-c)/|D|}{(|D|-c)c}
 = \frac{1}{c|D|}
\end{split}
\end{equation}
\end{small}
Therefore,
\begin{small}
\begin{equation}
    log(\xi(1,\mlow(1))) \geq log(1-\frac{1}{|D|}) -log(\delta)\frac{1}{c} ( \frac{1}{|D|} - \frac{|D|}{|D|-1} )
\end{equation}
\end{small}
which equals to
$\xi(1, \mlow(1)) \geq \delta^{\frac{-1}{c}( \frac{1}{|D|} - \frac{|D|}{|D|-1})} \cdot (1-\frac{1}{|D|})$
, as known as Eq. \ref{eq:csa_s1}.

Next, we prove Eq. \ref{eq:csa_m1}.
For $m=1$, we first show $s=s_1:=\lceil \frac{-log(\delta)}{\sum_{i=0}^{c-1} \frac{1}{|D|-i}} \rceil$ ensures high success probability. When $m=1$, the failure rate is $1-\poi(|D|,s,1,c) = \prod_{i=0}^{c-1} \frac{|D|-s-i}{|D|-i}$. By taking logarithm, we have $\sum_{i=0}^{c-1} log(\frac{|D|-s-i}{|D|-i}) \leq s\cdot \sum_{i=0}^{c-1} \frac{-1}{|D|-i}$. Given $\delta$, we require $s\cdot \sum_{i=0}^{c-1} \frac{-1}{|D|-i} \leq log(\delta)$ to ensure the success probability being no less than $1-\delta$, which equals to requiring $s \geq s_1$.

When $m=1$ and $s=s_1$, we have $\xi(s_1,1) \geq \frac{|D|-\eou(s_1,1)}{|D|-\eou(|D|-c, \mlowlow(|D|-c))}$. By plugging the expression of $\eou(s_1,1)$ and relaxing $\sum_{i=0}^{c-1} \frac{1}{|D|-i} \geq \frac{c}{|D|}$, we have $|D|-\eou(s_1,1) \geq |D|-1+|D|\frac{log(\delta)}{c}$. Similarly, by taking logarithm on both sides, we have,
\begin{small}
\begin{equation}
\begin{split}
    log(\xi(s_1,1)) &\geq log(1-\frac{1}{|D|}+\frac{log(\delta)}{c}) - log(\delta)h(c)\\
    &\geq log(1-\frac{1}{|D|}+\frac{log(\delta)}{c}) -\log(\delta) \frac{1}{c|D|}\\
\end{split}
\end{equation}
\end{small}
which equals to 
$\xi(s_1, 1)  \geq  \delta ^{\frac{-1}{|D|c}} \cdot (1 - \frac{1}{|D|} + \frac{log(\delta)}{c})
$, as known as Eq. \ref{eq:csa_m1}.

\end{proof}

\end{document}